\documentclass[11pt,a4paper]{article}
\usepackage{amssymb}
\usepackage{amsmath}
\usepackage{amsfonts}
\usepackage{bbm}
\usepackage{amsthm}
\usepackage{mathrsfs}
\usepackage{hyperref}
\usepackage{color}
\usepackage[margin=2.41cm]{geometry}
\usepackage[all,cmtip]{xy}
\usepackage[utf8]{inputenc}
\usepackage{graphicx}
\usepackage{varwidth}
\usepackage{comment}
\usepackage{enumitem}

\usepackage{mathtools}

\usepackage{upgreek}
\usepackage{rotating}

\usepackage{tikz}
\usetikzlibrary{shapes.geometric}
\usetikzlibrary{patterns.meta}
\usetikzlibrary{cd}
\usetikzlibrary{arrows}

\definecolor{darkred}{rgb}{0.8,0.1,0.1}
\hypersetup{
	colorlinks=false,         
	urlcolor=darkred,
	linkcolor=darkred,
	citecolor=blue,
}

\theoremstyle{plain}
\newtheorem{theo}{Theorem}[section]
\newtheorem{lem}[theo]{Lemma}
\newtheorem{propo}[theo]{Proposition}

\theoremstyle{definition}
\newtheorem{defi}[theo]{Definition}

\newtheorem{remm}[theo]{Remark}
\newenvironment{rem}
{\pushQED{\qed}\remm}
{\popQED\endremm}

\newtheorem{constrr}[theo]{Construction}
\newenvironment{constr}
{\pushQED{\qed}\constrr}
{\popQED\endconstrr}

\numberwithin{equation}{section}

\def\nn{\nonumber}

\def\bbA{\mathbb{A}}
\def\bbK{\mathbb{K}}
\def\bbR{\mathbb{R}}

\def\bbN{\mathbb{N}}

\def\bbF{\mathbb{F}}

\def\hom{\underline{\mathrm{hom}}}

\def\id{\mathrm{id}}

\def\cl{\mathrm{cl}}

\def\1{I}
\def\oone{\mathbbm{1}}
\def\op{\mathrm{op}}

\def\Loc{\mathbf{Loc}}

\def\Set{\mathbf{Set}}
\def\Alg{\mathbf{Alg}}

\def\Vec{\mathbf{Vec}}

\def\TT{\mathbf{T}}
\def\QQ{\mathbf{Q}}
\def\Cat{\mathbf{Cat}}

\def\PsOp{\mathbf{PsOp}}

\def\Grpd{\mathbf{Grpd}}

\def\Op{\mathbf{Op}}

\def\AQFT{\mathbf{AQFT}}
\def\FQFT{\mathbf{FQFT}}

\def\RC{\mathbf{RC}}

\def\LB{\mathcal{LB}\mathsf{ord}}
\def\LBop{\mathscr{L}\!\mathscr{B}\mathsf{op}}

\def\AAA{\mathfrak{A}}

\def\BBB{\mathfrak{B}}
\def\GGG{\mathfrak{G}}

\def\FFF{\mathfrak{F}}

\def\O{\mathcal{O}}
\def\P{\mathcal{P}}

\def\scrO{\mathscr{O}}
\def\scrP{\mathscr{P}}

\def\colim{\mathrm{colim}}

\def\add{\mathrm{add}}

\def\sla{{\scriptscriptstyle \slash}}

\newcommand\und[1]{\underline{#1}}

\DeclareMathOperator*{\cmp}{\text{\raisebox{0.25ex}{\scalebox{0.7}{$\odot\,$}}}}

\def\sk{\vspace{2mm}}

\def\hto{\nrightarrow}

\makeatletter
\let\@fnsymbol\@alph
\makeatother

\title{%
An equivalence theorem for algebraic and functorial QFT
}

\author{%
Severin Bunk$^{1,a}$, James MacManus$^{2,b}$\ and\ Alexander Schenkel$^{2,3,c}$\vspace{4mm}\\
{\small ${}^1$ Department of Physics, Astronomy and Mathematics, University of Hertfordshire,}\\ 
{\small College Lane, Hatfield, AL10 9AB, United Kingdom.}\vspace{2mm}\\
{\small ${}^2$ School of Mathematical Sciences, University of Nottingham,}\\
{\small University Park, Nottingham NG7 2RD, United Kingdom.}\vspace{2mm}\\
{\small ${}^3$ Dipartimento di Matematica, Universit{\`a} di Trento and INFN-TIFPA,}\\
{\small Via Sommarive 14, 38123 Povo (Trento), Italy.}\vspace{4mm}\\
{\small \begin{tabular}{ll}
Email: & ${}^a$~\href{mailto:s.bunk@herts.ac.uk}{\texttt{s.bunk@herts.ac.uk}}\\
& ${}^b$~\href{mailto:james.macmanus@nottingham.ac.uk}{\texttt{james.macmanus@nottingham.ac.uk}}\\
& ${}^c$~\href{mailto:alexander.schenkel@unitn.it}{\texttt{alexander.schenkel@unitn.it}}
\vspace{2mm}
\end{tabular}
}
}

\date{July 2026}

\begin{document}

\maketitle

\begin{abstract}
\noindent This paper develops a novel approach to functorial quantum field theories (FQFTs) in the context of Lorentzian geometry. The key challenge is that globally hyperbolic Lorentzian bordisms between two Cauchy surfaces cannot change the topology of the Cauchy surface. This is addressed and solved by introducing a more flexible concept of bordisms which provide morphisms from tuples of causally disjoint partial Cauchy surfaces to a later-in-time full Cauchy surface. They assemble into a globally hyperbolic Lorentzian bordism pseudo-operad, generalizing the geometric bordism pseudo-categories of Stolz and Teichner. The associated FQFTs are defined as pseudo-multifunctors into a symmetric monoidal category of unital associative algebras. The main result of this paper is an equivalence theorem between such globally hyperbolic Lorentzian FQFTs and algebraic quantum field theories (AQFTs), both subject to the time-slice axiom and a mild descent condition called additivity.
\end{abstract}
\vspace{-3mm}

\paragraph*{Keywords:} algebraic quantum field theory, functorial quantum field theory, Lorentzian geometry, bordisms, pseudo-operads
\vspace{-2mm}

\paragraph*{MSC 2020:} 81Txx, 18M60, 18N10, 53C50
\vspace{-2mm}

\renewcommand{\baselinestretch}{0.7}\normalsize
\tableofcontents
\renewcommand{\baselinestretch}{1.0}\normalsize


\section{\label{sec:intro}Introduction and summary}
Quantum field theory (QFT) is a fundamental pillar of 
modern theoretical physics, with extensive applications 
in high-energy physics, statistical physics, and exotic 
states of matter. Since its early days,
there has been a strong need and desire to develop
mathematical foundations for QFT in order to guide and
support the rapid developments in physics by identifying
the key principles underlying QFTs and providing rigorous
tools for their analysis. Over time this has led to 
a variety of mathematical axiomatizations 
of QFT, which differ crucially in their details and depend
on which of the many facets of QFT are considered 
to be most fundamental. Among the most prominent recent
approaches to mathematical QFT are the following: 
1.)~Algebraic QFT (AQFT) \cite{HaagKastler,BFV} 
assigns algebras of local quantum 
observables to suitable open subsets in a Lorentzian spacetime,
emphasizing the underlying Lorentzian geometry
of relativistic QFT and the role of operator algebras 
in quantum theory.
2.)~Functorial QFT (FQFT) \cite{Witten,Atiyah,Segal}
formalizes a concept of `time evolution' for the states or observables
of a QFT along bordisms between codimension-one hypersurfaces in `spacetime'.
Since this approach is used predominantly in 
the context of topological or conformal QFTs, the terms `spacetime'
and `time evolution' have to be interpreted with some care,
as they do not in general represent evolutions along the time dimension
of a Lorentzian manifold.
3.)~The concept of a prefactorization algebra (PFA) \cite{CG1,CG2}
captures the algebraic structures of local quantum observables
which are obtained via Batalin–Vilkovisky quantization. This approach
is very general, as it applies rather universally to 
many geometric scenarios (e.g.\ topological, complex, 
Riemannian and Lorentzian spacetimes), but it is currently only well
understood in the context of perturbative QFT.
\sk

Given the existence of multiple axiomatizations of QFT, it
is important to understand if and how the different
approaches are related to each other. 
This is not only conceptually 
valuable for the cohesiveness of mathematical QFT as a research area, 
but also offers excellent opportunities for a fruitful exchange 
of ideas and techniques across various research communities.
We will now briefly recall some of the existing comparison 
results between different axiomatizations of QFT,
focusing on the case of QFTs defined over (globally hyperbolic) Lorentzian
manifolds, which is also the context of our present paper.
The relationship between PFA and AQFT is by now relatively well understood.
The first result in this direction is due to Gwilliam and Rejzner \cite{GR1},
who observed that the construction of free (i.e.\ non-interacting) QFTs
in both frameworks admits a direct comparison. This was later generalized
to free (higher) gauge theories \cite{BMS} and 
examples of perturbatively interacting QFTs \cite{GR2}. An example-independent
equivalence theorem between suitable categories of PFAs and AQFTs was
proven in \cite{FAvsAQFT} for the case of a $1$-categorical target.
Significant steps towards an $\infty$-categorical generalization
of this result appeared recently in \cite{BCGS}.
However, a full proof at this level is still outstanding.
\sk

The relationship between FQFT and either AQFT or PFA is more subtle and to
the best of our knowledge less understood. One of the apparent difficulties
is that in FQFT one traditionally focuses on the state spaces
of a QFT, while AQFT and PFA emphasize observables, 
so one is tempted to explore their relationship by generalizing
the equivalence between the Schr\"odinger and Heisenberg 
pictures from quantum mechanics to QFT. Some of these aspects
were studied by Schreiber in \cite{Schreiber} and later by 
Johnson-Freyd in \cite{JohnsonFreyd}. See also 
\cite{Scheimbauer} for a construction of FQFTs from PFAs, 
though from a somewhat different perspective. However, in the context of
QFT on Lorentzian spacetimes, one does not expect such an 
equivalence between the Schr\"odinger and Heisenberg pictures
to exist. The reason for this is that the algebras of
quantum observables of a QFT admit multiple inequivalent representations,
i.e.\ the Stone-von Neumann uniqueness theorem from
quantum mechanics fails in QFT, which in general makes it impossible to single out
a distinguished state space. For accessible reviews which address these points see e.g.\ 
\cite{FewsterRejzner} and \cite{FV}. The inadequacy of state spaces
in Lorentzian geometric contexts motivates us to consider FQFTs which assign
algebras of observables to codimension $1$ \textit{spacelike} hypersurfaces
and implement their `time evolution' through bordisms.
This is conceptually compatible with the more traditional 
assignment of observable algebras in codimension $0$ in AQFT since, 
under the usual assumption of a well-posed time evolution (called the time-slice axiom),
data in codimension $0$ can be identified with data in (collar neighborhoods of) 
codimension $1$ \textit{spacelike} hypersurfaces.
A notion of functorial QFTs assigning state spaces can then be introduced in terms
of twisted or relative functorial field theories in the sense of \cite{StolzTeichner,FT},
which we will briefly discuss in our context in Section \ref{sec:twistedFFT}.
We have shown in our previous work \cite{FFTvAQFT} that 
every AQFT has an underlying FQFT which is defined on a suitable
globally hyperbolic Lorentzian bordism category and assigns observable algebras
to Cauchy surfaces. However, this construction is forgetful in the sense 
that the underlying FQFT does not capture the spatially local structure of the AQFT. The
reason for this lies in Lorentzian geometry \cite{BernalSanchez}: 
Every globally hyperbolic Lorentzian manifold is of the form $M\cong\bbR\times \Sigma$,
with $\bbR$ representing a time dimension and $\Sigma$ representing Cauchy surfaces,
which limits the available bordisms between two Cauchy surfaces. In particular,
globally hyperbolic Lorentzian bordisms between two Cauchy surfaces can neither change the topology
nor detect spatially local features in the Cauchy surface $\Sigma$.
\sk

The main aim of this paper is to develop and significantly improve
the concept of globally hyperbolic Lorentzian FQFTs from our previous 
work \cite{FFTvAQFT}, and to prove an equivalence theorem between AQFT and FQFT
in this context. The key idea behind our improvement can be explained straightforwardly
at an informal level, but its precise mathematical implementation is rather technical
and will occupy a large portion of the present paper. In contrast to \cite{FFTvAQFT} 
where we focused on bordisms between two full Cauchy surfaces, here we consider 
globally hyperbolic Lorentzian bordisms $N$ of the form
\begin{flalign}\label{eqn:bordpicture}
\begin{gathered}
\begin{tikzpicture}[scale=0.95]
\draw[thick, ->] (-2.6,-2) -- (-2.6,2) node[above] {{\footnotesize time}};
\draw[dotted,fill=gray!5] (-2.25,0) -- (0,2.25) -- (2.25,0) -- (0,-2.25) -- (-2.25,0);
\draw[thick] (-2.25,0) to [out=30,in=150] (2.25,0);
\draw[thick] (-1.5,-0.35) to [out=-10,in=190] (-1,-0.35);
\draw[thick] (-0.7,-0.35) to [out=-10,in=190] (0,-0.35);
\draw[thick] (0.9,-0.35) to [out=-10,in=190] (1.4,-0.35);
\draw (0.5,-0.15) node[below] {{\small $\cdots$}} ;
\draw (0,0.65) node[above] {{\footnotesize $\Sigma_1$}} ;
\draw (-1.20,-0.4) node[below] {{\footnotesize $\,\Sigma_{0_1}$}} ;
\draw (-0.3,-0.4) node[below] {{\footnotesize $\Sigma_{0_2}$}} ;
\draw (1.15,-0.4) node[below] {{\footnotesize $\Sigma_{0_n}\,$}} ;
\draw (0.75,2.25) node[below] {{\footnotesize $N$}} ;
\end{tikzpicture}
\end{gathered}
\end{flalign}
which go from a tuple $(\Sigma_{0_1},\dots,\Sigma_{0_n})$ of causally disjoint
partial Cauchy surfaces, for any non-negative integer $n\in\bbN_0$, 
to a later-in-time full Cauchy surface $\Sigma_1$.
Such partial $n$-to-$1$ bordisms clearly allow us to resolve spatially local features
since one can choose the individual partial Cauchy surfaces $\Sigma_{0_i}$ to be arbitrarily small.
Furthermore, they endow our globally hyperbolic Lorentzian FQFTs 
with a multiplicative structure which is similar
to the pair-of-pants bordisms in topological QFT and whose physical significance
in our Lorentzian geometric context is to encode a notion of relativistic causality, see
Remark \ref{rem:meaningofnto1}. The mathematical structure
encoding globally hyperbolic Lorentzian $n$-to-$1$ bordisms as in \eqref{eqn:bordpicture}
is not that of a (symmetric monoidal) category as in ordinary FQFT, but rather that of an operad.
More precisely, since geometric bordisms must be endowed with suitable collar regions in order
to allow for a well-defined gluing, see e.g.\ \cite{StolzTeichner} and \cite{FFTvAQFT},
they will form a \textit{pseudo-operad}, which is an operadic generalization
of the concept of pseudo-categories from \cite{PseudoCats}. We will review the basic theory 
of pseudo-operads and develop additional technology we require in that area 
in Appendix \ref{app:pseudooperads}.
\sk

We now explain our results in more detail by outlining the content of this paper.
In Section~\ref{sec:prelim}, we recall some basic aspects of Lorentzian geometry and AQFT.
From the many equivalent definitions of AQFTs available in the literature,
the most convenient one for our purposes is to define an AQFT as a multifunctor
\begin{flalign}
\AAA \,:\, \P_{\Loc_m^\perp} ~\longrightarrow~ \Alg_{\mathsf{uAs}}(\TT)
\end{flalign}
from a suitable operad 
\smash{$\P_{\Loc_m^\perp}$} that is constructed out of the category $\Loc_m^\perp$ of 
$m$-dimensional globally hyperbolic Lorentzian spacetimes (see Definition \ref{def:PLoc}) 
to the symmetric monoidal category $\Alg_{\mathsf{uAs}}(\TT)$ of unital associative algebras 
in a cocomplete closed symmetric monoidal category $\TT$. In Section \ref{sec:LBord},
we define our novel $m$-dimensional \textit{globally hyperbolic Lorentzian bordism 
pseudo-operad} $\LBop_m$ which describes $n$-to-$1$ bordisms as in \eqref{eqn:bordpicture}
(together with suitable collar regions) and their gluing. The precise definition
of this bordism pseudo-operad is technically rather challenging as a consequence 
of the peculiarities of globally hyperbolic Lorentzian geometry,
and it relies in particular on non-elementary geometric results
in this context that we will collect in Appendix \ref{app:Lorentzian},
see in particular Propositions \ref{prop:bordismgluingregions} and \ref{prop:pushout}.
We prove in Proposition \ref{prop:LBopfibrant} that the pseudo-operad $\LBop_m$ 
satisfies an operadic analogue of the fibrancy condition for pseudo-categories
from \cite{Shulman}, which leads to considerable simplifications when studying
its associated FQFTs.
\sk

In Section \ref{sec:FFTs}, we introduce a novel concept of globally hyperbolic
Lorentzian FQFTs which generalizes and considerably 
improves our previous definition in \cite{FFTvAQFT}. These are defined as
pseudo-multifunctors
\begin{flalign}
\FFF \,:\, \LBop_m~\longrightarrow~ \iota\big(\Alg_{\mathsf{uAs}}(\TT)\big)
\end{flalign}
from our globally hyperbolic Lorentzian bordism pseudo-operad $\LBop_m$
to a pseudo-operad $\iota\big(\Alg_{\mathsf{uAs}}(\TT)\big)$ which is constructed
canonically from the symmetric monoidal category of unital associative algebras.
(See Construction \ref{constr:iota} for the details.)
As a consequence of the fibrancy of the pseudo-operad $\LBop_m$
and Theorem \ref{theo:2adjunction}, there exists an equivalent,
but much simpler description of such FQFTs in terms of ordinary
multifunctors $\FFF: \tau(\LBop_m)\to \Alg_{\mathsf{uAs}}(\TT)$
from a certain truncation $\tau(\LBop_m)$ of the bordism 
pseudo-operad to an ordinary operad (see Construction \ref{constr:tau}).
It is important to stress that the truncated operad
$\tau(\LBop_m)$ remembers much of the rich geometric structure of the pseudo-operad $\LBop_m$.
In particular, its objects and operations keep track of the 
collar regions around Cauchy surfaces and bordisms, which are crucial 
for a well-defined operadic composition.
The truncation is analogous to how the usual bordism categories for TQFTs are 
constructed by taking diffeomorphism classes of bordisms in what would 
otherwise be some flavor of $(2,1)$-category (e.g.~the pseudo-categories of~\cite{StolzTeichner}).
\sk

In Section \ref{sec:equivalence}, we prove an equivalence theorem between
AQFT and FQFT which provides a significant improvement and completion of our earlier comparison
result in \cite{FFTvAQFT}, and thereby also justifies our improved axiomatics 
for globally hyperbolic Lorentzian FQFTs developed in this paper.
More precisely, we prove in Theorem \ref{theo:equivalence}
that there exists, for every spacetime dimension $m\in\bbN$, an equivalence
\begin{flalign}
\AQFT_m^{W,\add}\,\simeq\, \FQFT_m^{W,\add}
\end{flalign}
between the category $\AQFT_m^{W,\add}$ of AQFTs satisfying
the time-slice axiom and the additivity property 
(see Definitions \ref{def:AQFT} and \ref{def:AQFTadditivity})
and the category $\FQFT_m^{W,\add}$ of FQFTs satisfying the time-slice axiom
and the additivity property (see Definitions \ref{def:FFT} and \ref{def:FFTadditivity}).
The proof of this theorem relies on a combination of the
pseudo-operadic techniques we develop and their interaction with 
non-elementary technical results about globally hyperbolic Lorentzian manifolds.
The additional hypotheses of the time-slice axiom, encoding a well-defined
time evolution, and the additivity property, encoding a mild descent condition,
both seem to be crucial to establish our equivalence theorem.
We would like to emphasize that precisely the same hypotheses are also used in the 
equivalence theorem between AQFT and PFA from \cite{FAvsAQFT}, which indicates their 
relevance to relate AQFT with other axiomatizations of QFT.
To prove our equivalence theorem, we explicitly construct functors
between the categories of AQFTs and FQFTs, see Constructions \ref{constr:AQFTtoFQFT} 
and \ref{constr:FQFTtoAQFT}, which are then shown to be quasi-inverse to each other.
\sk

In Section \ref{sec:twistedFFT}, we briefly describe how the concept of twisted or relative FQFTs 
from \cite{StolzTeichner,FT} adapts to our Lorentzian geometric context and thereby provides
a notion of globally hyperbolic Lorentzian FQFTs assigning state spaces.
We will also give physical arguments why it will be difficult, if not impossible, 
to find physically relevant examples of such objects. This provides further justification as to why, in Lorentzian geometric contexts, it is more natural and useful
to consider our concept of algebra-valued FQFTs from Definition \ref{def:FFT}.
\sk

This paper includes two appendices. Appendix \ref{app:pseudooperads} develops the 
basic theory of pseudo-operads which is needed to define the
globally hyperbolic Lorentzian bordism pseudo-operad $\LBop_m$ from Section \ref{sec:LBord}
and its associated concept of FQFTs from Section \ref{sec:FFTs}.
Appendix \ref{app:Lorentzian} proves some non-elementary technical results of a Lorentzian 
geometric nature which are crucial for our developments in the main text.


\section{\label{sec:prelim}Preliminaries on algebraic QFTs}
In this section we recall some basic aspects of Lorentzian geometry
and algebraic quantum field theory (AQFT).
Excellent introductions to Lorentzian geometry can be found, for instance, in 
\cite{ONeill,BGP,Minguzzi}, and we also refer the reader to our 
previous paper \cite[Section 2.1]{FFTvAQFT} for a lightning review
of the most crucial definitions in our context.
Our first definition recalls the typical globally hyperbolic Lorentzian 
spacetime category used in the context of AQFT.
\begin{defi}\label{def:Loc}
For each positive integer $m\in\bbN$, we denote by $\Loc_m$ the 
category whose objects are all oriented and time-oriented globally hyperbolic
Lorentzian manifolds $M$ of dimension $m$ and whose morphisms are all
orientation and time-orientation preserving isometric open embeddings
$f:M\to N$ with causally convex image $f(M)\subseteq N$. The following
distinguished (tuples of) $\Loc_m$-morphisms will play a prominent role:
\begin{itemize}
\item[(a)] A $\Loc_m$-morphism $f:M\to N$ is called \textit{Cauchy}
if its image $f(M)\subseteq N$ contains a Cauchy surface of $N$.

\item[(b)] A cospan of $\Loc_m$-morphisms $f_1 : M_1\to N \leftarrow M_2 :f_2$ 
is called \textit{causally disjoint}, written as $f_1\perp f_2$, if there
exists no causal curve connecting the images $f_1(M_1)\subseteq N$ and $f_2(M_2)\subseteq N$.
\end{itemize}
\end{defi}

The concept of an AQFT over $\Loc_m$, which is sometimes also called an $m$-dimensional locally covariant AQFT,
has been proposed in the seminal work \cite{BFV} of Brunetti, Fredenhagen and Verch. 
See also \cite{FV} for an informative review.
More recently, it has been observed in \cite{AQFToperad} that such AQFTs admit an elegant and
powerful description in terms of algebras over a suitable operad which is constructed
from the category $\Loc_m$ and its causal disjointness relation $\perp$ from Definition \ref{def:Loc}.
For our purposes, the most convenient description of AQFTs is the equivalent 
one established in \cite[Theorem 2.9]{2AQFT}, whose focus is on the following
Lorentzian geometric prefactorization operad associated with Definition \ref{def:Loc}.
\begin{defi}\label{def:PLoc}
The \textit{prefactorization operad} $\P_{\Loc_m^\perp}$ 
is the (colored symmetric) operad which is defined by the following data:
\begin{itemize}
\item[(i)] The objects of the operad $\P_{\Loc_m^\perp}$ are the objects of the category $\Loc_m$.

\item[(ii)] For each non-negative integer $n\in\bbN_0$, the 
set of operations from an $n$-tuple of objects
$\und{M}=(M_1,\dots,M_n)\in\Loc_m^{\times n}$ to a single object
$N\in\Loc_m$ is given by
\begin{flalign}
\P_{\Loc_m^\perp}\big(\substack{N\\ \und{M}}\big)\,\coloneqq\,\bigg\{\und{f} = (f_1,\dots,f_n)\in \prod_{i=1}^n\Loc_m(M_i,N)\,:\, f_i\perp f_j\text{ for all }i\neq j\bigg\}\quad.
\end{flalign}
By our conventions, this means that there exists a unique $0$-ary
operation $\varnothing \to N$ from the empty tuple,
for all $N\in\Loc_m$, and for the $1$-ary operations we have that 
$\P_{\Loc_m^\perp}\big(\substack{N\\ M}\big)=\Loc_m(M,N)$ is 
the set of $\Loc_m$-morphisms, for all $M,N\in\Loc_m$.

\item[(iii)] Operadic composition of an $n$-ary operation
$\und{f}=(f_1,\dots,f_n): \und{M}\to N$ and 
an $n$-tuple $\und{\und{g}} = (\und{g}_1,\dots,\und{g}_n): \und{\und{L}}\to \und{M}$
of $k_i$-ary operations $\und{g}_i : \und{L}_i\to M_i$, for $i=1,\dots,n$, is given by 
the following compositions in the category $\Loc_m$
\begin{flalign}\label{eqn:PLocComp}
\und{f}\,\und{\und{g}}\,\coloneqq\,\big(f_1\,g_{11},\dots, f_1\,g_{1k_1},\dots,f_{n}\,g_{n1},\dots,f_n\,g_{nk_n}\big)\,:\,\und{\und{L}}~\longrightarrow~N\quad.
\end{flalign}
The identity operations $\oone_M\coloneqq\id_M\in \P_{\Loc_m^\perp}\big(\substack{M\\ M}\big)$
are given by the identity morphisms of the category $\Loc_m$.

\item[(iv)] For $\sigma \in S_n$ an element of the permutation group, we write $\und{M}\sigma = (M_{\sigma(1)}, \ldots, M_{\sigma(n)})$.
The permutation actions \smash{$\P_{\Loc_m^\perp}(\sigma):
\P_{\Loc_m^\perp}\big(\substack{N\\ \und{M}}\big)\to \P_{\Loc_m^\perp}\big(\substack{N\\ \und{M}\sigma}\big)$} are given by
\begin{flalign}
\P_{\Loc_m^\perp}(\sigma)(\und{f})\,\coloneqq\,\und{f}\sigma\,=\,
(f_{\sigma(1)},\dots,f_{\sigma(n)})\,:\,\und{M}\sigma~\longrightarrow~N\quad,
\end{flalign}
for all $\und{f}=(f_1,\dots,f_n) : \und{M}\to N$.
\end{itemize}
\end{defi}

By \cite[Theorem 2.9]{2AQFT}, one can describe AQFTs over $\Loc_m$ in terms of algebras
over the operad $\P_{\Loc_m^\perp}$ from Definition \ref{def:PLoc}, however it is important
to emphasize that such algebras must take values in a symmetric monoidal category of unital associative algebras.
To explain our notations, let us fix any cocomplete closed\footnote{Recall that a symmetric monoidal category $\TT$
is closed if the functor $(-)\otimes x : \TT\to \TT$ admits a right adjoint $\hom(x,-) :\TT\to \TT$,
for all $x\in \TT$. From this it follows that the monoidal product $(-)\otimes (-)$
preserves colimits in both entries, which is a property that will be used frequently in our proofs 
in Section \ref{sec:equivalence} below.} symmetric monoidal category $\TT$
and denote by $\Alg_{\mathsf{uAs}}(\TT)$ the category of unital associative algebras in $\TT$.
(In applications, one often considers the closed symmetric monoidal category $\TT=\Vec_\bbK$
of vector spaces over a field $\bbK$, but we do not have to restrict ourselves 
to this case.) 
The category of unital associative algebras $\Alg_{\mathsf{uAs}}(\TT)$ 
is symmetric monoidal with respect to the formation of tensor product algebras
$A\otimes B\in\Alg_{\mathsf{uAs}}(\TT)$ and the monoidal unit $I\in\Alg_{\mathsf{uAs}}(\TT)$
which is given by endowing the monoidal unit $I\in\TT$ with its canonical unital associative algebra structure.
Furthermore, it inherits cocompleteness from the underlying category $\TT$
and the forgetful functor $\Alg_{\mathsf{uAs}}(\TT)\to \TT$ both preserves and reflects
filtered colimits.
\sk

In the following definition, we make use of the standard fact
(see e.g.\ \cite{Mandell}) that every symmetric monoidal category
has an associated (colored symmetric) operad with the same objects, and whose operations are obtained 
by using the symmetric monoidal structure. 
In our case, this means that we can regard $\Alg_{\mathsf{uAs}}(\TT)$
as an operad with sets of operations given by
\begin{flalign}
\Alg_{\mathsf{uAs}}(\TT)\big(\substack{B \\ \und{A}}\big)
\, \coloneqq \,
\Alg_{\mathsf{uAs}}(\TT)\bigg(\bigotimes_{i=1}^n A_i,B\bigg) \quad ,
\end{flalign}
for all $\und{A} = (A_1,\dots,A_n)\in \Alg_{\mathsf{uAs}}(\TT)^{\times n}$ and $B\in \Alg_{\mathsf{uAs}}(\TT)$.
\begin{defi}\label{def:AQFT}
Fix any cocomplete closed symmetric monoidal category $\TT$.
\begin{itemize}
\item[(a)] The \textit{category of AQFTs over $\Loc_m$} is defined as the category
\begin{flalign}
\AQFT_m\,\coloneqq\, \Alg_{\P_{\Loc_m^\perp}}\big(\Alg_{\mathsf{uAs}}(\TT)\big)
\end{flalign}
of $\Alg_{\mathsf{uAs}}(\TT)$-valued algebras over the operad $\P_{\Loc_m^\perp}$
from Definition \ref{def:PLoc}.
This means that an AQFT is a multifunctor 
$\AAA : \P_{\Loc_m^\perp}\to \Alg_{\mathsf{uAs}}(\TT)$ and that a morphism between AQFTs
is a multinatural transformation $\zeta : \AAA\Rightarrow\BBB : \P_{\Loc_m^\perp}\to \Alg_{\mathsf{uAs}}(\TT)$.

\item[(b)] An object $\AAA \in \AQFT_m$ is said to satisfy the \textit{time-slice axiom}
if the multifunctor $\AAA : \P_{\Loc_m^\perp}\to \Alg_{\mathsf{uAs}}(\TT)$ sends every Cauchy morphism 
$f:M\to N$ in $\P_{\Loc_m^\perp}$ to an isomorphism
$\AAA(f) : \AAA(M)\stackrel{\cong}{\longrightarrow} \AAA(N)$ in $ \Alg_{\mathsf{uAs}}(\TT)$. We denote by
\begin{flalign}\label{eqn:AQFTW}
\AQFT_m^W\,\subseteq\, \AQFT_m
\end{flalign}
the full subcategory of all AQFTs satisfying the time-slice axiom.
\end{itemize}
\end{defi}

\begin{rem}\label{rem:AQFT}
The reader might be surprised that there is no reference to the Einstein causality
axiom in Definition \ref{def:AQFT}. This is due to the fact that Einstein causality 
follows by an Eckmann-Hilton argument from the structure of an 
$\Alg_{\mathsf{uAs}}(\TT)$-valued algebra over the operad $\P_{\Loc_m^\perp}$, see
\cite[Theorem 2.9]{2AQFT} for a detailed proof and also \cite[Section 2]{BSchapter}
for an intuitive argument. Furthermore, we would like to mention that the time-slice
axiom can be encoded structurally by a localization 
$L : \P_{\Loc_m^\perp}\to \P_{\Loc_m^\perp}[W^{-1}]$ of the operad
from Definition \ref{def:PLoc} at the set $W$ of Cauchy morphisms.
However, this more abstract point of view is not needed in our paper.
\end{rem}

Our main equivalence theorem in Section \ref{sec:equivalence} 
will hold true only for a certain class of AQFTs which satisfy a 
mild descent (i.e.\ local-to-global) condition that is known as 
the \textit{additivity property}, see e.g.\ \cite{FAvsAQFT}. Loosely speaking,
the additivity property demands that the algebra $\AAA(M)\in \Alg_{\mathsf{uAs}}(\TT)$ 
which is assigned to any spacetime $M\in\Loc_m$ is generated
in a suitable sense by the algebras $\AAA(U)\in \Alg_{\mathsf{uAs}}(\TT)$
which are assigned to all \textit{relatively compact} causally convex opens $U\subseteq M$,
i.e.\ the closure $\mathrm{cl}_M(U)\subseteq M$ of such subsets is compact.
The additivity property can be formalized as follows: For each object $M\in\Loc_m$, let us denote by $\RC_M$ 
the category whose objects are all relatively compact causally convex opens $U\subseteq M$
and whose morphisms are subset inclusions $U\subseteq U^\prime$.
Observe that there
exists a functor $\RC_M\to \Loc_m$ which assigns to each
object $(U\subseteq M)\in \RC_M$ the object $U\in\Loc_m$ (obtained by
endowing $U\subseteq M$ with the restricted orientation, time-orientation and metric)
and to each subset inclusion $U\subseteq U^\prime$ in $\RC_M$
the corresponding inclusion $\Loc_m$-morphism $\iota_U^{U^\prime} : U\to U^\prime$.
This implies that each $\AAA\in\AQFT_m$ can be restricted
along the multifunctor $\RC_M\to \Loc_m \to \P_{\Loc_m^\perp}$ 
to a functor $\AAA\vert : \RC_M\to \Alg_{\mathsf{uAs}}(\TT)$
on the category of relatively compact causally convex opens in $M$.
\begin{defi}\label{def:AQFTadditivity}
An object $\AAA \in \AQFT_m$ is called \textit{additive}
if the canonical $\Alg_{\mathsf{uAs}}(\TT)$-morphism
\begin{flalign}
\colim\Big(\AAA\vert : \RC_M\to \Alg_{\mathsf{uAs}}(\TT)\Big)~\stackrel{\cong}{\longrightarrow}~\AAA(M)
\end{flalign}
is an isomorphism in $\Alg_{\mathsf{uAs}}(\TT)$, for all $M\in\Loc_m$. We denote by
\begin{subequations}
\begin{flalign}
\AQFT_m^{\add}\,\subseteq \,\AQFT_m
\end{flalign}
the full subcategory of all additive AQFTs and by
\begin{flalign}
\AQFT_m^{W,\add}\,\subseteq \, \AQFT_m
\end{flalign}
\end{subequations}
the full subcategory of all additive AQFTs which also satisfy the time-slice axiom.
\end{defi}

\begin{rem}
It is shown in \cite[Lemma 2.10]{FAvsAQFT} that the category $\RC_M$ 
is filtered, for all $M\in\Loc_m$. 
Since the forgetful functor
$\Alg_{\mathsf{uAs}}(\TT)\to \TT$ preserves and reflects
filtered colimits, one can deduce additivity by verifying the 
simpler condition that 
\begin{flalign}
\colim\Big(\AAA\vert : \RC_M\to \TT\Big)~\stackrel{\cong}{\longrightarrow}~\AAA(M)
\end{flalign}
is an isomorphism in $\TT$, for all $M\in\Loc_m$.
\end{rem}


\section{\label{sec:LBord}The globally hyperbolic Lorentzian bordism pseudo-operad}
In this section we develop a non-trivial (and as we shall see 
in Section \ref{sec:equivalence}, particularly fruitful) extension of 
the $m$-dimensional globally hyperbolic Lorentzian 
bordism pseudo-category $\LB_m$ from our previous
work \cite{FFTvAQFT} to a pseudo-operad 
in the sense of Definition \ref{def:pseudooperad}. 
The resulting $m$-dimensional \textit{globally hyperbolic Lorentzian bordism 
pseudo-operad} $\LBop_m$ introduced below successfully
resolves the following main limitations of conventional bordisms 
in the globally hyperbolic Lorentzian context:
By a fundamental result in Lorentzian geometry \cite{BernalSanchez}, 
every globally hyperbolic Lorentzian spacetime $N\in\Loc_m$ has 
an underlying manifold that is diffeomorphic  $N\cong \bbR\times \Sigma$ to a cylinder.
Informally speaking, the real line $\bbR$ can be thought of as
representing a choice of time dimension,
and $\Sigma$ represents the spatial dimensions. 
The conventional type of bordisms appearing in \cite{FFTvAQFT} are described 
by globally hyperbolic Lorentzian spacetimes 
$N\in\Loc_m$. They go from a Cauchy surface $\Sigma_0\subset N$
to a later-in-time Cauchy surface $\Sigma_1\subset N$. 
It is important to stress that such bordisms encode only a concept of time evolution between Cauchy surfaces,
but are unable to capture local features in the spatial dimensions,
which are a crucial aspect of quantum field theory. 
Furthermore, global hyperbolicity forces the bordisms to be of cylindrical shape
and thereby prevents the existence of topology-changing bordisms,
which in other geometric contexts such as
\cite{StolzTeichner} capture multiplicative structures 
of the quantum field theory.
\sk

Our main idea to resolve these shortcomings of globally 
hyperbolic Lorentzian bordisms can be explained at an informal level as follows:
We will introduce below a generalized type of bordisms which are represented
by globally hyperbolic Lorentzian spacetimes $N\in\Loc_m$ and go 
from an $n$-tuple $\und{\Sigma_0} = (\Sigma_{0_1},\dots,\Sigma_{0_n})$
of causally disjoint partial Cauchy surfaces $\Sigma_{0_i}\subset N$
to a later-in-time full Cauchy surface $\Sigma_1\subset N$. See also
\eqref{eqn:bordismpicture} below for a pictorial illustration.
The relaxation from full Cauchy surfaces to partial ones allows us to capture
spatially local phenomena because one can choose the $\Sigma_{0_i}\subset N$ 
to be arbitrarily small. Furthermore, considering $n$-tuples
$\und{\Sigma_0} = (\Sigma_{0_1},\dots,\Sigma_{0_n})$
of causally disjoint partial Cauchy surfaces, for an arbitrary non-negative integer $n\in\bbN_0$, 
introduces similar multiplicative structures in our globally hyperbolic Lorentzian context 
which in other geometric contexts are encoded by topology-changing bordisms.
These multiplicative structures have a crucial physical significance
in our Lorentzian geometric context as they encode a notion of relativistic causality, see
Remark \ref{rem:meaningofnto1} for the details.
\sk

The aim of the remainder of this section is to make the above ideas 
mathematically precise. To that end, we will now define the $m$-dimensional
globally hyperbolic Lorentzian bordism pseudo-operad $\LBop_m\in\PsOp$ 
by listing the required data for a pseudo-operad from Definition~\ref{def:pseudooperad}.
\paragraph{The groupoid $(\LBop_m)_0$:}
The groupoid of objects $(\LBop_m)_0$ of the pseudo-operad $\LBop_m$ 
coincides with the groupoid of objects of the pseudo-category $\LB_m$ from \cite{FFTvAQFT}.
For completeness, let us briefly recall the definition of this groupoid:
\begin{description}
\item[\und{$\mathsf{Obj}$:}] An object in $(\LBop_m)_0$ is a pair $(M,\Sigma)$
consisting of an object $M\in\Loc_m$ and a Cauchy surface $\Sigma\subset M$.
We interpret such data as Cauchy surfaces $\Sigma$ with globally hyperbolic 
Lorentzian collar regions $M$ and we visualize them by pictures of the form
\begin{flalign}
\begin{gathered}
\begin{tikzpicture}[scale=0.9]
\draw[thick, ->] (-1.45,-0.8) -- (-1.45,0.8) node[above] {{\footnotesize time}};
\draw[dotted,fill=gray!5] (-1,0) -- (0,0.6) -- (1,0) -- (0,-0.6) -- (-1,0);
\draw (0,0.8) node {{\footnotesize $M$}}; 
\draw[thick] (-1,0) .. controls (-0.4,0.1) .. (0,0) .. controls (0.4,-0.1) .. (1,0) 
node[pos=0.2, above] {{\footnotesize $\Sigma$}};
\end{tikzpicture}
\end{gathered}\qquad.
\end{flalign}

\item[\und{$\mathsf{Mor}$:}] A morphism in $(\LBop_m)_0$
is a germ of local Cauchy morphisms $[U,g]: (M,\Sigma)\to (M^\prime,\Sigma^\prime)$.
In more detail, a morphism is an equivalence class of pairs $(U,g)$ consisting
of a causally convex open subset $U\subseteq M$ which contains the Cauchy surface $\Sigma\subset U$
and a Cauchy morphism $g:U\to M^\prime$ in $\Loc_m$ satisfying $g(\Sigma)=\Sigma^\prime$.
Two pairs $(U,g)$ and $(\widetilde{U},\widetilde{g})$ 
are equivalent if and only if there exists a causally convex open subset $\widehat{U}\subseteq U\cap \widetilde{U}\subseteq M$
which contains the Cauchy surface $\Sigma\subset \widehat{U}$, such that the restrictions
$g\vert_{\widehat{U}} = \widetilde{g}\vert_{\widehat{U}}$ coincide.
Such morphisms provide identifications between objects which are represented 
by isomorphic Cauchy surfaces with locally isomorphic collar regions,
and we visualize them by pictures of the form
\begin{flalign}
\begin{gathered}
\begin{tikzpicture}[scale=0.9]
\draw[dotted,fill=gray!5] (-1,0) -- (0,1) -- (1,0) -- (0,-1) -- (-1,0);
\draw[dotted,fill=gray!20] (-1,0) -- (0,0.5) -- (1,0) -- (0,-0.5) -- (-1,0);
\draw (0,1.2) node {{\footnotesize $M$}}; 
\draw[thick] (-1,0) .. controls (-0.4,0.1) .. (0,0) .. controls (0.4,-0.1) .. (1,0) 
node[pos=0.2, above] {{\footnotesize $\Sigma$}};
\draw[->,thick] (1.8,0) -- (1.2,0) node [midway, above] {{\footnotesize $\subseteq$}};
\draw[dotted,fill=gray!20] (2,0) -- (3,0.5) -- (4,0) -- (3,-0.5) -- (2,0);
\draw (3,0.7) node {{\footnotesize $U$}}; 
\draw[thick] (2,0) .. controls (2.6,0.1) .. (3,0) .. controls (3.4,-0.1) .. (4,0) 
node[pos=0.2, above] {{\footnotesize $\Sigma$}};
\draw[->,thick] (4.2,0) -- (4.8,0) node [midway, above] {{\footnotesize $g$}};
\draw[dotted,fill=gray!5] (5,0) -- (6,0.75) -- (7,0) -- (6,-0.75) -- (5,0);
\draw[dotted,fill=gray!20] (5,0) -- (6,0.5) -- (7,0) -- (6,-0.5) -- (5,0);
\draw (6,0.95) node {{\footnotesize $M^\prime$}}; 
\draw[thick] (5,0) .. controls (5.6,0.1) .. (6,0) .. controls (6.4,-0.1) .. (7,0) 
node[pos=0.2, above] {{\footnotesize $\Sigma^\prime$}};
\end{tikzpicture}
\end{gathered}\qquad.
\end{flalign}
\end{description}
The identity morphisms in the groupoid $(\LBop_m)_0$
are given by $[M,\id_M] : (M,\Sigma)\to(M,\Sigma)$. The composition of
two morphisms $[U,g] : (M,\Sigma)\to (M^\prime,\Sigma^\prime)$
and $[U^\prime,g^\prime] :(M^\prime,\Sigma^\prime)\to (M^{\prime\prime},\Sigma^{\prime\prime}) $
in $(\LBop_m)_0$ is given by factorizing any choice of representatives according to
\begin{subequations}
\begin{flalign}
\begin{gathered}
\xymatrix{
M &\ar[l]_-{\subseteq} U \ar[r]^-{g}& M^\prime & \ar[l]_-{\subseteq} U^\prime \ar[r]^-{g^\prime}& M^{\prime\prime}\\
&\ar[u]^-{\subseteq}g^{-1}(U^\prime) \ar@{-->}[rru]_-{g}&&&
}
\end{gathered}
\end{flalign}
and defining the composite morphism by
\begin{flalign}
[U^\prime,g^\prime]\,[U,g]\,\coloneqq\, \big[g^{-1}(U^\prime),g^\prime\,g\big]\,:\, (M,\Sigma)~\longrightarrow~(M^{\prime\prime},\Sigma^{\prime\prime})\quad.
\end{flalign}
\end{subequations}
Note that every morphism $[U,g] : (M,\Sigma)\to (M^\prime,\Sigma^\prime)$
in $(\LBop_m)_0$ is invertible with inverse given explicitly by
$[g(U),g^{-1}] : (M^\prime,\Sigma^\prime)\to (M,\Sigma)$.

\paragraph{The groupoids $(\LBop_m)_1^n$:}
The groupoids of $n$-ary operations $(\LBop_m)_1^n$ of the pseudo-operad $\LBop_m$ 
provide a non-trivial generalization of the $1$-ary globally hyperbolic Lorentzian 
bordisms between Cauchy surfaces in the pseudo-category $\LB_m$ from \cite{FFTvAQFT},
taking into account spatial locality and multiplicative structures.
Their precise definition is as follows:
\begin{description}
\item[\und{$\mathsf{Obj}$:}] An object in $(\LBop_m)_1^n$ is a tuple
$(N,\und{\iota_0},\iota_1) : (\und{M_0},\und{\Sigma_0})\hto (M_1,\Sigma_1)$ 
consisting of an object $N\in\Loc_m$ and zig-zags
\begin{flalign}\label{eqn:bordismzigzag}
\xymatrix{
\und{M_0} & \ar[l]_-{\subseteq} \und{V_0}  \ar[r]^-{\und{\iota_0}} & N &\ar[l]_-{\iota_1} V_1 \ar[r]^-{\subseteq} & M_1   
}
\end{flalign}
of operations in the operad $\P_{\Loc_m^\perp}$ from Definition \ref{def:PLoc},
where $V_{0_i}\subseteq M_{0_i}$ is a causally convex open subset which
contains the Cauchy surface $\Sigma_{0_i}\subset V_{0_i}$, for all
$i=1,\dots,n$, and $V_1\subseteq M_1$ is a causally convex open subset
which contains the Cauchy surface $\Sigma_1\subset V_1$.
The $1$-ary operation $\iota_1 : V_1\to N$ in $\P_{\Loc_m^\perp}$ is required to be a Cauchy morphism,
but the $n$-ary operation $\und{\iota_0}: \und{V_0}\to N$ in $\P_{\Loc_m^\perp}$ is allowed to be general.
(In particular, in the case of $n=1$, the $1$-ary operation $\iota_0: V_0\to N$ in $\P_{\Loc_m^\perp}$ 
does not necessarily have to be a Cauchy morphism.) The images in 
$N$ of the Cauchy surfaces must satisfy the following condition:
\begin{itemize}
\item \underline{\textit{Cauchy case:}} In the case where $n=1$ and the $1$-ary operation $\iota_0: V_0\to N$ 
in $\P_{\Loc_m^\perp}$ is a Cauchy morphism, we require that
\begin{flalign}\label{eqn:bordismsurfacesI}
\iota_0(\Sigma_0) \,\subset\,  J^-_N\big(\iota_1(\Sigma_1)\big)
\end{flalign}
is contained in the causal past in $N$ of $\iota_1(\Sigma_1)$.

\item \underline{\textit{Non-Cauchy case:}} In all other cases, we require that the closure
\begin{flalign}\label{eqn:bordismsurfacesII}
\mathrm{cl}_N\bigg(\bigcup_{i=1}^n\iota_{0_i}(\Sigma_{0_i})\bigg)\,\subset\, I^-_N\big(\iota_1(\Sigma_1)\big)
\end{flalign}
is contained in the chronological past in $N$ of $\iota_1(\Sigma_1)$.
\end{itemize}
Recall that the chronological past $I_N^-(S)\subseteq N$ of a subset $S\subseteq N$
consists of all points in $N$ which can be reached by past-pointing time-like curves
emanating from $S$, while the causal past $J_N^-(S)\subseteq N$ consists of $S\subseteq N$ 
and all points in $N$ which can be reached by past-pointing causal 
curves emanating from $S$. 
Note in particular that the condition \eqref{eqn:bordismsurfacesI} allows the images
in $N$ of different Cauchy surfaces to intersect, as long as their causal order is maintained,
while the stronger condition in \eqref{eqn:bordismsurfacesII} forbids such intersections.
Our reasons for this asymmetric treatment of the Cauchy and non-Cauchy cases
are of a technical Lorentzian geometric origin and they 
will become clearer when we define the operadic compositions in
the pseudo-operad $\LBop_m$, see Remark \ref{rem:CauchyvsnonCauchy} below.
\sk

The geometric interpretation of the $1$-ary operations 
$(N,\iota_0,\iota_1) : (M_0,\Sigma_0)\hto (M_1,\Sigma_1)$ 
with $\iota_0: V_0\to N$  a Cauchy morphism is given by globally hyperbolic Lorentzian bordisms
\begin{subequations}\label{eqn:bordismpicture}
\begin{flalign}
\begin{gathered}\begin{tikzpicture}[scale=0.9]
\draw[dotted,fill=gray!5] (-1,0) -- (0,0.75) -- (1,0) -- (0,-0.75) -- (-1,0);
\draw[dotted,fill=gray!20] (-1,0) -- (0,-0.2) -- (1,0) -- (0,-0.75) -- (-1,0);
\draw (0,0.95) node {{\footnotesize $M_0$}}; 
\draw[thick] (-1,0) .. controls (-0.4,-0.3) .. (0,-0.35) .. controls (0.4,-0.3) .. (1,0);
\draw[->,thick] (1.8,0) -- (1.2,0) node [midway, above] {{\footnotesize $\subseteq$}};
\draw[dotted,fill=gray!20] (2,0) -- (3,-0.2) -- (4,0) -- (3,-0.75) -- (2,0);
\draw (3,0.1) node {{\footnotesize $V_0$}}; 
\draw[thick] (2,0) .. controls (2.6,-0.3) .. (3,-0.35) .. controls (3.4,-0.3) .. (4,0);
\draw[->,thick] (4.2,0) -- (4.8,0) node [midway, above] {{\footnotesize $\iota_0$}};
\draw[dotted,fill=gray!5] (5,0) -- (6,1) -- (7,0) -- (6,-1) -- (5,0);
\draw[dotted,fill=gray!20] (5,0) -- (6,-0.2) -- (7,0) -- (6,-0.75) -- (5,0);
\draw[dotted,fill=gray!20] (5,0) -- (6,0.2) -- (7,0) -- (6,0.75) -- (5,0);
\draw (6,1.2) node {{\footnotesize $N$}}; 
\draw[thick] (5,0) .. controls (5.6,-0.3) .. (6,-0.35) .. controls (6.4,-0.3) .. (7,0);
\draw[thick] (5,0) .. controls (5.6,0.3) .. (6,0.35) .. controls (6.4,0.3) .. (7,0);
\draw[->,thick] (7.8,0) -- (7.2,0) node [midway, above] {{\footnotesize $\iota_1$}};
\draw[dotted,fill=gray!20] (8,0) -- (9,0.2) -- (10,0) -- (9,0.75) -- (8,0);
\draw (9,-0.1) node {{\footnotesize $V_1$}}; 
\draw[thick] (8,0) .. controls (8.6,0.3) .. (9,0.35) .. controls (9.4,0.3) .. (10,0);
\draw[->,thick] (10.2,0) -- (10.8,0) node [midway, above] {{\footnotesize $\subseteq$}};
\draw[dotted,fill=gray!5] (11,0) -- (12,0.75) -- (13,0) -- (12,-0.75) -- (11,0);
\draw[dotted,fill=gray!20] (11,0) -- (12,0.2) -- (13,0) -- (12,0.75) -- (11,0);
\draw (12,0.95) node {{\footnotesize $M_1$}}; 
\draw[thick] (11,0) .. controls (11.6,0.3) .. (12,0.35) .. controls (12.4,0.3) .. (13,0);
\end{tikzpicture}
\end{gathered}
\end{flalign}
between two full Cauchy surfaces. Note that these are precisely the bordisms 
described by the pseudo-category $\LB_m$ from \cite{FFTvAQFT}. The additional $n$-ary
operations $(N,\und{\iota_0},\iota_1) : (\und{M_0},\und{\Sigma_0})\hto (M_1,\Sigma_1)$ in $\LBop_m$
can be interpreted geometrically in terms of globally hyperbolic Lorentzian bordisms
\begin{flalign}
\begin{gathered}\begin{tikzpicture}[scale=0.9]
\draw[dotted,fill=gray!5] (-1,0) -- (-0.75,0.25) -- (-0.5,0) -- (-0.75,-0.25) -- (-1,0);
\draw[dotted,fill=gray!20] (-1,0) -- (-0.75,0.15) -- (-0.5,0) -- (-0.75,-0.15) -- (-1,0);
\draw (-0.75,0.45) node {{\footnotesize $M_{0_1}$}}; 
\draw[thick] (-1,0) .. controls (-0.75,-0.05) .. (-0.5,0);
\draw[dotted,fill=gray!5] (0.3,0) -- (0.55,0.25) -- (0.8,0) -- (0.55,-0.25) -- (0.3,0);
\draw[dotted,fill=gray!20] (0.3,0) -- (0.55,0.15) -- (0.8,0) -- (0.55,-0.15) -- (0.3,0);
\draw (0.55,0.45) node {{\footnotesize $M_{0_n}$}};
\draw[thick] (0.3,0) .. controls (0.55,-0.05) .. (0.8,0);
\draw (-0.1,0) node {{\footnotesize $\cdots$}}; 
\draw[->,thick] (1.8,0) -- (1.2,0) node [midway, above] {{\footnotesize $\subseteq$}};
\draw[dotted,fill=gray!20] (2.1,0) -- (2.35,0.15) -- (2.6,0) -- (2.35,-0.15) -- (2.1,0);
\draw (2.35,0.45) node {{\footnotesize $V_{0_1}$}}; 
\draw[thick] (2.1,0) .. controls (2.35,-0.05) .. (2.6,0);
\draw[dotted,fill=gray!20] (3.4,0) -- (3.65,0.15) -- (3.9,0) -- (3.65,-0.15) -- (3.4,0);
\draw (3.65,0.45) node {{\footnotesize $V_{0_n}$}};
\draw[thick] (3.4,0) .. controls (3.65,-0.05) .. (3.9,0);
\draw (3,0) node {{\footnotesize $\cdots$}}; 
\draw[->,thick] (4.2,0) -- (4.8,0) node [midway, above] {{\footnotesize $\und{\iota_0}$}};
\draw[dotted,fill=gray!5] (5,0) -- (6,1) -- (7,0) -- (6,-1) -- (5,0);
\draw[dotted,fill=gray!20] (5.25,-0.15) -- (5.5,0) -- (5.75,-0.15) -- (5.5,-0.3) -- (5.25,-0.15);
\draw[thick] (5.25,-0.15) .. controls (5.5,-0.2) .. (5.75,-0.15);
\draw[dotted,fill=gray!20] (6.25,-0.15) -- (6.5,0) -- (6.75,-0.15) -- (6.5,-0.3) -- (6.25,-0.15);
\draw[thick] (6.25,-0.15) .. controls (6.5,-0.2) .. (6.75,-0.15);
\draw (6.05,-0.15) node {{\footnotesize $\cdots$}}; 
\draw[dotted,fill=gray!20] (5,0) -- (6,0.2) -- (7,0) -- (6,0.75) -- (5,0);
\draw (6,1.2) node {{\footnotesize $N$}}; 
\draw[thick] (5,0) .. controls (5.6,0.3) .. (6,0.35) .. controls (6.4,0.3) .. (7,0);
\draw[->,thick] (7.8,0) -- (7.2,0) node [midway, above] {{\footnotesize $\iota_1$}};
\draw[dotted,fill=gray!20] (8,0) -- (9,0.2) -- (10,0) -- (9,0.75) -- (8,0);
\draw (9,-0.1) node {{\footnotesize $V_1$}}; 
\draw[thick] (8,0) .. controls (8.6,0.3) .. (9,0.35) .. controls (9.4,0.3) .. (10,0);
\draw[->,thick] (10.2,0) -- (10.8,0) node [midway, above] {{\footnotesize $\subseteq$}};
\draw[dotted,fill=gray!5] (11,0) -- (12,0.75) -- (13,0) -- (12,-0.75) -- (11,0);
\draw[dotted,fill=gray!20] (11,0) -- (12,0.2) -- (13,0) -- (12,0.75) -- (11,0);
\draw (12,0.95) node {{\footnotesize $M_1$}}; 
\draw[thick] (11,0) .. controls (11.6,0.3) .. (12,0.35) .. controls (12.4,0.3) .. (13,0);
\end{tikzpicture}
\end{gathered}
\end{flalign}
\end{subequations}
from an $n$-tuple of causally disjoint partial Cauchy surfaces to a full Cauchy surface.

\item[\und{$\mathsf{Mor}$:}] A morphism in $(\LBop_m)_1^n$ is a germ of local Cauchy morphisms
\begin{flalign}\label{eqn:LBop2cell}
\begin{gathered}
\xymatrix@R=0.25em@C=0.25em{
 (\und{M_0^\prime},\und{\Sigma_0^\prime}) \ar[rr]|{\sla}^{(N^\prime,\und{\iota_0^\prime},\iota_1^\prime)}&&  (M_1^\prime,\Sigma_1^\prime)\\
 &{\scriptstyle [Z,f]}~\rotatebox[origin=c]{90}{$\Longrightarrow$}~~~&\\
 (\und{M_0},\und{\Sigma_0}) \ar[rr]|{\sla}_{(N,\und{\iota_0},\iota_1)}&&  (M_1,\Sigma_1)
}
\end{gathered}\quad.
\end{flalign}
In more detail, a morphism is an equivalence class of pairs $(Z,f)$ 
consisting of a causally convex open subset $Z\subseteq N$ which 
contains within $N$ the images of the Cauchy surfaces\footnote{Note that,
as a consequence of causal convexity of $Z\subseteq N$, 
this condition implies that the intermediate causal region 
$\bigcup_{i=1}^n J_N^+\big(\iota_{0_i}(\Sigma_{0_i})\big) \cap 
J_N^-\big(\iota_1(\Sigma_1)\big)\subseteq Z$ is also contained in $Z$.}
$\bigcup_{i=1}^n \iota_{0_i}(\Sigma_{0_i}) \cup \iota_1(\Sigma_1)\subset Z$
and a Cauchy morphism $f:Z\to N^\prime$ in $\Loc_m$ satisfying 
$f\big(\iota_{0_i}(\Sigma_{0_i})\big)=\iota^\prime_{0_i}(\Sigma^\prime_{0_i})$, for all $i=1,\dots,n$,
and $f\big(\iota_1(\Sigma_1)\big) = \iota^\prime_1(\Sigma^\prime_1)$.
Two pairs $(Z,f)$ and $(\widetilde{Z},\widetilde{f})$ 
are equivalent if and only if there exists a causally convex open $\widehat{Z}\subseteq Z\cap \widetilde{Z}\subseteq N$
which contains the images in $N$ of the Cauchy surfaces 
$\bigcup_{i=1}^n \iota_{0_i}(\Sigma_{0_i}) \cup \iota_1(\Sigma_1)\subset \widehat{Z}$,
such that the restrictions
$f\vert_{\widehat{Z}} = \widetilde{f}\vert_{\widehat{Z}}$ coincide.
\end{description}
The identity morphisms and compositions in the groupoid
of $n$-ary operations $(\LBop_m)_1^n$ are analogous to those in the groupoid of objects $(\LBop_m)_0$,
which we have spelled out above.

\paragraph{Source and target functors:} The source functor is defined by
\begin{subequations}\label{eqn:sourcefunctor}
\begin{flalign}
 s^n\,:\,  (\LBop_m)_1^n~&\longrightarrow~ (\LBop_m)_0^{\times n}\quad,\\
\nn \Big((N,\und{\iota_0},\iota_1) : (\und{M_0},\und{\Sigma_0})\hto (M_1,\Sigma_1)\Big)~&\longmapsto~(\und{M_0},\und{\Sigma_0})\quad,\\
\nn \begin{gathered}
\xymatrix@R=0.25em@C=0.25em{
 (\und{M_0^\prime},\und{\Sigma_0^\prime}) \ar[rr]|{\sla}^{(N^\prime,\und{\iota_0^\prime},\iota_1^\prime)}&&  (M_1^\prime,\Sigma_1^\prime)\\
 &{\scriptstyle [Z,f]}~\rotatebox[origin=c]{90}{$\Longrightarrow$}~~~&\\
 (\und{M_0},\und{\Sigma_0}) \ar[rr]|{\sla}_{(N,\und{\iota_0},\iota_1)}&&  (M_1,\Sigma_1)
}
\end{gathered} ~&\longmapsto~
\begin{gathered}
\xymatrix@R=1em@C=0.25em{
 (\und{M_0^\prime},\und{\Sigma_0^\prime}) \\
    ~\\
 \ar[uu]^-{\big[\und{\iota_0^{-1}} \big( f^{-1} \big(\und{\iota_0^\prime}(\und{V_0^\prime})\big)\big),\, 
 \und{\iota_0^{\prime -1}} \,f\, \und{\iota_0} \big]} 
 (\und{M_0},\und{\Sigma_0})
}
\end{gathered}\quad,
\end{flalign}
where the $n$-tuple of $(\LBop_m)_0$-morphisms is obtained component-wise by the factorization
\begin{flalign}
\resizebox{.9\hsize}{!}{$
\begin{gathered}
\xymatrix@C=2em{
M_{0_i} &
\ar[l]_-{\subseteq} V_{0_i} \ar[r]^-{\iota_{0_i}} & 
N &
\ar[l]_-{\subseteq} Z \ar[r]^-{f} & 
N^\prime &
\ar[l]_-{\iota_{0_i}^\prime} \ar[dl]_-{\cong}^-{\iota_{0_i}^\prime} V_{0_i}^\prime \ar[r]^-{\subseteq} & 
M_{0_i}^\prime\\
  &  
\ar[u]^-{\subseteq} \iota_{0_i}^{-1}\big(f^{-1}\big(\iota_{0_i}^\prime(V_{0_i}^\prime)\big)\big)  
\ar@{-->}[rr]_-{\iota_{0_i}} & 
  &   
\ar[u]^-{\subseteq} f^{-1}\big(\iota_{0_i}^\prime(V_{0_i}^\prime)\big) \ar@{-->}[r]_-{f} &   
\ar[u]^-{\subseteq} \iota_{0_i}^\prime(V_{0_i}^\prime)  &            &
}
\end{gathered}\quad,$}
\end{flalign}
\end{subequations}
for all $i=1,\dots,n$.
The target functor is defined similarly by
\begin{subequations}\label{eqn:targetfunctor}
\begin{flalign}
t^n\,:\,  (\LBop_m)_1^n~&\longrightarrow~ (\LBop_m)_0\quad,\\
\nn \Big((N,\und{\iota_0},\iota_1) : (\und{M_0},\und{\Sigma_0})\hto (M_1,\Sigma_1)\Big)~&\longmapsto~(M_1,\Sigma_1)\quad,\\
\nn \begin{gathered}
\xymatrix@R=0.25em@C=0.25em{
 (\und{M_0^\prime},\und{\Sigma_0^\prime}) \ar[rr]|{\sla}^{(N^\prime,\und{\iota_0^\prime},\iota_1^\prime)}&&  (M_1^\prime,\Sigma_1^\prime)\\
 &{\scriptstyle [Z,f]}~\rotatebox[origin=c]{90}{$\Longrightarrow$}~~~&\\
 (\und{M_0},\und{\Sigma_0}) \ar[rr]|{\sla}_{(N,\und{\iota_0},\iota_1)}&&  (M_1,\Sigma_1)
}
\end{gathered} ~&\longmapsto~
\begin{gathered}
\xymatrix@R=1.25em@C=0.25em{
 (M_1^\prime,\Sigma_1^\prime) \\
    ~\\
 \ar[uu]^-{\big[\iota_1^{-1} \big( f^{-1} \big(\iota_1^\prime(V_1^\prime)\big)\big), \iota_1^{\prime -1} \,f\, \iota_1 \big]} (M_1,\Sigma_1)
}
\end{gathered}\quad,
\end{flalign}
where the $(\LBop_m)_0$-morphism is obtained by the factorization
\begin{flalign}
\begin{gathered}
\xymatrix@C=2em{
M_1 &
\ar[l]_-{\subseteq} V_1 \ar[r]^-{\iota_1} & 
N &
\ar[l]_-{\subseteq} Z \ar[r]^-{f} & 
N^\prime & 
\ar[l]_-{\iota_1^\prime} \ar[dl]_-{\cong}^-{\iota_1^\prime} V_1^\prime \ar[r]^-{\subseteq} & 
M_1^\prime\\
  &  
\ar[u]^-{\subseteq} \iota_1^{-1}\big(f^{-1}\big(\iota_1^\prime(V_1^\prime)\big)\big)   \ar@{-->}[rr]_-{\iota_1} & 
  &   
\ar[u]^-{\subseteq} f^{-1}\big(\iota_1^\prime(V_1^\prime)\big) \ar@{-->}[r]_-{f} &   
\ar[u]^-{\subseteq} \iota_1^\prime(V_1^\prime)  &            &
}
\end{gathered}\quad.
\end{flalign}
\end{subequations}

\paragraph{Operadic compositions:} To define the operadic composition functors
\begin{flalign}\label{eqn:bordismcmp}
\cmp\,:\, (\LBop_m)_1^{n}\times_{(\LBop_m)_{0}^{\times n}} (\LBop_m)_1^{\und{k}} 
~\longrightarrow~(\LBop_m)_1^{\Sigma\und{k}}\quad,
\end{flalign}
let us consider any $n$-ary operation $(N_1,\und{\iota_{10}},\iota_{11}) : 
(\und{M_{1}},\und{\Sigma_{1}})\hto (M_{2},\Sigma_{2})$ in $(\LBop_m)_1^{n}$ and
any $n$-tuple $(\und{N_0},\und{\und{\iota_{00}}},\und{\iota_{01}}) : 
(\und{\und{M_{0}}},\und{\und{\Sigma_{0}}})\hto (\und{M_{1}},\und{\Sigma_{1}})$
of $k_i$-ary operations $(N_{0_i},\und{\iota_{00_i}},\iota_{01_i}) : 
(\und{M_{0_i}},\und{\Sigma_{0_i}})\hto (M_{1_i},\Sigma_{1_i})$ in $(\LBop_m)_1^{k_i}$,
for $i=1,\dots,n$. Recall from \eqref{eqn:bordismzigzag} that these operations
are represented by zig-zags
\begin{flalign}\label{eqn:bordismzigzagcomposite} 
\xymatrix{
\und{\und{M_0}} & \ar[l]_-{\subseteq} \und{\und{V_{00}}}  \ar[r]^-{\und{\und{\iota_{00}}}} &
\und{N_{0}}
&\ar[l]_-{\und{\iota_{01}}} \und{V_{01}} \ar[r]^-{\subseteq} & \und{M_{1}}
& \ar[l]_-{\subseteq} \und{V_{10}}  \ar[r]^-{\und{\iota_{10}}} & N_1 
&\ar[l]_-{\iota_{11}} V_{11} \ar[r]^-{\subseteq} & M_2
}
\end{flalign}
of operations in the operad $\P_{\Loc_m^\perp}$.
\sk

To explain the key ideas, let us start with an informal discussion 
of the operadic composition $(N_1,\und{\iota_{10}},\iota_{11}) 
\cmp(\und{N_0},\und{\und{\iota_{00}}},\und{\iota_{01}})$ of these operations in $\LBop_m$.
Roughly speaking, this composition will be given by the 
gluing (i.e.\ a pushout) of bordisms, but in order to make this well-defined 
one has to remove certain parts of the collar regions of the individual bordisms
which would otherwise obstruct the glued bordism from being an object in $\Loc_m$.
For this purpose, we define the subset
\begin{subequations}\label{eqn:overhangingremoved}
\begin{flalign}\label{eqn:overhangingremoved1}
N_1^+\,\coloneqq\, \Bigg(N_1 \setminus \mathrm{cl}_{N_1}\bigg(\bigcup_{i=1}^n J^-_{N_1}\big(\iota_{10_i}(\Sigma_{1_i})\big)\bigg)\Bigg) \cup 
\bigcup_{i=1}^n \iota_{10_i}\big(V_{01_i}\cap V_{10_i}\big)\,\subseteq\,N_1\quad,
\end{flalign}
where $\mathrm{cl}_{N_1}$ denotes the closure of subsets in $N_1$, and the subsets
\begin{flalign}\label{eqn:overhangingremoved2}
N_{0_i}^-\,\coloneqq\, J_{N_{0_i}}^-\Big(\iota_{01_i}\big(V_{01_i}\cap V_{10_i}\big)\Big)\,\subseteq\, N_{0_i}\quad,
\end{flalign}
\end{subequations}
for all $i=1,\dots,n$. Observe that the `later' bordism $N_1$ gets a part of its past collar region
removed, while the `earlier' bordisms $N_{0_i}$ get a part of their future collar regions removed. 
The removed parts are determined by the causal structure and the intersections
$V_{01_i}\cap V_{10_i} \subseteq M_{1_i}$ of the collar regions of the intermediate Cauchy surfaces
$\Sigma_{1_i}\subset M_{1_i}$. A graphical illustration of the subsets in
\eqref{eqn:overhangingremoved} is given by
\begin{flalign}\label{eqn:bordismremoved}
\begin{gathered}
\begin{tikzpicture}[scale=1]
\draw[dotted,fill=gray!5] (9.75,0) -- (11,1.25) -- (12.25,0) -- (11,-1.25) -- (9.75,0);
\draw[dotted,fill=gray!20] (9.75,0) -- (11,0.5) -- (12.25,0) -- (11,-0.5) -- (9.75,0);
\draw (11,1.5) node {{\footnotesize $N_{0_i}^-$}}; 
\draw[thick] (9.75,0) .. controls (11,-0.1) .. (12.25,0);
\draw[black, fill=darkred!50] 
(9.75,0) -- (11,0.5) -- (12.25,0) -- (11,1.25) -- (9.75,0);
\end{tikzpicture}
\qquad\qquad
\begin{tikzpicture}[scale=1]
\draw[dotted,fill=gray!5] (4.75,0) -- (6,1.25) -- (7.25,0) -- (6,-1.25) -- (4.75,0);
\draw[dotted,fill=gray!20] (5.25,0) -- (5.5,0.15) -- (5.75,0) -- (5.5,-0.15) -- (5.25,0);
\draw[thick] (5.25,0) .. controls (5.5,-0.05) .. (5.75,0);
\draw[dotted,fill=gray!20] (6.25,0) -- (6.5,0.15) -- (6.75,0) -- (6.5,-0.15) -- (6.25,0);
\draw[thick] (6.25,0) .. controls (6.5,-0.05) .. (6.75,0);
\draw (6,1.5) node {{\footnotesize $N_1^+$}}; 
\draw[black, fill=darkred!50] 
(5,-0.25) -- (5.25,0) -- (5.5,-0.15) -- (5.75,0) -- (6,-0.25) -- (6.25,0) -- (6.5,-0.15)-- (6.75,0) -- (7,-0.25) -- (6,-1.25) -- (5,-0.25);
\end{tikzpicture}
\end{gathered}\qquad,
\end{flalign}
where the red parts are removed. From these pictures it looks plausible
that the bordism $N_1^+$ can be glued with the family of bordisms 
$\und{N_{0}^-} = (N_{0_1}^-,\dots,N_{0_n}^-)$ along their common dark gray regions, 
because all `overhanging' parts of the collar regions 
of the individual bordisms which would obstruct the gluing have been removed.
\begin{rem}\label{rem:CauchyvsnonCauchy}
We can now explain the origin of the asymmetry between the conditions
for Cauchy bordisms in \eqref{eqn:bordismsurfacesI} and for non-Cauchy bordisms in \eqref{eqn:bordismsurfacesII}.
The more restrictive condition in the non-Cauchy case \eqref{eqn:bordismsurfacesII} is 
needed to ensure that the subset $N_1^+\subseteq N_1$ in \eqref{eqn:overhangingremoved1} 
contains the image $\iota_{11}(\Sigma_2)\subset N_1$ of the `later' full Cauchy surface. 
Indeed, if one would weaken this condition by dropping
the closures or by using instead of the \textit{chronological} past $I_N^-$
the \textit{causal} past $J_N^-$ as in \eqref{eqn:bordismsurfacesI}, 
the image $\iota_{11}(\Sigma_2)\subset N_1$ of the `later' full Cauchy surface
may intersect the subset which is removed in \eqref{eqn:overhangingremoved1} 
(in particular, in this case the earlier partial Cauchy surfaces and the 
later Cauchy surface would be allowed to intersect).
That, however, would imply that $\iota_{11}(\Sigma_2)\not\subset N_1^+\subseteq N_1$.
Pictorially as on the right-hand side in
\eqref{eqn:bordismremoved}, the problematic parts are the tips at the top
of the red-colored region, which may intersect non-trivially with 
the image $\iota_{11}(\Sigma_2)\subset N_1$ of the `later' full Cauchy surface.
We would like to stress that this issue does not arise for $n=1$
and $\iota_{10} :V_{10}\to N_1$ a Cauchy morphism, because in this case
the subset $N_{1}^+\subseteq N_1$ in \eqref{eqn:overhangingremoved1} simplifies
to $N_1^+ = J_{N_1}^+\big(\iota_{10}(V_{01}\cap V_{10})\big)\subseteq N_1$
(note that this is similar to \eqref{eqn:overhangingremoved2} with past and future exchanged)
and it contains the image $\iota_{11}(\Sigma_2)\subset N_1$ of the `later' full Cauchy surface 
also under the weaker conditions in \eqref{eqn:bordismsurfacesI}. Note that these weaker conditions 
in the Cauchy case are required for the pseudo-operad $\LBop_m$ to admit operadic units, 
see the corresponding paragraph below.
\end{rem}

Let us now make the above gluing construction precise.
As a first step, we prove in 
Proposition~\ref{prop:bordismgluingregions} that the subsets in \eqref{eqn:overhangingremoved} 
are causally convex and open, so that they inherit the structure of objects in $\Loc_m$,
and that they contain the images of the Cauchy surfaces.
Next, we concatenate as in \eqref{eqn:bordismzigzagcomposite} the zig-zags of operations in $\P_{\Loc_m^\perp}$
representing the individual bordisms and factorize the result according to
\begin{flalign}\label{eqn:bordismcomposition}
\begin{gathered}
\xymatrix@R=1.5em@C=1em{
\und{\und{M_0}} & \ar[l]_-{\subseteq} \und{\und{V_{00}}} \ar[r]^-{\und{\und{\iota_{00}}}} & \und{N_0}
& \ar[l]_-{\und{\iota_{01}}} \und{V_{01}} \ar[r]^-{\subseteq} & \und{M_1} &\ar[l]_-{\subseteq}\und{V_{10}}
\ar[r]^-{\und{\iota_{10}}} & N_1  & \ar[l]_-{\iota_{11}} V_{11} \ar[r]^-{\subseteq} & M_2\\
&\ar@/_3pc/[dddrrr]_-{\iota_{-}\,\und{\und{\iota_{00}}}}\ar[u]^-{\subseteq} \und{\und{\iota_{00}^{-1}}}(\und{N_0^-})
\ar[dr]^-{\und{\und{\iota_{00}}}}&&& \ar[dll]_-{\und{\iota_{01}}}\ar[lu]^-{\subseteq}\und{V_{01}\cap V_{10}} 
\ar[d]^-{\subseteq}\ar[ru]_-{\subseteq}
\ar[drr]^-{\und{\iota_{10}}}&&& \iota_{11}^{-1}(N_1^+)  \ar@/^3pc/[dddlll]^-{\iota_+\,\iota_{11}}\ar[dl]_-{\iota_{11}}\ar[u]_-{\subseteq}&\\
&&\ar[dr]_-{\subseteq} \und{N_0^-}\ar[uu]_-{\subseteq}&& 
\ar[dl]_-{\iota_{01}} V_{01}\cap V_{10}\ar[dr]^-{\iota_{10}} &&\ar@{=}[ld] N_1^+ \ar[uu]^-{\subseteq}&&\\
&&&\ar@{-->}[dr]_-{\iota_-} N_0^- & & \ar@{-->}[dl]^-{\iota_+}N_1^+&&&\\
&&&& N_0^- \sqcup_{V_{01}\cap V_{10}}^{~}  N_1^+&&&&
}
\end{gathered}\quad.
\end{flalign}
In the bottom part of this commutative diagram of operations in $\P_{\Loc_m^\perp}$, 
we use the condensed notations 
$V_{01}\cap V_{10} \coloneqq \bigsqcup_{i=1}^n (V_{01_i}\cap V_{10_i})$
and $N_0^-\coloneqq \bigsqcup_{i=1}^n N_{0_i}^-$
for the disjoint unions of the $n$-tuples $\und{V_{01}\cap V_{10}}$ 
and $\und{N_0^-}$ of objects in $\Loc_m$. We show in Proposition \ref{prop:pushout} 
that the pushout $N_0^- \sqcup_{V_{01}\cap V_{10}}^{~}  N_1^+$ exists as an object in $\Loc_m$.
The composite $\Sigma\und{k}$-ary operation in $\LBop_m$ is then 
defined by the outer zig-zags in the diagram
\eqref{eqn:bordismcomposition}, i.e.
\begin{flalign}\label{eqn:bordismcomp}
(N_1,\und{\iota_{10}},\iota_{11}) 
\cmp(\und{N_0},\und{\und{\iota_{00}}},\und{\iota_{01}})\,\coloneqq\,
\Big(N_0^- \sqcup_{V_{01}\cap V_{10}}^{~}  N_1^+ ,\,\iota_{-}\,\und{\und{\iota_{00}}},\, \iota_{+}\,\iota_{11} \Big) \,:\,
(\und{\und{M_0}},\und{\und{\Sigma_0}})\,\hto\, (M_2,\Sigma_2)\quad.
\end{flalign}

Next, we define the operadic composition functor \eqref{eqn:bordismcmp} on $2$-cells.
Given two operadically composable (tuples of) $2$-cells
\begin{flalign}
\begin{gathered}
\xymatrix@R=0.25em@C=0.25em{
 (\und{\und{M_0^\prime}},\und{\und{\Sigma_0^\prime}}) \ar[rr]|{\sla}^{(\und{N_0^\prime},\und{\und{\iota_{00}^\prime}},\und{\iota_{01}^\prime})}&&  (\und{M_1^\prime},\und{\Sigma_1^\prime})
\ar[rr]|{\sla}^{(N_1^\prime,\und{\iota_{10}^\prime},\iota_{11}^\prime)}   &&  (M_2^\prime,\Sigma_2^\prime)\\
 &{\scriptstyle [\und{Z_0},\und{f_0}]}~\rotatebox[origin=c]{90}{$\Longrightarrow$}~~~& &{\scriptstyle [Z_1,f_1]}~\rotatebox[origin=c]{90}{$\Longrightarrow$}~~~&\\
\ar[uu]^-{s^{\und{k}}[\und{Z_0},\und{f_0}]} (\und{\und{M_0}},\und{\und{\Sigma_0}}) 
\ar[rr]|{\sla}_{(\und{N_0},\und{\und{\iota_{00}}},\und{\iota_{01}})}&&  \ar[uu] (\und{M_1},\und{\Sigma_1}) 
\ar[rr]|{\sla}_{(N_1,\und{\iota_{10}},\iota_{11})} &&  (M_2,\Sigma_2) \ar[uu]_-{t^n[Z_1,f_1]}
}
\end{gathered}\quad,
\end{flalign}
i.e.\ $t^{\und{k}}[\und{Z_0},\und{f_0}] = s^n[Z_1,f_1]$ in $(\LBop_m)_0^{\times n}$, 
there exists, by definition of source \eqref{eqn:sourcefunctor} 
and target \eqref{eqn:targetfunctor}, a family of causally convex open subsets 
$U_i\subseteq V_{01_i}\cap V_{10_i}\subseteq M_{1_i}$ 
containing the Cauchy surfaces $\Sigma_{1_i}\subset U_i$,
such that $f_i^\cap \coloneqq \iota_{01_i}^{\prime -1} \,f_{0_i}\, \iota_{01_i}\big\vert_{U_i} =
\iota_{10_i}^{\prime -1} \,f_{1}\, \iota_{10_i}\big\vert_{U_i}$, for all $i=1,\dots,n$.
We define the subsets
\begin{subequations}
\begin{flalign}
Z_i^- \,&\coloneqq\, N_{0_i}^{-}\cap f_{0_i}^{-1}(N_{0_i}^{\prime -})\,\subseteq\,N_{0_i}\quad,\\
Z^+\,&\coloneqq\,  N_1^+ \cap f_1^{-1}(N_1^{\prime +})\,\subseteq N_1\quad,\\
Z_i^\cap\,&\coloneqq\, \iota_{01_i}^{-1}(Z_i^-)\cap \iota_{10_i}^{-1}(Z^+) \cap U_i\,\subseteq M_{1_i}\quad,
\end{flalign}
\end{subequations}
for all $i=1,\dots,n$. Denoting again disjoint unions by the condensed notations
$Z^- \coloneqq \bigsqcup_{i=1}^n Z_i^-$ and $Z^\cap \coloneqq\bigsqcup_{i=1}^n Z_i^\cap $, we obtain the commutative diagram
\begin{subequations}
\begin{flalign}
\begin{gathered}
\xymatrix{
N_0^{\prime -} &
\ar[l]_-{i_{01}^\prime} V_{01}^\prime \cap V_{10}^\prime  \ar[r]^-{i_{10}^\prime}&
N_1^{\prime +}\\
\ar[d]_-{\subseteq} Z^- \ar[u]^-{f_0}& 
\ar[d]_-{\subseteq} \ar[l]_-{i_{01}}   Z^{\cap} \ar[r]^-{i_{10}} \ar[u]^-{f^\cap}& 
\ar[d]^-{\subseteq} Z^+ \ar[u]_-{f_1}\\
N_0^- & \ar[l]^-{i_{01}} V_{01}\cap V_{10}  \ar[r]_-{i_{10}}& 
N_1^+
}
\end{gathered}
\end{flalign}
of Cauchy morphisms, which induces morphisms
\begin{flalign}
\xymatrix@C=3.5em{
N_0^{-}\sqcup_{V_{01}\cap V_{10}}^{~} N_1^+ &\ar[l]_-{\subseteq}
Z^-\sqcup_{Z^\cap}^{~} Z^+ \ar[r]^-{f_0\sqcup_{f^\cap}^{} f_1} &
N_0^{\prime -} \sqcup_{V_{01}^\prime \cap V_{10}^\prime}^{~} N_1^{\prime +}
}
\end{flalign}
\end{subequations}
between the pushouts entering the operadic compositions in \eqref{eqn:bordismcomposition}.
The operadic composite of $2$-cells 
\begin{subequations}\label{eqn:composed2cell}
\begin{flalign}
[Z_1,f_1]\cmp[\und{Z_0},\und{f_0}] \,:\, (N_1,\und{\iota_{10}},\iota_{11})\cmp(\und{N_0},\und{\und{\iota_{00}}},\und{\iota_{01}})~\Longrightarrow~
(N_1^\prime,\und{\iota_{10}^\prime},\iota_{11}^\prime)\cmp(\und{N_0^\prime},\und{\und{\iota_{00}^\prime}},\und{\iota_{01}^\prime})
\end{flalign}
is then defined by
\begin{flalign}
[Z_1,f_1]\cmp[\und{Z_0},\und{f_0}]\,\coloneqq\, \big[Z^-\sqcup_{Z^\cap}^{~} Z^+, f_0\sqcup_{f^\cap}^{} f_1\big]\quad.
\end{flalign}
\end{subequations}

\paragraph{Operadic units:} We define the operadic unit functor by
\begin{flalign}\label{eqn:identitybordism}
u\,:\, (\LBop_m)_0~&\longrightarrow~(\LBop_m)_1^1\quad,\\
\nn (M,\Sigma)~&\longmapsto~ \Big((M,\id_M,\id_M) : (M,\Sigma)\hto(M,\Sigma)\Big)\quad,\\
\nn \begin{gathered}
\xymatrix@R=1.25em@C=0.25em{
 (M^\prime,\Sigma^\prime) \\
    ~\\
 \ar[uu]^-{[U,g]} (M,\Sigma)
}
\end{gathered}~&\longmapsto~
\begin{gathered}
\xymatrix@R=0.25em@C=0.5em{
 (M^\prime,\Sigma^\prime) \ar[rr]|{\sla}^{(M^\prime,\id_{M^\prime},\id_{M^\prime})}&&  (M^\prime,\Sigma^\prime)\\
 &{\scriptstyle [U,g]}~\rotatebox[origin=c]{90}{$\Longrightarrow$}~~~&\\
 (M,\Sigma) \ar[rr]|{\sla}_{(M,\id_M,\id_M)}&&  (M,\Sigma)
}
\end{gathered}
\quad,
\end{flalign}
where $(M,\id_M,\id_M) : (M,\Sigma)\hto(M,\Sigma)$ is the identity bordism from the full Cauchy surface
$\Sigma$ to itself. Note that this bordism satisfies our condition in \eqref{eqn:bordismsurfacesI} 
for the Cauchy case.

\paragraph{Permutation actions:} We define the permutation actions by
\begin{flalign}\label{eqn:permutationbordisms}
(\LBop_m)_1^n(\sigma)\,:\, (\LBop_m)_1^n~&\longrightarrow~(\LBop_m)_1^n\quad,\\
\nn \Big((N,\und{\iota_0},\iota_1) : (\und{M_0},\und{\Sigma_0})\hto (M_1,\Sigma_1)\Big)~&\longmapsto~
\nn \Big((N,\und{\iota_0}\sigma,\iota_1) : (\und{M_0}\sigma,\und{\Sigma_0}\sigma)\hto (M_1,\Sigma_1)\Big)\quad,\\
\nn \begin{gathered}
\xymatrix@R=0.25em@C=0.25em{
 (\und{M_0^\prime},\und{\Sigma_0^\prime}) \ar[rr]|{\sla}^{(N^\prime,\und{\iota_0^\prime},\iota_1^\prime)}&&  (M_1^\prime,\Sigma_1^\prime)\\
 &{\scriptstyle [Z,f]}~\rotatebox[origin=c]{90}{$\Longrightarrow$}~~~&\\
 (\und{M_0},\und{\Sigma_0}) \ar[rr]|{\sla}_{(N,\und{\iota_0},\iota_1)}&&  (M_1,\Sigma_1)
}
\end{gathered} ~&\longmapsto~
\begin{gathered}
\xymatrix@R=0.25em@C=0.25em{
 (\und{M_0^\prime}\sigma,\und{\Sigma_0^\prime}\sigma) \ar[rr]|{\sla}^{(N^\prime,\und{\iota_0^\prime}\sigma,\iota_1^\prime)}&&  (M_1^\prime,\Sigma_1^\prime)\\
 &{\scriptstyle [Z,f]}~\rotatebox[origin=c]{90}{$\Longrightarrow$}~~~&\\
 (\und{M_0}\sigma,\und{\Sigma_0}\sigma) \ar[rr]|{\sla}_{(N,\und{\iota_0}\sigma,\iota_1)}&&  (M_1,\Sigma_1)
}
\end{gathered}\quad,
\end{flalign}
for all $\sigma\in S_n$, where we recall that the right $S_n$-action on $n$-tuples
is given by permuting the tuple, e.g.\ $\und{M_0}\sigma = (M_{0_{\sigma(1)}},\dots,M_{0_{\sigma(n)}})$
and $\und{\iota_0}\sigma = (\iota_{0_{\sigma(1)}},\dots, \iota_{0_{\sigma(n)}})$.

\paragraph{Coherence isomorphisms:}
Operadic compositions of bordisms in $\LBop_m$ are \textit{not} strictly associative and unital 
due to potential mismatches of the gluing regions $\und{V_{01}\cap V_{10}}\subseteq \und{M_{1}}$ 
and the collar regions $\und{\und{i_{00}^{-1}}}(\und{N_0^-})\subseteq \und{\und{M_0}}$ 
and $i_{11}^{-1}(N_1^+)\subseteq M_2$
in the defining diagram \eqref{eqn:bordismcomposition}, as well as due to the fact that
pushouts $N_0^-\sqcup_{V_{01}\cap V_{10}}^{~} N_1^+$ are only defined uniquely up to canonical isomorphism.
Note that this is completely analogous to the case of the globally hyperbolic
Lorentzian bordism pseudo-category $\LB_m$ in \cite{FFTvAQFT} and the 
geometric bordism pseudo-categories in \cite{StolzTeichner}.
One can show that these potential mismatches lead to canonical globular $2$-cell isomorphisms 
between the corresponding bordisms, which define the associator and the unitor coherence isomorphisms 
$(\mathsf{a},\mathsf{l},\mathsf{r})$ of the pseudo-operad $\LBop_m$.
We refer the reader to \cite{FFTvAQFT} for the technical details, which generalize directly
from pseudo-categories to our present case of pseudo-operads.
\begin{rem}\label{rem:LBordsubcat}
The globally hyperbolic Lorentzian bordism pseudo-operad $\LBop_m$ defined above has an underlying
pseudo-category which is formed by its $1$-ary operations. It is important
to emphasize that this pseudo-category does \textit{not} coincide with 
the globally hyperbolic Lorentzian bordism pseudo-category $\LB_m$ from \cite{FFTvAQFT}:
These two pseudo-categories contain the same objects, but the 
horizontal morphisms in $\LB_m$ describe only Cauchy bordisms, i.e.\ $1$-ary operations
$(N,\iota_0,\iota_1) : (M_0,\Sigma_0)\hto (M_1,\Sigma_1)$ with $\iota_0 : V_0\to N$
a Cauchy morphism, while the pseudo-operad $\LBop_m$
contains also $1$-ary operations with $\iota_0 : V_0\to N$ non-Cauchy.
Nevertheless, the bordism pseudo-category from \cite{FFTvAQFT}
is a non-full sub-pseudo-category $\LB_m\subseteq \LBop_m$ of our bordism pseudo-operad.
\end{rem}

\begin{propo}\label{prop:LBopfibrant}
For each $m\in\bbN$, the $m$-dimensional
globally hyperbolic Lorentzian bordism pseudo-operad defined above is
fibrant in the sense of Definition \ref{def:fibrant}, 
i.e.\ $\LBop_m\in\PsOp^{\mathrm{fib}}$. A choice of companion
for a vertical morphism $[U,g]: (M,\Sigma)\to (M^\prime,\Sigma^\prime)$
is given by the $1$-ary operation 
$(M^\prime,g,\id_{M^\prime}) : (M,\Sigma)\hto (M^\prime,\Sigma^\prime)$
which is defined by the zig-zags of Cauchy morphisms
\begin{flalign}
\xymatrix{
M &\ar[l]_-{\subseteq} U \ar[r]^-{g}& M^\prime & \ar[l]_-{\,\id_{M^\prime}} M^\prime \ar[r]^-{=}& M^\prime 
}\quad,
\end{flalign}
together with the $2$-cells
\begin{flalign}
\begin{gathered}
\xymatrix@R=0.25em@C=0.25em{
 (M^\prime,\Sigma^\prime) \ar[rr]|{\sla}^{(M^\prime,\id_{M^\prime},\id_{M^\prime})} &&  (M^\prime,\Sigma^\prime) \\
 &{\scriptstyle [M^\prime,\id_{M^\prime}]}~\rotatebox[origin=c]{90}{$\Longrightarrow$}~~&\\
\ar[uu]^-{[U,g]} (M,\Sigma) \ar[rr]|{\sla}_{(M^\prime,g,\id_{M^\prime})}&&  \ar@{=}[uu] (M^\prime,\Sigma^\prime) 
}
\end{gathered}
\qquad\text{and}\qquad
\begin{gathered}
\xymatrix@R=0.25em@C=0.25em{
 (M,\Sigma) \ar[rr]|{\sla}^{(M^\prime,g,\id_{M^\prime})} &&  (M^\prime,\Sigma^\prime) \\
 &{\scriptstyle [U,g]}~\rotatebox[origin=c]{90}{$\Longrightarrow$}~~&\\
\ar@{=}[uu] (M,\Sigma) \ar[rr]|{\sla}_{(M,\id_M,\id_M)}&&  \ar[uu]_-{[U,g]}  (M,\Sigma) 
}
\end{gathered}\quad.
\end{flalign}
\end{propo}
\begin{proof}
Observe that our candidate for the companion is given by
a horizontal morphism and $2$-cells in the sub-pseudo-category $\LB_m\subseteq \LBop_m$
because it consists only of Cauchy morphisms, see also Remark \ref{rem:LBordsubcat}.
It coincides with the companion used in \cite[Proposition 3.3]{FFTvAQFT}
to prove fibrancy of the pseudo-category $\LB_m$, hence our claim
follows from this result.
\end{proof}


\section{\label{sec:FFTs}Globally hyperbolic Lorentzian functorial QFTs}
In this section we define a concept of globally hyperbolic Lorentzian functorial QFTs (FQFTs)
which is richer than the one previously introduced in \cite{FFTvAQFT}.
This increased richness comes from replacing the globally hyperbolic Lorentzian
bordism pseudo-category $\LB_m$ from \cite{FFTvAQFT} with our novel 
bordism pseudo-operad $\LBop_m$ from Section \ref{sec:LBord}.
Loosely speaking, an $m$-dimensional globally hyperbolic Lorentzian FQFT is a pseudo-multifunctor
$\FFF : \LBop_m\to \mathscr{Q}$ in the sense of Definition \ref{def:pseudomultifunctor}
to a suitable target pseudo-operad $\mathscr{Q}$. In the context of this paper,
we will choose the target pseudo-operad as follows:
Choosing any cocomplete closed symmetric monoidal category $\TT$, we 
form the symmetric monoidal category $\Alg_{\mathsf{uAs}}(\TT)$ of 
unital associative algebras in $\TT$ and consider its underlying operad, 
which we denote by the same symbol.
We define our target pseudo-operad \smash{$\mathscr{Q} = \iota\big(\Alg_{\mathsf{uAs}}(\TT)\big)$} by applying the
inclusion $2$-functor $\iota: \Op^{(2,1)}\to \PsOp^{\mathrm{fib}}$ from Construction~\ref{constr:iota}.
This choice is motivated by the choice of target for AQFTs (see Definition~\ref{def:AQFT})
and it will be crucial for proving our FQFT/AQFT equivalence theorem in Section \ref{sec:equivalence}.
In Section \ref{sec:twistedFFT} below, we will briefly discuss how a notion of
functorial QFTs assigning state spaces can be introduced in our 
context as twisted or relative functorial QFTs \cite{StolzTeichner,FT}.
Using the $2$-adjunction from Theorem \ref{theo:2adjunction}, we obtain an equivalent but simpler
description of such FQFTs in terms of ordinary multifunctors $\FFF : \tau(\LBop_m)\to\Alg_{\mathsf{uAs}}(\TT)$
from the operad given by applying to $\LBop_m$ the truncation $2$-functor $\tau:\PsOp^{\mathrm{fib}}\to  \Op^{(2,1)}$ 
from Construction \ref{constr:tau}. It is important to stress that the truncated operad
$\tau(\LBop_m)$ remembers much of the rich geometric structure of the pseudo-operad $\LBop_m$ 
from Section \ref{sec:LBord}. In particular, its objects
$(M,\Sigma)\in \tau(\LBop_m)$ and operations $[N,\und{\iota_0},\iota_1]: 
(\und{M_0},\und{\Sigma_0})\to (M_1,\Sigma_1)$ keep track of the 
collar regions around Cauchy surfaces and bordisms, which are crucial 
for a well-defined operadic composition. 
\begin{defi}\label{def:FFT}
Fix any cocomplete closed symmetric monoidal category $\TT$.
\begin{itemize}
\item[(a)] The \textit{category of $m$-dimensional globally hyperbolic Lorentzian FQFTs} is defined as the category
\begin{flalign}
\FQFT_m\,\coloneqq\, \Alg_{\tau(\LBop_m)}\big(\Alg_{\mathsf{uAs}}(\TT)\big)
\end{flalign}
of $\Alg_{\mathsf{uAs}}(\TT)$-valued algebras over the operad $\tau(\LBop_m)$
which is obtained by applying the truncation $2$-functor $\tau:\PsOp^{\mathrm{fib}}\to  \Op^{(2,1)}$ 
from Construction \ref{constr:tau} to the $m$-dimensional globally hyperbolic Lorentzian bordism pseudo-operad
$\LBop_m$ from Section \ref{sec:LBord}.
This means that an FQFT is a multifunctor 
$\FFF : \tau(\LBop_m) \to \Alg_{\mathsf{uAs}}(\TT)$ and that a morphism between FQFTs
is a multinatural transformation $\zeta : \FFF\Rightarrow\GGG : \tau(\LBop_m)\to \Alg_{\mathsf{uAs}}(\TT)$.

\item[(b)] An object $\FFF\in \FQFT_m$ is said to satisfy the \textit{time-slice axiom}
if the multifunctor $\FFF : \tau(\LBop_m) \to \Alg_{\mathsf{uAs}}(\TT)$ sends every Cauchy bordism
in $\tau(\LBop_m)$, i.e.\ every $1$-ary operation $[N,\iota_0,\iota_1]:(M_0,\Sigma_0)\to (M_1,\Sigma_1)$
such that $\iota_0$ is a Cauchy morphism, to an isomorphism
$\FFF([N,\iota_0,\iota_1]): \FFF(M_0,\Sigma_0) \stackrel{\cong}{\longrightarrow} \FFF(M_1,\Sigma_1) $
in $\Alg_{\mathsf{uAs}}(\TT)$. We denote by
\begin{flalign}\label{eqn:FQFTW}
\FQFT_m^W\,\subseteq\, \FQFT_m
\end{flalign}
the full subcategory of all FQFTs satisfying the time-slice axiom.
\end{itemize}
\end{defi}

\begin{rem}\label{rem:meaningofnto1}
We would like to highlight that the 
$(n\geq 2)$-ary bordisms $[N,\und{\iota_0},\iota_1] : (\und{M_0},\und{\Sigma_0})\to (M_1,\Sigma_1)$
in $\tau(\LBop_m)$ are of direct physical significance as they endow every globally hyperbolic
Lorentzian FQFT $\FFF\in \FQFT_m$ with a structure that implies
relativistic causality through an Eckmann-Hilton-type argument.
To illustrate this important point, let us consider concretely the case of a $2$-ary bordism to which $\FFF\in \FQFT_m$
assigns an $\Alg_{\mathsf{uAs}}(\TT)$-morphism
\begin{flalign}
\FFF([N,\und{\iota_0},\iota_1])\,:\, \FFF(M_{0_1},\Sigma_{0_1})\otimes \FFF(M_{0_2},\Sigma_{0_2})~
\longrightarrow~\FFF(M_1,\Sigma_1)~,~~
a_{0_1}\otimes a_{0_2}~\longmapsto~a_{0_1}\bullet a_{0_2}
\end{flalign}
from a tensor product algebra. Denoting by $\oone_{0_i}\in \FFF(M_{0_i},\Sigma_{0_i})$ the units 
of these algebras, we can use the $\Alg_{\mathsf{uAs}}(\TT)$-morphism property to 
write the element $a_{0_1}\bullet a_{0_2}\in \FFF(M_1,\Sigma_1)$ in two different ways, namely
\begin{subequations}
\begin{flalign}
\nn a_{0_1}\bullet a_{0_2}\,&=\,(a_{0_1}\,\oone_{0_1})\bullet(\oone_{0_2}\,a_{0_2})\,=\,
(a_{0_1}\bullet \oone_{0_2})\,(\oone_{0_1}\bullet a_{0_2})\\
\,&=\, \FFF([N,\iota_{0_1},\iota_1])(a_{0_1})~\FFF([N,\iota_{0_2},\iota_1])(a_{0_2})
\end{flalign}
and
\begin{flalign}
\nn a_{0_1}\bullet a_{0_2}\,&=\,(\oone_{0_1}\,a_{0_1})\bullet(a_{0_2}\,\oone_{0_2})\,=\,
(\oone_{0_1}\bullet a_{0_2})\,(a_{0_1}\bullet \oone_{0_2})\\
\,&=\, \FFF([N,\iota_{0_2},\iota_1])(a_{0_2})~\FFF([N,\iota_{0_1},\iota_1])(a_{0_1})\quad.
\end{flalign}
\end{subequations}
This simple Eckmann-Hilton-type argument implies the following notion of relativistic causality
for every globally hyperbolic Lorentzian FQFT $\FFF\in \FQFT_m$: 
Observables $a_{0_i}\in \FFF(M_{0_i},\Sigma_{0_i})$
localized on \textit{causally disjoint} partial Cauchy surfaces $\iota_{0_i}(\Sigma_{0_i})\subset N$
evolve along the bordisms $[N,\iota_{0_i},\iota_1]:(M_{0_i},\Sigma_{0_i})\to (M_1,\Sigma_1)$
to \textit{commuting} observables $\FFF([N,\iota_{0_i},\iota_1])(a_{0_i})\in \FFF(M_1,\Sigma_1)$ 
on the later-in-time full Cauchy surface $\iota_1(\Sigma_{1})\subset N$.
\end{rem}

\begin{rem}\label{rem:FFT}
We would like to note that only the
\textit{invertible} morphisms in $\FQFT_m$ admit a pseudo-operadic interpretation
in terms of our $2$-adjunction $\tau : \PsOp^{\mathrm{fib}}\rightleftarrows \Op^{(2,1)} : \iota$
from Theorem \ref{theo:2adjunction} because $\Op^{(2,1)}$ is the $(2,1)$-category of operads, multifunctors
and multinatural \textit{isomorphisms}. We believe that the non-invertible
morphisms in $\FQFT_m$ correspond to a pseudo-operadic variant of the (horizontal) 
pseudo-natural transformations from \cite{PseudoCats}.
However, generalizing the $2$-adjunction in Theorem \ref{theo:2adjunction} to 
this context goes beyond the scope of this paper, which is why we decided 
instead to add the non-invertible multinatural transformations in Definition \ref{def:FFT} by hand.
\end{rem}

For proving our FQFT/AQFT equivalence theorem in Section \ref{sec:equivalence}, 
we require an analogue for FQFTs of the additivity property for AQFTs from Definition \ref{def:AQFTadditivity}.
Since the key component of an object $(M,\Sigma)\in\tau(\LBop_m)$ is the Cauchy surface
$\Sigma\subset M$, and not its collar region $M\in\Loc_m$, it would make
sense to formulate an additivity property for FQFTs
in terms of pairs $(U,S)$ consisting of a relatively compact causally convex open subset
$U\subseteq M$ and a Cauchy surface $S\subset U$ which is contained in the given Cauchy surface
$S\subseteq \Sigma$. However, due to the appearance of the chronological past in
the condition \eqref{eqn:bordismsurfacesII} for operations of the bordism
pseudo-operad $\LBop_m$ in the non-Cauchy case, we cannot directly 
consider subsets $S\subseteq \Sigma$ of the given Cauchy surface,
but we have to move them to the chronological past of $\Sigma$. 
Since the size of the collar region is inessential, we can then also require
that $U\subseteq I^-_M(\Sigma) $ lies in the chronological past of $\Sigma$. 
This leads to the following
definition for an analogue of the category $\RC_M$ in the context of FQFT.
\begin{defi}\label{def:RCMSigma}
For each object $(M,\Sigma)\in\tau(\LBop_m)$, we denote by
$\RC_{(M,\Sigma)}$ the category whose objects are all pairs
$(U,S)$ consisting of a causally convex open subset
$U\subseteq I^-_M(\Sigma)$ which is relatively compact in $I^-_M(\Sigma)$,
i.e.\ $\mathrm{cl}_{I^-_M(\Sigma)}(U)\subseteq I^-_M(\Sigma)$ is compact,
and a Cauchy surface $S\subset U$.\footnote{Note that these properties 
imply that $U\subseteq M$ is also relatively compact in $M$ and that
the condition $\cl_M(S)= \cl_{I^-_M(\Sigma)}(S) \subset I^-_M(\Sigma)$ 
from \eqref{eqn:bordismsurfacesII} holds true.}
Given two objects $(U,S)$ and $(U^\prime,S^\prime)$,
there exists a unique morphism $(U,S)\to (U^\prime,S^\prime)$ in $\RC_{(M,\Sigma)}$
if $U\subseteq U^\prime$ and the relevant case of the conditions in 
\eqref{eqn:bordismsurfacesI} and \eqref{eqn:bordismsurfacesII} holds true.
This means that if $U\subseteq U^\prime$ is Cauchy we require
that $S \subset J^-_{U^\prime}(S^\prime)$ and otherwise we require that
$\mathrm{cl}_{U^\prime}(S)\subset I^-_{U^\prime}(S^\prime)$.
\end{defi}

This category has been designed such that there exists
a multifunctor $\RC_{(M,\Sigma)}\to \tau(\LBop_m)$ which
assigns to each object $(U,S)\in \RC_{(M,\Sigma)}$ the object $\big(U,S\big)\in \tau(\LBop_m)$
and to each morphism $(U,S)\to (U^\prime,S^\prime)$ in $\RC_{(M,\Sigma)}$ the $1$-ary operation 
$\big[U^\prime,\iota_{U}^{U^\prime},\id_{U^\prime}\big] : 
(U,S)\to (U^\prime,S^\prime)$ in $\tau(\LBop_m)$ which is defined by the zig-zags (see \eqref{eqn:bordismzigzag})
\begin{flalign}\label{eqn:zigzagDM}
\xymatrix{
U & \ar[l]_-{=} U  \ar[r]^-{\iota_{U}^{U^\prime}} & 
U^\prime  &\ar[l]_-{\id_{U^\prime}}U^\prime\ar[r]^-{=} & U^\prime
}\quad,
\end{flalign}
where we recall that  $\iota_{U}^{U^\prime}$ denotes the inclusion
$\Loc_m$-morphism. Note that this operation is well-defined
as a consequence of Definition~\ref{def:RCMSigma} and that 
functoriality follows from the definition  \eqref{eqn:bordismcomposition} of
horizontal composition in $\LBop_m$.
This implies that each $\FFF\in \FQFT_m$ can be restricted along the multifunctor
$\RC_{(M,\Sigma)}\to \tau(\LBop_m)$ to a functor
$\FFF\vert :\RC_{(M,\Sigma)}\to \Alg_{\mathsf{uAs}}(\TT) $
on the category in Definition \ref{def:RCMSigma}.
\begin{defi}\label{def:FFTadditivity}
An object $\FFF \in \FQFT_m$ is called \textit{additive}
if the canonical $\Alg_{\mathsf{uAs}}(\TT)$-morphism
\begin{flalign}
\colim\Big(\FFF\vert : \RC_{(M,\Sigma)}\to \Alg_{\mathsf{uAs}}(\TT)\Big)~\stackrel{\cong}{\longrightarrow}~\FFF(M,\Sigma)
\end{flalign}
is an isomorphism in $\Alg_{\mathsf{uAs}}(\TT)$, for all 
$(M,\Sigma)\in\tau(\LBop_m)$. We denote by
\begin{subequations}\label{eqn:FQFTadd}
\begin{flalign}
\FQFT_m^{\add}\,\subseteq \,\FQFT_m
\end{flalign}
the full subcategory of all additive FQFTs and by
\begin{flalign}
\FQFT_m^{W,\add}\,\subseteq \, \FQFT_m
\end{flalign}
\end{subequations}
the full subcategory of all additive FQFTs which also satisfy the time-slice axiom.
\end{defi}

\begin{rem}
We show in Proposition \ref{prop:RCMSigmafiltered} that the category 
$\RC_{(M,\Sigma)}$ from Definition \ref{def:RCMSigma} is filtered, for all 
$(M,\Sigma)\in\tau(\LBop_m)$.
Since the forgetful functor $\Alg_{\mathsf{uAs}}(\TT)\to \TT$ preserves and reflects
filtered colimits, one can deduce additivity by verifying the 
simpler condition that 
\begin{flalign}
\colim\Big(\FFF\vert : \RC_{(M,\Sigma)}\to \TT\Big)~\stackrel{\cong}{\longrightarrow}~\FFF(M,\Sigma)
\end{flalign}
is an isomorphism in $\TT$, for all 
$(M,\Sigma)\in\tau(\LBop_m)$.
\end{rem}


\section{\label{sec:equivalence}The equivalence theorem}
The aim of this section is to prove an 
equivalence theorem between AQFTs and FQFTs in our globally hyperbolic Lorentzian context. 
More specifically, for any spacetime 
dimension $m\in\bbN$, we will exhibit an equivalence 
\begin{flalign}
\AQFT_m^{W,\add}\,\simeq\, \FQFT_m^{W,\add}
\end{flalign}
between the category $\AQFT_m^{W,\add}$ of additive AQFTs satisfying
the time-slice axiom (see Definitions \ref{def:AQFT} and \ref{def:AQFTadditivity})
and the category $\FQFT_m^{W,\add}$ of additive FQFTs satisfying the time-slice axiom
(see Definitions \ref{def:FFT} and \ref{def:FFTadditivity}). We would like 
to emphasize that our proof relies heavily on both the time-slice axiom and the additivity property, 
and so we do not expect that either of these assumptions can be removed from our hypotheses.
It is interesting to note that the equivalence theorem between AQFTs and time-orderable
prefactorization algebras (tPFAs) in \cite{FAvsAQFT} uses the same time-slice and additivity assumptions,
which indicates that both are crucial in order to relate AQFT to other axiomatic frameworks
for quantum field theory.
\sk

Let us start with presenting a functorial construction from AQFTs to FQFTs,
upgrading the one in \cite[Construction 4.9]{FFTvAQFT} to our operadic context.
This construction crucially relies on the time-slice axiom for AQFTs, 
but it does not require the additivity property.
\begin{constr}\label{constr:AQFTtoFQFT}
We will define a functor
\begin{flalign}\label{eqn:bbFfunctor}
\bbF\,:\,\AQFT_m^{W}~\longrightarrow~\FQFT_m
\end{flalign}
from the category of AQFTs satisfying the time-slice axiom (see Definition \ref{def:AQFT})
to the category of FQFTs (see Definition \ref{def:FFT}). 
\sk

For every object $\AAA\in\AQFT_m^W$,
i.e.\ a multifunctor $\AAA : \P_{\Loc_m^\perp}\to \Alg_{\mathsf{uAs}}(\TT)$
sending Cauchy morphisms in $\P_{\Loc_m^\perp}$ to isomorphisms in $\Alg_{\mathsf{uAs}}(\TT)$,
we have to define an object $\bbF_{\AAA}\in \FQFT_m$, i.e.\ a multifunctor
\begin{subequations}\label{eqn:bbFAAA}
\begin{flalign}\label{eqn:bbFAAA1}
\bbF_{\AAA} \,:\, \tau(\LBop_m)~\longrightarrow~ \Alg_{\mathsf{uAs}}(\TT)\quad.
\end{flalign}
To an object $(M,\Sigma)\in \tau(\LBop_m)$, we assign the algebra
\begin{flalign}\label{eqn:bbFAAA2}
\bbF_{\AAA}(M,\Sigma)\,\coloneqq\, \AAA(M)\,\in\,\Alg_{\mathsf{uAs}}(\TT)
\end{flalign}
obtained by forgetting the Cauchy surface $\Sigma\subset M$. To an $n$-ary operation
$[N,\und{\iota_0},\iota_1] : (\und{M_0},\und{\Sigma_0})\to (M_1,\Sigma_1)$
in $\tau(\LBop_m)$, which we recall is represented by zig-zags \eqref{eqn:bordismzigzag}
of operations in the operad $\P_{\Loc_m^\perp}$, we assign the $\Alg_{\mathsf{uAs}}(\TT)$-morphism
defined by
\begin{flalign}\label{eqn:bbFAAA3}
\begin{gathered}
\xymatrix{
\ar@{=}[d]\bbF_{\AAA}(\und{M_0},\und{\Sigma_0})\ar@{-->}[rrrr]^-{\bbF_{\AAA}([N,\und{\iota_0},\iota_1] )} 
~&~ ~&~ ~&~ ~&~ \bbF_{\AAA}(M_1,\Sigma_1)\ar@{=}[d]\\
\AAA(\und{M_0})~&~ \ar[l]_-{\cong} \AAA(\und{V_0}) \ar[r]_-{\AAA(\und{\iota_0})}~&~ \AAA(N) ~&~
\ar[l]_-{\cong}^{\AAA(\iota_1)}\AAA(V_1) \ar[r]^-{\cong}~&~\AAA(M_1)
}
\end{gathered}\quad,
\end{flalign}
\end{subequations}
where $\AAA(\und{M_0})= \bigotimes_{i=1}^n\AAA(M_{0_i})\in \Alg_{\mathsf{uAs}}(\TT)$
and $\AAA(\und{V_0}) = \bigotimes_{i=1}^n\AAA(V_{0_i})\in \Alg_{\mathsf{uAs}}(\TT)$
denote the tensor product algebras, and $\Alg_{\mathsf{uAs}}(\TT)$-isomorphisms 
resulting from the time-slice axiom for $\AAA\in\AQFT_m^W$ are indicated by $\cong$. 
By a similar argument as in \cite[Construction 4.9]{FFTvAQFT},
one easily verifies that \eqref{eqn:bbFAAA3} does not depend
on the choice of representative $(N,\und{\iota_0},\iota_1)$ for the equivalence class 
$[N,\und{\iota_0},\iota_1] : (\und{M_0},\und{\Sigma_0})\to (M_1,\Sigma_1)$.
To check that the assignment \eqref{eqn:bbFAAA} defines a multifunctor,
we observe that it clearly preserves the identities \eqref{eqn:identitybordism} 
and it is equivariant with respect to the permutation actions \eqref{eqn:permutationbordisms}.
The preservation of operadic compositions is shown by applying $\AAA\in \AQFT_m^W$
to the commutative diagram \eqref{eqn:bordismcomposition} of operations in the operad
$\P_{\Loc_m^\perp}$.
\sk

For every $\AQFT_m^W$-morphism $\zeta : \AAA\Rightarrow\BBB:\P_{\Loc_m^\perp}\to  \Alg_{\mathsf{uAs}}(\TT)$,
i.e.\ a multinatural transformation with components $\zeta_M : \AAA(M)\to \BBB(M)$, for all
$M\in \P_{\Loc_m^\perp}$, we have to define an $\FQFT_m$-morphism $\bbF_{\zeta} : 
\bbF_{\AAA}\Rightarrow\bbF_{\BBB}:\tau(\LBop_m)\to  \Alg_{\mathsf{uAs}}(\TT)$, 
i.e.\ a multinatural transformation with components $(\bbF_\zeta)_{(M,\Sigma)}: \bbF_\AAA(M,\Sigma)
\to \bbF_{\BBB}(M,\Sigma)$, for all $(M,\Sigma)\in \tau(\LBop_m)$. We define
\begin{flalign}\label{eqn:bbFzeta}
\begin{gathered}
\xymatrix{
\ar@{=}[d]\bbF_{\AAA}(M,\Sigma) \ar@{-->}[rr]^-{(\bbF_\zeta)_{(M,\Sigma)}}~&~~&~\bbF_{\BBB}(M,\Sigma)\ar@{=}[d]\\
\AAA(M)\ar[rr]_-{\zeta_M}~&~ ~&~\BBB(M)
}
\end{gathered}\quad,
\end{flalign}
for all $(M,\Sigma)\in \tau(\LBop_m)$, to be independent of the Cauchy surface $\Sigma\subset M$.
Multinaturality of these components, for all 
$n$-ary operations $[N,\und{\iota_0},\iota_1] : (\und{M_0},\und{\Sigma_0})\to (M_1,\Sigma_1)$
in $\tau(\LBop_m)$, is easily verified by combining \eqref{eqn:bbFAAA3} and \eqref{eqn:bbFzeta}.
\sk

Functoriality as in \eqref{eqn:bbFfunctor} of the above constructions is evident.
\end{constr}

\begin{lem}\label{lem:bbFrestriction}
The functor \eqref{eqn:bbFfunctor} factorizes through the full subcategory
$\FQFT_m^W\subseteq \FQFT_m$ of FQFTs satisfying the time-slice axiom from Definition \ref{def:FFT}.
Furthermore, it restricts to a functor
\begin{flalign}
\bbF\,:\, \AQFT_m^{W,\add}~\longrightarrow~\FQFT_m^{W,\add}
\end{flalign}
between the categories of additive theories satisfying the time-slice axiom 
from Definitions~\ref{def:AQFTadditivity} and~\ref{def:FFTadditivity}, respectively.
\end{lem} 
\begin{proof}
The first claim is obvious: For every Cauchy bordism $[N,\iota_0,\iota_1] :(M_0,\Sigma_0)\to(M_1,\Sigma_1)$ 
in $\tau(\LBop_m)$ both $\iota_0$ and $\iota_1$ are Cauchy morphisms in $\P_{\Loc_m^\perp}$,
hence all arrows in \eqref{eqn:bbFAAA3} are isomorphisms in $\Alg_{\mathsf{uAs}}(\TT)$.
\sk

Let us now prove the second claim. Given any additive $\AAA\in \AQFT_m^{W,\add}$,
we have to show that $\bbF_{\AAA}\in \FQFT_m^{W}$ defined in \eqref{eqn:bbFAAA}
satisfies the additivity property from Definition \ref{def:FFTadditivity}. 
Directly from the definitions in \eqref{eqn:bbFAAA}, one observes that 
the restricted functor $\bbF_{\AAA}\vert : \RC_{(M,\Sigma)}\to \Alg_{\mathsf{uAs}}(\TT)$
factorizes through the forgetful functor
\begin{flalign}\label{eqn:forgetfulTMP}
\RC_{(M,\Sigma)}~&\longrightarrow~\RC_{I_M^-(\Sigma)}\quad,\\
\nn (U,S)~&\longmapsto~U\quad,\\
\nn \big((U,S)\to(U^\prime,S^\prime) \big)~&\longmapsto~\big(U\subseteq U^\prime\big)
\end{flalign}
in the sense that we have a commutative diagram
\begin{flalign}
\begin{gathered}
\xymatrix{
\ar[dr]\RC_{(M,\Sigma)} \ar[rr]^-{\bbF_{\AAA}\vert}~&~ ~&~ \Alg_{\mathsf{uAs}}(\TT)\\
~&~ \RC_{I^-_M(\Sigma)}\ar[ru]_-{\AAA\vert}~&~
}
\end{gathered}
\end{flalign}
of categories and functors, for all objects $(M,\Sigma)\in\tau(\LBop_m)$. The fact that $\bbF_{\AAA}\vert$ 
is the pullback of the restricted AQFT $\AAA\vert$ along the forgetful functor induces a morphism
between the corresponding colimits which fits into the commutative diagram
\begin{flalign}\label{eqn:additivitydiagramTMP}
\begin{gathered}
\xymatrix{
\ar[d] \colim\Big(\bbF_{\AAA}\vert : \RC_{(M,\Sigma)}\to\Alg_{\mathsf{uAs}}(\TT)\Big) \ar[r]~&~
\bbF_{\AAA}(M,\Sigma) \,=\,\AAA(M) \\
\colim\Big(\AAA\vert : \RC_{I^-_M(\Sigma)}\to\Alg_{\mathsf{uAs}}(\TT)\Big) \ar[r]_-{\cong}~&~\AAA(I^-_M(\Sigma))
\ar[u]^-{\cong}_-{\AAA\big(\iota_{I^-_M(\Sigma)}^{M}\big)}
}
\end{gathered}
\end{flalign}
in $\Alg_{\mathsf{uAs}}(\TT)$. The bottom horizontal arrow in this diagram is an isomorphism
because $\AAA$ is, by hypothesis, additive and the right vertical arrow is an isomorphism
as a consequence of the time-slice axiom for $\AAA$. Using the result from Proposition \ref{prop:RCMSigmafiltered}
that the category $\RC_{(M,\Sigma)}$ is filtered, one easily checks that 
the forgetful functor \eqref{eqn:forgetfulTMP} is final, hence the left vertical arrow
in \eqref{eqn:additivitydiagramTMP} is an isomorphism too. This implies that the top horizontal
arrow is an isomorphism, for every $(M,\Sigma)\in\tau(\LBop_m)$, hence
$\bbF_{\AAA}\in \FQFT_m^{W,\add}$ is additive.
\end{proof}

We will now present a functorial construction from FQFTs to AQFTs,
which turns out to be more difficult
and, in contrast to Construction~\ref{constr:AQFTtoFQFT}, uses additivity. 
Our construction is inspired by, and non-trivially 
extends, the proof of essential surjectivity in \cite[Theorem 4.11]{FFTvAQFT}.
Let us start with some preparations which allow us to present our construction
more concisely. 
\begin{defi}\label{def:SigmaMcategory}
Given any object $M\in\P_{\Loc_m^\perp}$, we denote by
$\mathbf{\Sigma}_M$ the category whose objects are all Cauchy surfaces
$\Sigma\subset M$ and there exists a unique morphism $\Sigma\to \Sigma^\prime$
if and only if $\Sigma\subset J^-_M(\Sigma^\prime)$ lies in the causal past of $\Sigma^\prime$.
This category is filtered as a consequence of Lemma \ref{lem:surfaces}.
\end{defi}

Note that there exists a multifunctor $\mathbf{\Sigma}_M\to \tau(\LBop_m)$ which assigns
to $\Sigma\subset M$ the object $(M,\Sigma)\in\tau(\LBop_m)$ and to
a morphism $\Sigma\to \Sigma^\prime$ the equivalence class of bordisms
$[M,\id_M,\id_M] : (M,\Sigma)\to (M,\Sigma^\prime)$ represented by the identity zig-zags.
This allows us to restrict any $\FFF\in \FQFT_m$ to a functor
$\FFF\vert : \mathbf{\Sigma}_M\to \Alg_{\mathsf{uAs}}(\TT)$. Assigning to
$\FFF\in \FQFT_m$ an AQFT requires in particular the specification
of algebras $\bbA_{\FFF}(M)\in\Alg_{\mathsf{uAs}}(\TT)$, for all $M\in\P_{\Loc_m^\perp}$,
which do not depend on the choice of a Cauchy surface for $M$. A natural candidate 
is given by taking the colimit
\begin{subequations}\label{eqn:SigmaMcolimit}
\begin{flalign}\label{eqn:SigmaMcolimit1}
\bbA_{\FFF}(M)\,\coloneqq\, \colim\Big(\FFF\vert : \mathbf{\Sigma}_M\to \Alg_{\mathsf{uAs}}(\TT) \Big)
\end{flalign}
over the category $\mathbf{\Sigma}_M$ of Cauchy surfaces from Definition \ref{def:SigmaMcategory}.
It is shown in the proof of \cite[Theorem 4.11]{FFTvAQFT} that these
algebras can naturally be endowed with actions
of Cauchy morphisms $f:M\to N$ in $\P_{\Loc_m^\perp}$ since, given any Cauchy surface
$\Sigma\subset M$ of $M$, the image $f(\Sigma)\subset N$ is a Cauchy surface of $N$.
Explicitly, these $\Alg_{\mathsf{uAs}}(\TT)$-morphisms 
are defined via the universal property of the colimit by
\begin{flalign}\label{eqn:SigmaMcolimit2}
\begin{gathered}
\xymatrix@C=3em{
\bbA_{\FFF}(M) \ar@{-->}[rr]^-{\bbA_{\FFF}(f)}~&&~\bbA_{\FFF}(N)\\
\ar[u]^-{\iota_\Sigma}\FFF(M,\Sigma)\ar[rr]_-{\FFF([N,f,\id_N])}~&&~\FFF(N,f(\Sigma))\ar[u]_-{\iota_{f(\Sigma)}}
}
\end{gathered}\quad,
\end{flalign}
\end{subequations}
for all $\Sigma\in\mathbf{\Sigma}_M$.
\sk

Unfortunately, this construction does not extend to 
general $n$-ary operations $\und{f} : \und{M}\to N$
in the operad $\P_{\Loc_m^\perp}$. The reason being that, given any family of Cauchy surfaces
$\Sigma_i \subset M_i$ with $i=1,\dots,n$, in general there does not exist an extension
of the image $\bigcup_{i=1}^n f_i(\Sigma_i)\subset N$ to a Cauchy surface of $N$.
In order to circumvent these issues, we will
require $\FFF\in \FQFT_m^{\add}$ to be additive in the sense of Definition \ref{def:FFTadditivity}
and use the extension results from \cite[Proposition 3.6]{BernalSanchez2} for achronal \textit{compact}
subsets to Cauchy surfaces. 
\begin{constr}\label{constr:FQFTtoAQFT}
We will define a functor
\begin{flalign}\label{eqn:bbAfunctor}
\bbA\,:\,\FQFT_m^{\add}~\longrightarrow~\AQFT_m
\end{flalign}
from the category of additive FQFTs (see Definition \ref{def:FFTadditivity})
to the category of AQFTs (see Definition \ref{def:AQFT}). 
\sk

For every object $\FFF\in\FQFT_m^{\add}$, i.e.\ a multifunctor
$\FFF : \tau(\LBop_m)\to\Alg_{\mathsf{uAs}}(\TT)$ satisfying
the additivity property from Definition \ref{def:FFTadditivity},
we have to define an object $\bbA_{\FFF}\in \AQFT_m$, i.e.\ a multifunctor
\begin{subequations}\label{eqn:bbAFFF}
\begin{flalign}\label{eqn:bbAFFF1}
\bbA_{\FFF}\,:\,\P_{\Loc_m^\perp}~\longrightarrow~\Alg_{\mathsf{uAs}}(\TT)\quad.
\end{flalign}
To an object $M\in \P_{\Loc_m^\perp}$, we assign the algebra from \eqref{eqn:SigmaMcolimit1}, i.e.\
\begin{flalign}\label{eqn:bbAFFF2}
\bbA_{\FFF}(M)\,\coloneqq\, \colim\Big(\FFF\vert: \mathbf{\Sigma}_M \to \Alg_{\mathsf{uAs}}(\TT)\Big)\quad,
\end{flalign}
where $\mathbf{\Sigma}_M$ is the category of Cauchy surfaces from Definition \ref{def:SigmaMcategory}.
To an $n$-ary operation $\und{f} : \und{M}\to N$ in $\P_{\Loc_m^\perp}$, we assign 
the $\Alg_{\mathsf{uAs}}(\TT)$-morphism defined via the universal property of the colimit and additivity
of $\FFF$ (see Definition \ref{def:FFTadditivity}) by
\begin{flalign}\label{eqn:bbAFFF3}
\begin{gathered}
\xymatrix@C=3em{
\bbA_{\FFF}(\und{M}) \ar@{-->}[rr]^-{\bbA_{\FFF}(\und{f})}~&&~ \bbA_{\FFF}(N)\\
\ar[u]^-{\iota_{\und{\Sigma}}} \FFF(\und{M},\und{\Sigma})\ar@{-->}[rr]~&&~ \FFF(N,\Sigma^\prime)\ar[u]_-{\iota_{\Sigma^\prime}}\\
\ar[u]^-{\FFF([\und{M},\iota_{\und{U}}^{\und{M}},\id_{\und{M}}])}
\FFF(\und{U},\und{S})\ar[rru]_-{~~~~~~\FFF([N,\und{f}\vert,\id_N])}~&&~
}
\end{gathered}\quad,
\end{flalign}
\end{subequations}
for all $\und{\Sigma}\in\mathbf{\Sigma}_{\und{M}} = \prod_{i=1}^n\mathbf{\Sigma}_{M_i}$
and $(\und{U},\und{S})\in \RC_{(\und{M},\und{\Sigma})} = \prod_{i=1}^n\RC_{(M_i,\Sigma_i)}$,
where $\und{f}\vert : \und{U}\to N$ denotes the restriction of $\und{f} : \und{M}\to N$
and $\Sigma^\prime\in \mathbf{\Sigma}_N$ is any choice of Cauchy surface of $N$
such that $\mathrm{cl}_N\big(\bigcup_{i=1}^n f_i(S_i)\big)\subset I^-_N(\Sigma^\prime)$.
(This implies that the relevant condition in \eqref{eqn:bordismsurfacesI} and \eqref{eqn:bordismsurfacesII} 
holds true, hence $[N,\und{f}\vert,\id_N]: (\und{U},\und{S})\to (N,\Sigma^\prime)$ defines
an $n$-ary operation in $\tau(\LBop_m)$.)
The existence of such a Cauchy surface is guaranteed by the fact that the subset
$\mathrm{cl}_N\big(\bigcup_{i=1}^n f_i(S_i)\big)\subset N$ is achronal and 
compact, as a consequence of the relative compactness of $U_i\subseteq M_i$ for all $i$,
and therefore extends by \cite[Proposition 3.6]{BernalSanchez2} to a Cauchy surface
of $N$. We can then choose for $\Sigma^\prime\subset N$ any (chronologically) 
later Cauchy surface, which exists by Lemma \ref{lem:surfaces}. Note that the definition in \eqref{eqn:bbAFFF3}
does not depend on the choice of such a $\Sigma^\prime$ since the category $\mathbf{\Sigma}_N$ is filtered.
Additionally, we note that the algebras 
$\bbA_\FFF(\und{M}) = \bigotimes_{i = 1}^n \bbA_\FFF(M_i)\in \Alg_{\mathsf{uAs}}(\TT)$ 
and $\FFF(\und{M},\und{\Sigma}) = \bigotimes_{i=1}^n \FFF(M_i,\Sigma_i)\in \Alg_{\mathsf{uAs}}(\TT)$ 
in \eqref{eqn:bbAFFF3} are tensor products of 
colimits over, respectively, the filtered categories $\mathbf{\Sigma}_{M_i}$
and $\RC_{(M_i,\Sigma_i)}$, for $i=1,\dots,n$.
Since the forgetful functor $\Alg_{\mathsf{uAs}}(\TT) \to \TT$ preserves and reflects 
filtered colimits, and the monoidal product in $\TT$ preserves colimits (since $\TT$ is closed monoidal), 
it follows that $\bbA_\FFF(\und{M})$ can be expressed as a colimit over the product category 
$\mathbf{\Sigma}_{\und{M}} = \prod_{i=1}^n\mathbf{\Sigma}_{M_i}$
and $\FFF(\und{M},\und{\Sigma})$ as a colimit over the product category 
$\RC_{(\und{M},\und{\Sigma})} = \prod_{i=1}^n\RC_{(M_i,\Sigma_i)}$.
This implies that the components $\iota_{\und{\Sigma}}$ and
$\iota_{(\und{U}, \und{S})}$ in \eqref{eqn:bbAFFF3} are indeed labeled by the product category.
\sk

Multifunctoriality of the assignment \eqref{eqn:bbAFFF} is straightforward 
to check. The key observation is that, for every two composable operations
$\und{\und{g}}: \und{\und{L}}\to \und{M}$ and $\und{f}: \und{M} \to N$
in $\P_{\Loc_m^{\perp}}$, one obtains directly from \eqref{eqn:bordismcomposition} 
the composition identity
\begin{flalign}
\begin{gathered}
\xymatrix@C=3em{
\ar[drr]_-{[N,\und{f}\,\und{\und{g}},\id_{N}]~~~~} 
(\und{\und{L}},\und{\und{\Sigma_L}}) \ar[rr]^-{[\und{M},\und{\und{g}},\id_{\und{M}}]}
~&&~ (\und{M},\und{\Sigma_M}) \ar[d]^-{[N,\und{f},\id_N]}\\
~&&~ (N,\Sigma_N)
}
\end{gathered}
\end{flalign}
in $\tau(\LBop_m)$ for the associated equivalence classes of bordisms as in \eqref{eqn:bbAFFF3}.
\sk

For every $\FQFT_m^\add$-morphism $\zeta: \FFF\Rightarrow \GGG: \tau(\LBop_m)\to \Alg_{\mathsf{uAs}}(\TT)$,
i.e.\ a multinatural transformation with components $\zeta_{(M,\Sigma)} :\FFF(M,\Sigma)\to \GGG(M,\Sigma)$,
for all $(M,\Sigma)\in\tau(\LBop_m)$, we have to define an $\AQFT_m$-morphism
$\bbA_{\zeta} : \bbA_{\FFF}\Rightarrow\bbA_{\GGG} : \P_{\Loc_m^\perp} \to\Alg_{\mathsf{uAs}}(\TT)$,
i.e.\ a multinatural transformation with components $(\bbA_{\zeta})_M : \bbA_{\FFF}(M)\to \bbA_{\GGG}(M)$,
for all $M\in\P_{\Loc_m^\perp}$. Restricting $\zeta$ to a natural transformation
$\zeta\vert: \FFF\vert\Rightarrow\GGG\vert : \mathbf{\Sigma}_M\to \Alg_{\mathsf{uAs}}(\TT)$, we define
\begin{flalign}\label{eqn:bbAzeta}
\begin{gathered}
\xymatrix{
\bbA_\FFF(M) \ar@{-->}[rr]^-{(\bbA_\zeta)_M}~&&~\bbA_\GGG(M)\\
\ar@{=}[u]\colim(\FFF\vert)\ar[rr]_-{\colim(\zeta\vert)}~&&~ \colim(\GGG\vert)\ar@{=}[u]
}
\end{gathered}\quad,
\end{flalign}
for all $M\in\P_{\Loc_m^\perp}$. Multinaturality of these components,
for all $n$-ary operations $\und{f} : \und{M}\to N$
in $\P_{\Loc_m^\perp}$, is easily verified by combining \eqref{eqn:bbAFFF3} and \eqref{eqn:bbAzeta}.
\sk

Functoriality as in \eqref{eqn:bbAfunctor} of the above constructions is evident.
\end{constr}

\begin{rem}\label{rem:bbAsimplification}
We would like to emphasize that, for a general $n$-ary operation 
$\und{f}:\und{M}\to N$ in $\P_{\Loc_m^\perp}$,
the middle horizontal dashed arrow in \eqref{eqn:bbAFFF3} does \textit{not}
admit a description in terms of an $n$-ary 
bordism $[N,\und{f},\id_N]:(\und{M},\und{\Sigma})\to (N,\Sigma^\prime)$
in $\tau(\LBop_m)$ since there might not exist a Cauchy surface $\Sigma^\prime\subset N$
such that the relevant condition in \eqref{eqn:bordismsurfacesI} and \eqref{eqn:bordismsurfacesII}
is satisfied. However, there exist special cases in which such 
Cauchy surfaces $\Sigma^\prime\subset N$ can be found.
This includes the case where the image of $\und{f}:\und{M}\to N$ is
relatively compact, which we have leveraged in our Construction \ref{constr:FQFTtoAQFT},
and also the case where $f:M\to N$ is a Cauchy morphism. In those situations,
the bottom triangle in \eqref{eqn:bbAFFF3}, which uses additivity, is not required,
and diagram~\eqref{eqn:bbAFFF3} simplifies to
\begin{flalign}\label{eqn:bbAsimplification}
\begin{gathered}
\xymatrix@C=3em{
\bbA_{\FFF}(\und{M}) \ar@{-->}[rr]^-{\bbA_{\FFF}(\und{f})}~&&~ \bbA_{\FFF}(N)\\
\ar[u]^-{\iota_{\und{\Sigma}}} \FFF(\und{M},\und{\Sigma})\ar[rr]_-{\FFF([N,\und{f},\id_N])}~&&~ 
\FFF(N,\Sigma^\prime)\ar[u]_-{\iota_{\Sigma^\prime}}
}
\end{gathered}\quad,
\end{flalign}
where $\Sigma^\prime\subset N$ is any choice of Cauchy surface such that 
the relevant condition in \eqref{eqn:bordismsurfacesI} and \eqref{eqn:bordismsurfacesII} is satisfied.
In particular, this implies that our general Construction \ref{constr:FQFTtoAQFT}
specializes to our previous treatment of Cauchy morphisms in \cite[Theorem 4.11]{FFTvAQFT},
see also \eqref{eqn:SigmaMcolimit2} above.
\end{rem}

\begin{lem}\label{lem:bbArestriction}
The functor \eqref{eqn:bbAfunctor} factorizes through the full subcategory
$\AQFT_m^{\add}\subseteq \AQFT_m$ of additive AQFTs from Definition \ref{def:AQFTadditivity}.
Furthermore, it restricts to a functor
\begin{flalign}
\bbA\,:\, \FQFT_m^{W,\add}~\longrightarrow~\AQFT_m^{W,\add}
\end{flalign}
between the categories of additive theories satisfying the time-slice axiom
from Definitions~\ref{def:AQFTadditivity} and~\ref{def:FFTadditivity}, respectively.
\end{lem}
\begin{proof}
To prove the initial claim, we first observe that, by
combining \eqref{eqn:bbAFFF2} and additivity of $\FFF$ (see Definition \ref{def:FFTadditivity}),
one can write
\begin{flalign}
\bbA_\FFF(M)\,\cong\, \colim_{\Sigma\in\mathbf{\Sigma}_M}^{}\colim_{(U,S)\in\RC_{(M,\Sigma)}}^{}\FFF(U,S)
\end{flalign}
as an iterated colimit. This iterated colimit can be identified canonically with a single colimit
\begin{flalign}\label{eqn:colimitGrothendieck}
\bbA_\FFF(M)\,\cong\, \colim\bigg(\FFF\vert: \int_{\mathbf{\Sigma}_M}\RC \to \Alg_{\mathsf{uAs}}(\TT)\bigg) 
\end{flalign}
over the Grothendieck construction
$\int_{\mathbf{\Sigma}_M}\!\RC$ of the functor $\RC : \mathbf{\Sigma}_M\to \Cat$
which assigns to $\Sigma\in \mathbf{\Sigma}_M$ the category $\RC_{(M,\Sigma)}$
from Definition \ref{def:RCMSigma} and to a $\mathbf{\Sigma}_M$-morphism
$\Sigma\to \Sigma^\prime$ the obvious inclusion functor $\RC_{(M,\Sigma)}\to \RC_{(M,\Sigma^\prime)}$.
Explicitly, the relevant Grothendieck construction is given by the category
\begin{flalign}\label{eqn:Grothendieck}
\int_{\mathbf{\Sigma}_M}\!\RC\,=\,\begin{cases}
\text{$\mathsf{Obj}$:} & ~(\Sigma,(U,S)) \text{ with }\Sigma\in\mathbf{\Sigma}_M\text{ and }(U,S)\in\RC_{(M,\Sigma)}
\\[4pt]
\text{$\mathsf{Mor}$:} & ~\exists! : (\Sigma,(U,S))\to (\Sigma^\prime,(U^\prime,S^\prime)) \text{ iff }
\Sigma\to \Sigma^\prime \text{ in }\mathbf{\Sigma}_M \\
&~ \text{ and }(U,S)\to (U^\prime,S^\prime) \text{ in }\RC_{(M,\Sigma^\prime)}
\end{cases}
\end{flalign}
and the functor $\FFF\vert : \int_{\mathbf{\Sigma}_M}\!\RC \to\Alg_{\mathsf{uAs}}(\TT)$
acts on objects as $(\Sigma,(U,S))\mapsto \FFF(U,S)$ and on morphisms as 
$\big((\Sigma,(U,S))\to (\Sigma^\prime,(U^\prime,S^\prime))\big) \mapsto \big(
\FFF([U^\prime,\iota_U^{U^\prime},\id_{U^\prime}]):\FFF(U,S)\to \FFF(U^\prime,S^\prime)\big)$.
Since the latter functor is insensitive to the Cauchy surfaces $\Sigma\subset M$,
we can work out a simplified model for the colimit in \eqref{eqn:colimitGrothendieck}.
For this we introduce the category
\begin{flalign}
\QQ_M\,\coloneqq\, \begin{cases}
\text{$\mathsf{Obj}$:} & ~(U,S) \text{ with }U\in\RC_M\text{ and }S\in\mathbf{\Sigma}_U
\\[4pt]
\text{$\mathsf{Mor}$:} & ~\exists! : (U,S)\to (U^\prime,S^\prime) \text{ iff }
U\to U^\prime \text{ in }\RC_M \\
&~ \text{ and } S \subset J^-_{U^\prime}(S^\prime) \text{ for $U\to U^\prime$ Cauchy or } 
\mathrm{cl}_{U^\prime}(S)\subset I^-_{U^\prime}(S^\prime) \text{ else }
\end{cases}\quad,
\end{flalign}
which by design receives the forgetful functor $\int_{\mathbf{\Sigma}_M}\!\RC\to \QQ_M\,,~(\Sigma,(U,S))\mapsto (U,S)$.
One easily shows that this forgetful functor is final (see Lemma~\ref{lem:finality})
and that $\FFF\vert : \int_{\mathbf{\Sigma}_M}\!\RC \to\Alg_{\mathsf{uAs}}(\TT)$ 
factorizes through $\int_{\mathbf{\Sigma}_M}\!\RC\to \QQ_M$.
Hence, we obtain yet another isomorphic description
\begin{flalign}
\bbA_\FFF(M)\,\cong\, \colim\Big(\FFF\vert: \QQ_M \to \Alg_{\mathsf{uAs}}(\TT)\Big) 
\end{flalign}
of the algebra $\bbA_\FFF(M)$ in  \eqref{eqn:bbAFFF2}.
\sk

The advantage of this isomorphic description is that 
there exists a functor $\QQ : \RC_M \to \Cat$ which assigns
to $V\in \RC_M$ the category $\QQ_V$ and to an $\RC_M$-morphism 
$V\to V^\prime$ the obvious inclusion functor
$\QQ_V\to \QQ_{V^\prime}$. Denoting by $\int_{\RC_M}\QQ$ the associated
Grothendieck construction, we can identify the colimit in the
additivity condition (see Definition \ref{def:AQFTadditivity})
for $\bbA_{\FFF}$ with
\begin{flalign}
\colim\Big(\bbA_{\FFF}\vert: \RC_M\to \Alg_{\mathsf{uAs}}(\TT)\Big)\,\cong\,
\colim\bigg(\bbA_{\FFF}\vert: \int_{\RC_M}\QQ\to \Alg_{\mathsf{uAs}}(\TT)\bigg)\quad.
\end{flalign}
Additivity then follows from the fact that the functor
$\int_{\RC_M}\QQ\to \QQ_M\,,~(V,(U,S))\mapsto (U,S)$ is final,
for all $M\in\P_{\Loc_m^\perp}$.
\sk

To prove the second claim, we have to show that,
for any $\FFF\in\FQFT_m^{W,\add}$, the associated multifunctor 
$\bbA_\FFF : \P_{\Loc_m^\perp} \to \Alg_{\mathsf{uAs}}(\TT)$ from
Construction \ref{constr:FQFTtoAQFT} sends every Cauchy morphism
$f:M\to N$ in $\P_{\Loc_m^\perp}$ to an isomorphism 
$\bbA_\FFF(f) : \bbA_\FFF(M)\to \bbA_\FFF(N)$ in 
$\Alg_{\mathsf{uAs}}(\TT)$. Using the observation in Remark 
\ref{rem:bbAsimplification}, this follows immediately from
the fact that all solid arrows in \eqref{eqn:SigmaMcolimit2}
are isomorphisms as a consequence of the time-slice axiom
for $\FFF\in \FQFT_{m}^{W,\add}$, and hence 
$\bbA_\FFF(f) : \bbA_\FFF(M)\to \bbA_\FFF(N)$
is an isomorphism too.
\end{proof}

The main result of this section is the following equivalence theorem.
\begin{theo}\label{theo:equivalence}
For every spacetime dimension $m\in\bbN$,
the two functors $\bbF:\AQFT_m^{W,\add}\to \FQFT_m^{W,\add}$ and $\bbA:
\FQFT_m^{W,\add}\to\AQFT_m^{W,\add}$ from Lemmas \ref{lem:bbFrestriction} and \ref{lem:bbArestriction}
are quasi-inverse to each other. Hence, they exhibit an equivalence
\begin{flalign}
\AQFT_m^{W,\add} \, \simeq \,\FQFT_m^{W,\add}
\end{flalign}
between the category $\AQFT_m^{W,\add}$ of additive AQFTs satisfying
the time-slice axiom (see Definitions \ref{def:AQFT} and \ref{def:AQFTadditivity})
and the category $\FQFT_m^{W,\add}$ of additive FQFTs satisfying the time-slice axiom
(see Definitions \ref{def:FFT} and \ref{def:FFTadditivity}).
\end{theo}
\begin{proof}
We consider first the composition $\bbA\circ \bbF : \AQFT_m^{W,\add}\to \AQFT_m^{W,\add}$.
Given any object $\AAA\in  \AQFT_m^{W,\add}$, the multifunctor
$(\bbA\circ \bbF)_{\AAA} : \P_{\Loc_m^\perp}\to \Alg_{\mathsf{uAs}}(\TT)$ 
can be described explicitly by using the definitions in 
Constructions \ref{constr:AQFTtoFQFT} and \ref{constr:FQFTtoAQFT}.
To an object $M\in\P_{\Loc_m^\perp}$, it assigns the algebra
\begin{subequations}
\begin{flalign}
(\bbA\circ \bbF)_{\AAA}(M) \,=\,
\colim\Big(\bbF_{\AAA}\vert : \mathbf{\Sigma}_M\to \Alg_{\mathsf{uAs}}(\TT)\Big)\,=\,
\AAA(M)\quad,
\end{flalign}
where in the last step we used that $\bbF_{\AAA}\vert$ is the constant functor
with value $\bbF_{\AAA}(M,\Sigma) = \AAA(M)$, for all $\Sigma\in \mathbf{\Sigma}_M$, 
hence the filtered colimit is given by this value. To an $n$-ary operation
$\und{f} : \und{M}\to N$ in $\P_{\Loc_m^\perp}$, this multifunctor assigns
\begin{flalign}
\begin{gathered}
\xymatrix@C=3em{
\ar@{=}[d](\bbA\circ \bbF)_{\AAA}(\und{M})\ar[rr]^-{(\bbA\circ \bbF)_{\AAA}(\und{f})}
~&&~(\bbA\circ \bbF)_{\AAA}(N)\ar@{=}[d]\\
\AAA(\und{M})\ar[rr]_-{\AAA(\und{f})}~&&~\AAA(N)
}
\end{gathered}\quad.
\end{flalign}
\end{subequations}
Hence, we find that $(\bbA\circ \bbF)_\AAA = \AAA$, for all objects $\AAA\in  \AQFT_m^{W,\add}$.
Given any morphism $\zeta: \AAA\Rightarrow\BBB$ in $\AQFT_m^{W,\add}$,
we compute the components of the multinatural transformation $(\bbA\circ \bbF)_\zeta$ 
by using Constructions \ref{constr:AQFTtoFQFT} and \ref{constr:FQFTtoAQFT} and find
\begin{flalign}
\big((\bbA\circ \bbF)_\zeta\big)_M \,=\, \colim\Big(\bbF_{\zeta}\vert :
\bbF_{\AAA}\vert\Rightarrow\bbF_{\BBB}\vert:\mathbf{\Sigma}_M\to \Alg_{\mathsf{uAs}}(\TT)\Big)
\,=\,\zeta_M\quad,
\end{flalign}
where in the last step we use that $(\bbF_{\zeta})_{(M,\Sigma)} = \zeta_M$ is constant
in $\Sigma\in \mathbf{\Sigma}_M$. This shows that 
$(\bbA\circ \bbF)_\zeta =\zeta$, for all morphisms $\zeta: \AAA\Rightarrow\BBB$ in $\AQFT_m^{W,\add}$.
Summing up, we obtain that the composition $\bbA\circ \bbF = \id$ is equal to the identity functor
on $\AQFT_m^{W,\add}$.
\sk

Consider now the composition $\bbF\circ\bbA : \FQFT_m^{W,\add}\to\FQFT_m^{W,\add}$. 
Given any object $\FFF\in \FQFT_m^{W,\add}$, the multifunctor
$(\bbF\circ\bbA)_\FFF : \tau(\LBop_m)\to\Alg_{\mathsf{uAs}}(\TT)$ assigns to
an object $(M,\Sigma)\in\tau(\LBop_m)$ the colimit
\begin{subequations}
\begin{flalign}
(\bbF\circ\bbA)_\FFF(M,\Sigma)\,=\, \bbA_\FFF(M)\,=\, 
\colim\Big(\FFF\vert:\mathbf{\Sigma}_M\to\Alg_{\mathsf{uAs}}(\TT)\Big)\quad.
\end{flalign}
As a consequence of the time-slice axiom for $\FFF$, 
the canonical inclusion morphism $\iota_\Sigma : \FFF(M,\Sigma)\to (\bbF\circ\bbA)_\FFF(M,\Sigma)$
associated with the given Cauchy surface $\Sigma\subset M$ is an isomorphism.
Hence, one can define morphisms out of these colimits by specifying only a single component.
Using the observation from Remark \ref{rem:bbAsimplification}, in particular \eqref{eqn:bbAsimplification},
we then find that the action of the multifunctor $(\bbF\circ\bbA)_\FFF$ on an 
$n$-ary operation $[N,\und{\iota_0},\iota_1]:(\und{M_0},\und{\Sigma_0})\to (M_1,\Sigma_1)$
in $\tau(\LBop_m)$ is given by the $\Alg_{\mathsf{uAs}}(\TT)$-morphism
\begin{flalign}\label{eqn:bbFbbAmorphisms}
\resizebox{.9\hsize}{!}{$
\begin{gathered}
\xymatrix@C=4em{
\ar@{=}[d](\bbF\circ\bbA)_\FFF(\und{M_0},\und{\Sigma_0}) \ar[rrrr]^-{(\bbF\circ\bbA)_\FFF([N,\und{\iota_0},\iota_1])}
&&&& (\bbF\circ\bbA)_\FFF(M_1,\Sigma_1)\ar@{=}[d]\\
\bbA_\FFF(\und{M_0}) & \ar[l]_-{\cong }\bbA_\FFF(\und{V_0})\ar[r]^-{\bbA_\FFF(\und{\iota_0})} & 
\bbA_\FFF(N) & \ar[l]_-{\bbA_\FFF(\iota_1)}^-{\cong}\bbA_\FFF(V_1) \ar[r]^-{\cong}& \bbA_\FFF(M_1)\\
\ar[u]_-{\cong}^-{\iota_{\und{\Sigma_0}}} \FFF(\und{M_0},\und{\Sigma_0}) & 
\ar[l]_-{\cong}^-{\FFF([\und{M_0},\iota_{\und{V_0}}^{\und{M_0}},\id_{\und{M_0}}])}
\ar[u]_-{\cong}^-{\iota_{\und{\Sigma_0}}} \FFF(\und{V_0},\und{\Sigma_0}) 
\ar[r]_-{\FFF([N,\und{\iota_0}\vert,\id_N])}& 
\ar[u]_-{\cong}^-{\iota_{\iota_1(\Sigma_1)}}\FFF(N,\iota_1(\Sigma_1)) & 
\ar[l]_-{\cong}^-{\FFF([N,\iota_1\vert,\id_N])}
\ar[u]_-{\cong}^-{\iota_{\Sigma_1}} \FFF(V_1,\Sigma_1) 
\ar[r]^\cong_-{\FFF([M_1,\iota_{V_1}^{M_1},\id_{M_1}])}& 
\ar[u]_-{\cong}^-{\iota_{\Sigma_1}}\FFF(M_1,\Sigma_1) 
}
\end{gathered}~~.
$}
\end{flalign}
\end{subequations}
The left-pointing isomorphisms in the bottom row can be inverted explicitly
by using \cite[Lemma 3.4]{FFTvAQFT}. This tells us the composite of the bottom row is the $\Alg_{\mathsf{uAs}}(\TT)$-morphism
$\FFF([N,\und{\iota_0},\iota_1]) : \FFF(\und{M_0},\und{\Sigma_0})\to \FFF(M_1,\Sigma_1)$
which is assigned by $\FFF$ to the given $n$-ary operation $[N,\und{\iota_0},\iota_1]$.
Note that the commutative diagrams in \eqref{eqn:bbFbbAmorphisms} verify that the component morphisms
$\iota_{\Sigma} : \FFF(M,\Sigma)\to (\bbF\circ\bbA)_\FFF(M,\Sigma)$ define
a multinatural isomorphism \smash{$\FFF\stackrel{\cong}{\Longrightarrow} (\bbF\circ\bbA)_\FFF$}, for all $\FFF
\in \FQFT_m^{W,\add}$. These multinatural isomorphisms are also natural with respect to morphisms 
$\zeta : \FFF\Rightarrow\GGG$ in $\FQFT_m^{W,\add}$, hence we obtain a natural 
isomorphism \smash{$\id \stackrel{\cong}{\Longrightarrow} \bbF\circ\bbA$}
between the identity functor on $\FQFT_m^{W,\add}$ and the composite $\bbF\circ\bbA$.
\end{proof}


\section{\label{sec:twistedFFT}On twisted globally hyperbolic Lorentzian functorial QFTs}
Let us recall that the main motivation behind our concept of globally hyperbolic Lorentzian FQFTs 
$\FFF : \tau(\LBop_m) \to \Alg_{\mathsf{uAs}}(\TT)$ assigning observable algebras, in contrast to state spaces, 
is rooted in the fact that, especially in Lorentzian geometric contexts, 
algebras are the primary objects of a QFT, while state spaces arise 
only as secondary objects when one studies the representation theory of these algebras.
It is worthwhile to note that a similar point of view has also been taken
in the broader FQFT literature, where a notion of FQFT assigning state spaces
can be introduced in terms of a twisted or relative FQFT with respect to another 
FQFT which assigns algebras (or their representation categories),
see e.g.\ \cite{StolzTeichner,FT} and also \cite[Section 2.3.4]{SatiSchreiber}. 
In this section we will briefly describe how this concept of twisted or relative FQFT adapts to our Lorentzian geometric
context and provide arguments why, in contrast to our algebra-valued FQFTs
from Definition \ref{def:FFT}, it will be difficult, if not impossible, 
to find physically relevant examples of such FQFTs assigning state spaces. 
\sk

Since observable algebras in our context are given by objects $A\in \Alg_{\mathsf{uAs}}(\TT)$, 
it is reasonable to describe state spaces in terms of representations of $A$, i.e.\ (say, left) $A$-modules
internal to the cocomplete closed symmetric monoidal category $\TT$. Such objects can be
accommodated naturally in our framework 
by extending the symmetric monoidal $1$-category $\Alg_{\mathsf{uAs}}(\TT)$ to the symmetric monoidal 
bicategory $\mathbf{Bimod}(\TT)$ whose objects are unital associative algebras $A$,
morphisms $A\to B$ are bimodules ${}_BS_{A}$ and $2$-morphisms
${}_BS_{A}\Rightarrow {}_BT_A$ are bimodule homomorphisms, with everything considered internally in $\TT$.
Indeed, a bimodule ${}_AS_I$ with $I\in \Alg_{\mathsf{uAs}}(\TT)$ the monoidal unit (which is also
an initial object) is precisely the same datum as a left $A$-module.
Using the canonical symmetric monoidal pseudo-functor
\begin{flalign}\label{eqn:AlgtoBimod}
\Alg_{\mathsf{uAs}}(\TT) ~\longrightarrow~\mathbf{Bimod}(\TT)
\end{flalign}
which acts as the identity on objects and sends an $\Alg_{\mathsf{uAs}}(\TT)$-morphism
$f:A\to B$ to its associated bimodule ${}_BB_{A}$, with right $A$-action defined through
$f:A\to B$, we can postcompose any globally hyperbolic Lorentzian FQFT 
$\FFF \in\FQFT_m$ with \eqref{eqn:AlgtoBimod} to obtain a 
FQFT $\FFF : \tau(\LBop_m) \to \mathbf{Bimod}(\TT)$ taking values in the symmetric 
monoidal bicategory $\mathbf{Bimod}(\TT)$. 
\begin{defi}\label{def:twistedFFT}
Let $\FFF \in\FQFT_m$ be a globally hyperbolic Lorentzian FQFT in the sense of Definition \ref{def:FFT}. 
An \textit{$\FFF$-twisted globally hyperbolic Lorentzian FQFT} $\mathfrak{S}$ is a lax 
transformation\footnote{Here the term lax transformation refers to the 
context of $\Cat$-enriched operads in the sense of \cite[Appendix A]{2AQFT}. 
These are concretely defined by dropping the invertibility requirement of the natural isomorphism data
of a pseudo-transformation as in \cite[Definition A.4]{2AQFT}.}
\begin{equation}
\begin{tikzcd}[column sep=large]
\tau(\LBop_m)
  \arrow[bend left=40]{r}[name=U,label=above:$\mathfrak{I}$]{}
  \arrow[bend right=40]{r}[name=D,label=below:$\FFF$]{} &
\mathbf{Bimod}(\TT)
  \arrow[shorten <=4pt,shorten >=4pt,Rightarrow,to path={(U) -- node[label=left:$\mathfrak{S}$] {} (D)}]{}
\end{tikzcd}
\end{equation}
from the constant multifunctor $\mathfrak{I}$
assigning the initial algebra to the pseudo-multifunctor 
given by postcomposing $\FFF \in\FQFT_m$ with \eqref{eqn:AlgtoBimod}.
\end{defi}

\begin{rem}
It is a straightforward exercise in higher category theory to unpack this definition,
which leads to the following explicit description of
an $\FFF$-twisted globally hyperbolic Lorentzian FQFT $\mathfrak{S}$:
\begin{itemize}
\item To each object $(M,\Sigma)\in \tau(\LBop_m)$ is assigned a
left $\FFF(M,\Sigma)$-module $\mathfrak{S}(M,\Sigma)$. We will denote the left module
structure by $l_{(M,\Sigma)}: \FFF(M,\Sigma)\otimes \mathfrak{S}(M,\Sigma)\to \mathfrak{S}(M,\Sigma)$.

\item To each bordism $[N,\und{\iota_0},\iota_1] : (\und{M_0},\und{\Sigma_0})\to (M_1,\Sigma_1)$
in $\tau(\LBop_m)$ is assigned a left module homomorphism
\begin{flalign}\label{eqn:statebordism}
\mathfrak{S}([N,\und{\iota_0},\iota_1])\,:\, \mathfrak{S}(\und{M_0},\und{\Sigma_0})\,:=\,\bigotimes_{i=1}^n
\mathfrak{S}(M_{0_i},\Sigma_{0_i})~\longrightarrow~\mathfrak{S}(M_1,\Sigma_1)
\end{flalign}
relative to the $\Alg_{\mathsf{uAs}}(\TT)$-morphism $\FFF([N,\und{\iota_0},\iota_1]) : \FFF(\und{M_0},\und{\Sigma_0})\to 
\FFF(M_1,\Sigma_1)$, i.e.\ the diagram
\begin{flalign}
\begin{gathered}
\xymatrix@C=4em{
\ar[d]_-{\FFF([N,\und{\iota_0},\iota_1])\otimes \mathfrak{S}([N,\und{\iota_0},\iota_1])}
\FFF(\und{M_0},\und{\Sigma_0})\otimes \mathfrak{S}(\und{M_0},\und{\Sigma_0}) \ar[r]^-{l_{(\und{M_0},\und{\Sigma_0})}}
~&~\mathfrak{S}(\und{M_0},\und{\Sigma_0})\ar[d]^-{\mathfrak{S}([N,\und{\iota_0},\iota_1])}\\
\FFF(M_1,\Sigma_1)\otimes \mathfrak{S}(M_1,\Sigma_1)\ar[r]_-{l_{(M_1,\Sigma_1)}}~&~\mathfrak{S}(M_1,\Sigma_1)
}
\end{gathered}
\end{flalign}
in $\TT$ commutes. In particular, to each $0$-to-$1$ bordism $\varnothing\to (M,\Sigma)$
is assigned a left module homomorphism $I \to \mathfrak{S}(M,\Sigma)$
relative to the algebra unit $I\to \FFF(M,\Sigma)$, which selects an
element of the left $\FFF(M,\Sigma)$-module $\mathfrak{S}(M,\Sigma)$.
\end{itemize}
These data have to satisfy the unitality, associativity and permutation equivariance
conditions of an algebra over the operad $\tau(\LBop_m)$. Explicitly, 
for every identity bordism we have that $\mathfrak{S}([M,\id_M,\id_M]) = \id$, 
for every operadically composable family of bordisms the diagram
\begin{flalign}\label{eqn:statecomposition}
\begin{gathered}
\xymatrix@C=5em@R=2em{
~&~ \mathfrak{S}(\und{M_1},\und{\Sigma_1}) \ar[rd]^-{~~~\mathfrak{S}([N_1,\und{\iota_{10}},\iota_{11}])}~&~\\
\mathfrak{S}(\und{\und{M_0}},\und{\und{\Sigma_0}}) \ar[ru]^-{\mathfrak{S}([\und{N_0},\und{\und{\iota_{00}}},\und{\iota_{01}}])~~~}
\ar[rr]_-{\mathfrak{S}([N_1,\und{\iota_{10}},\iota_{11}]\,[\und{N_0},\und{\und{\iota_{00}}},\und{\iota_{01}}])}~&~ ~&~\mathfrak{S}(M_2,\Sigma_2)
}
\end{gathered}
\end{flalign}
commutes, and for every bordism and permutation the diagram
\begin{flalign}
\begin{gathered}
\xymatrix@C=2em@R=2em{
~&~ \mathfrak{S}(M_1,\Sigma_1)~&~\\
\mathfrak{S}(\und{M_0},\und{\Sigma_0}) \ar[rr]_-{\cong}
\ar[ru]^-{\mathfrak{S}([N,\und{\iota_0},\iota_1])~~~}~&~ ~&~\mathfrak{S}(\und{M_0}\sigma,\und{\Sigma_0}\sigma)
\ar[lu]_-{~~~\mathfrak{S}([N,\und{\iota_0}\sigma,\iota_1])}
}
\end{gathered}
\end{flalign}
commutes, where the unlabeled isomorphism is given by the symmetric braiding on $\TT$.
\sk

We would like to highlight that, by forgetting the left module structures, one obtains a multifunctor
$\mathfrak{S}^{\mathrm{forget}}: \tau(\LBop_m)\to \TT$, i.e.\ a globally hyperbolic Lorentzian FQFT 
with values in $\TT$. This implies that
an $\FFF$-twisted globally hyperbolic Lorentzian FQFT $\mathfrak{S}$ is a genuinely richer concept
than an ordinary $\TT$-valued globally hyperbolic Lorentzian FQFT since it encodes
through the left module structures additional information about how the observables described 
by $\FFF$ act on the state spaces described by $\mathfrak{S}$.
\end{rem}

We will now argue why it will be difficult, if not impossible, to obtain physically relevant examples
of $\FFF$-twisted globally hyperbolic Lorentzian FQFTs. For presenting our arguments,
it will be convenient to assume that $\FFF\in\FQFT_m^{W,\add}$ is additive and satisfies the time-slice axiom,
which, by our Equivalence Theorem \ref{theo:equivalence}, implies that $\FFF\cong \mathbb{F}_\AAA$
can be represented through Construction \ref{constr:AQFTtoFQFT} by an additive AQFT satisfying the time-slice axiom 
$\AAA\in \AQFT_m^{W,\add}$. From this perspective, the specification of an
$\FFF$-twisted globally hyperbolic Lorentzian FQFT $\mathfrak{S}$
requires in particular a choice of representation $\mathfrak{S}(M,\Sigma)$
of the AQFT observable algebra $\AAA(M)$ and a distinguished `vacuum' state $I\to \mathfrak{S}(M,\Sigma)$, 
for every spacetime $M\in\Loc_m$ and every Cauchy surface $\Sigma\subset M$. At the level of simple 
examples, such as the free Klein-Gordon quantum field, it is well understood that
there are no distinguished representations and vacuum states
on general (non-stationary) spacetimes $M\in\Loc_m$, see e.g.\ \cite[Chapter 4]{Wald}, hence the selection
of these data will necessarily come with some ad-hoc choices. 
\sk

However, this creates problems when defining the structure maps \eqref{eqn:statebordism} of an 
$\FFF$-twisted globally hyperbolic Lorentzian FQFT $\mathfrak{S}$. As an illustrative example,
let us consider the Cauchy bordism $[N,\iota_{V_0}^N,\iota_{V_1}^N] : (V_0,\Sigma_0)\to (V_1,\Sigma_1)$
which is associated with a choice of spacetime $N\in\Loc_m$, two causally convex open time-slabs $V_0,V_1\subseteq N$
with $V_0\subseteq I^-_{N}(V_1)$ in the past of $V_1$, and two Cauchy surfaces $\Sigma_0\subset V_0$ and $\Sigma_{1}\subset V_1$.
A choice of structure map 
\begin{flalign}\label{eqn:timeevolutionstates}
\mathfrak{S}([N,\iota_{V_0}^N,\iota_{V_1}^{N}])\,:\, \mathfrak{S}(V_0,\Sigma_0)~\longrightarrow~ \mathfrak{S}(V_1,\Sigma_1)
\end{flalign}
associated with this bordism is then precisely an implementation of the 
time-evolution of quantum observables $\AAA([N,\iota_{V_0}^N,\iota_{V_1}^{N}]) : 
\AAA(V_0) \to \AAA(V_1)$ at the level of state spaces.
For simple examples such as the free Klein-Gordon quantum field, 
it is well understood that such an implementation is in general not possible
as a consequence of the existence of multiple inequivalent representations
of the observable algebras, see e.g.\ \cite[Chapter 4.4]{Wald}. 
It is worthwhile to highlight that, even in those cases in which an implementation 
\eqref{eqn:timeevolutionstates} exists, it will, in general, fail to preserve the distinguished
`vacuum' states and thereby violate the property \eqref{eqn:statecomposition} of a twisted FQFT.
To explain the origin of this phenomenon, let us assume that \eqref{eqn:timeevolutionstates}
preserves the distinguished `vacuum' states for a fixed choice of Cauchy 
bordism $[N,\iota_{V_0}^N,\iota_{V_1}^N] : (V_0,\Sigma_0)\to (V_1,\Sigma_1)$.
Deforming the geometry of $N\in\Loc_m$ in the intermediate region between $V_0$ and $V_1$, one can construct
a new spacetime $\tilde{N}\in\Loc_m$ and a new Cauchy bordism 
$[\tilde{N},\iota_{V_0}^{\tilde{N}},\iota_{V_1}^{\tilde{N}}] : (V_0,\Sigma_0)\to (V_1,\Sigma_1)$
with the same source and target. The structure map
$\mathfrak{S}([\tilde{N},\iota_{V_0}^{\tilde{N}},\iota_{V_1}^{\tilde{N}}]): 
\mathfrak{S}(V_0,\Sigma_0)\to \mathfrak{S}(V_1,\Sigma_1)$ associated with this new bordism,
if it exists, does not, in general, preserve the distinguished `vacuum' states as a consequence
of the physical effect called `particle creation in curved spacetimes', see e.g.\ \cite{Wald}.
Indeed, in general the `vacuum' state $I\to \mathfrak{S}(V_0,\Sigma_0)$ 
does not evolve under $\mathfrak{S}([\tilde{N},\iota_{V_0}^{\tilde{N}},\iota_{V_1}^{\tilde{N}}])$
to the `vacuum' state $I\to \mathfrak{S}(V_1,\Sigma_1)$, but rather to 
a different state with non-trivial particle content.
\sk

We hope that these physical and example-based arguments plausibly 
substantiate why we believe that it will be difficult, if not impossible, to obtain physically relevant 
examples of $\FFF$-twisted globally hyperbolic Lorentzian FQFTs $\mathfrak{S}$.
It may be possible to prove a precise and example-independent statement
supporting this claim by adapting the no-go result \cite[Theorem 6.13]{FVstate} 
for the existence of natural states in AQFT to our context. This is an 
interesting open problem which we will however not investigate in the present paper.


\section*{Acknowledgments}
We would like to thank Marco Benini for useful discussions
and Alastair Grant-Stuart for sharing technical notes
which helped us to establish the results in Appendix \ref{app:Lorentzian}.
We also would like to thank the reviewers for their useful comments
which helped us to improve our paper and encouraged us to add Section \ref{sec:twistedFFT}.
J.M.\ is funded by an EPSRC PhD scholarship (2742043) 
of the School of Mathematical Sciences at the University of Nottingham.
A.S.\ was supported by the Royal Society (UK) through a Royal Society University 
Research Fellowship (URF\textbackslash R\textbackslash 211015)
and Enhancement Grants (RF\textbackslash ERE\textbackslash 210053 and 
RF\textbackslash ERE\textbackslash 231077).


\appendix

\section{\label{app:pseudooperads}Basic theory of pseudo-operads in $\Grpd$}
In this appendix we introduce and develop an operadic generalization
of the concept of pseudo-categories \cite{PseudoCats} internal to the $2$-category
$\Grpd$ of groupoids, functors and natural transformations.
As we will explain in the main text, such pseudo-operads provide a suitable framework
to describe globally hyperbolic Lorentzian bordisms with spatially local features
and multiplicative structures.
They generalize the globally hyperbolic Lorentzian bordism pseudo-categories from
\cite{FFTvAQFT} which are based on the approach of Stolz and Teichner \cite{StolzTeichner}.
\begin{defi}\label{def:pseudooperad}
A \textit{(colored symmetric) pseudo-operad $\scrO$ in $\Grpd$} consists of the following data:
\begin{itemize}
\item[(i)] A groupoid of objects $\scrO_0$.

\item[(ii)] A sequence of spans of groupoids and functors
\begin{flalign}
\begin{gathered}
\xymatrix@C=1em@R=1.5em{
~&~\ar[dl]_-{t^n}\scrO_1^n \ar[dr]^-{s^n}~&~\\
\scrO_0 ~&~ ~&~\scrO_0^{\times n}
}
\end{gathered}\quad,
\end{flalign}
for all non-negative integers $n\in\bbN_0$, describing the $n$-ary operations with their source $s^n$ and target $t^n$.

\item[(iii)] A family of operadic composition functors
\begin{flalign}
\cmp\,:\, \scrO_1^{n}\times_{\scrO_{0}^{\times n}} \scrO_1^{\und{k}} ~\longrightarrow~\scrO_1^{\Sigma\und{k}}\quad,
\end{flalign}
for all positive integers $n\in\bbN$ and all $n$-tuples of non-negative integers
$\und{k} = (k_1,\dots,k_n)\in \bbN_0^n$, where 
$\scrO_1^{\und{k}}\coloneqq \prod_{i=1}^n\scrO_1^{k_i}$ denotes the product
in $\Grpd$ and $\Sigma\und{k} \coloneqq \sum_{i=1}^n k_i\in\bbN_0$. The domain of
these functors is the (strict) fiber product
\begin{flalign}
\begin{gathered}
\xymatrix@C=0.5em@R=1.8em{
~&~\ar@{-->}[dl]_-{p_1} \scrO_1^{n}\times_{\scrO_{0}^{\times n}} \scrO_1^{\und{k}} \ar@{-->}[dr]^-{p_2}~&~\\
\ar[dr]_-{s^n}\scrO_1^{n} ~&~ ~&~ \ar[dl]^-{t^{\und{k}}}\scrO_1^{\und{k}}\\
~&~ \scrO_{0}^{\times n}~&~
}
\end{gathered}
\end{flalign}
in $\Grpd$, where $t^{\und{k}}\coloneqq \prod_{i=1}^n t^{k_i}$. The  operadic composition functors are required to
be maps of spans, i.e.\ the diagrams
\begin{flalign}
\begin{gathered}
\xymatrix@C=0.5em@R=1.8em{
~&~\ar[dl]_-{t^n\,p_1} \scrO_1^{n}\times_{\scrO_{0}^{\times n}} \scrO_1^{\und{k}} \ar[dd]^-{\cmp} 
\ar[dr]^-{ s^{\und{k}} \,p_2}~&~\\
\scrO_0 ~&~ ~&~ \scrO_0^{\times \Sigma\und{k}}\\
~&~ \ar[lu]^-{t^{\Sigma\und{k}}}\scrO_1^{\Sigma\und{k}}\ar[ru]_-{s^{\Sigma\und{k}}}~&~
}
\end{gathered}
\end{flalign}
in $\Grpd$ commute strictly.

\item[(iv)] An operadic unit functor 
\begin{flalign}
u \,:\, \scrO_0~\longrightarrow~\scrO_1^1
\end{flalign}
which is required to be a map of spans, i.e.\ the diagram
\begin{flalign}
\begin{gathered}
\xymatrix@C=1em@R=1.2em{
~&~\ar@{=}[dl] \scrO_0 \ar[dd]^-{u}  \ar@{=}[dr]~&~\\
\scrO_0 ~&~ ~&~ \scrO_0\\
~&~ \ar[lu]^-{t^{1}}\scrO_1^{1}\ar[ru]_-{s^{1}}~&~
}
\end{gathered}
\end{flalign}
in $\Grpd$ commutes strictly.

\item[(v)] A sequence of right actions $\scrO_1^n : \mathbf{B}S_n^\op\to \Grpd$ of the permutation groups
$S_n$ on the groupoids of $n$-ary operations from item (ii), for all $n\in\bbN_0$. These permutation
actions must be compatible with the source and target functors in the sense that the diagrams
\begin{flalign}
\begin{gathered}
\xymatrix{
\ar@{=}[d]\scrO_0 ~&~\ar[l]_-{t^n}\ar[d]_-{\scrO_1^n(\sigma)} \scrO_1^n \ar[r]^-{s^n}~&~\scrO_0^{\times n} \ar[d]^-{\sigma^\ast}\\
\scrO_0 ~&~\ar[l]^-{t^n} \scrO_1^n \ar[r]_-{s^n}~&~\scrO_0^{\times n}
}
\end{gathered}
\end{flalign}
in $\Grpd$ commute strictly, for all $\sigma\in S_n$, where $\sigma^\ast$ denotes the right action
which is given by permuting the factors of an $n$-fold product. Furthermore,
the operadic composition functors are required to be equivariant with respect to the permutation actions,
i.e.\ the diagrams
\begin{subequations}
\begin{flalign}
\begin{gathered}
\xymatrix@C=3em{
\ar[d]_-{\scrO_1^n(\sigma)\times\sigma^\ast}\scrO_1^n \times_{\scrO_0^{\times n}} \scrO_1^{\und{k}} \ar[r]^-{\cmp}~&~\scrO_1^{\Sigma\und{k}}\ar[d]^-{\scrO_1^{\Sigma\und{k}}(\sigma\langle k_1,\dots,k_n\rangle)} ~&~\\
\scrO_1^n \times_{\scrO_0^{\times n}} \scrO_1^{\und{k}\sigma} \ar[r]_-{\cmp}~&~\scrO_1^{\Sigma\und{k}}
}
\end{gathered}
\end{flalign}
\begin{flalign}
\begin{gathered}
\xymatrix@C=3em{
\ar[d]_-{\id\times\scrO_1^{\und{k}}(\und{\sigma})}\scrO_1^n \times_{\scrO_0^{\times n}} \scrO_1^{\und{k}} \ar[r]^-{\cmp}~&~\scrO_1^{\Sigma\und{k}}\ar[d]^-{\scrO_1^{\Sigma\und{k}}(\sigma_1\oplus\cdots\oplus \sigma_n)} ~&~\\
\scrO_1^n \times_{\scrO_0^{\times n}} \scrO_1^{\und{k}} \ar[r]_-{\cmp}~&~\scrO_1^{\Sigma\und{k}}
}
\end{gathered}
\end{flalign}
\end{subequations}
in $\Grpd$ commute strictly, for all $\sigma\in S_n$ and 
$\und{\sigma} = (\sigma_1,\dots,\sigma_n)\in S_{\und{k}}$ 
with $\sigma_i\in S_{k_i}$, for $i=1,\dots,n$,
where $\sigma\langle k_1,\dots,k_n\rangle\in S_{\Sigma\und{k}}$ denotes the induced block permutation
and $\sigma_1\oplus\cdots\oplus\sigma_n\in S_{\Sigma\und{k}}$ denotes the induced sum permutation.

\item[(vi)] Associator natural isomorphisms filling the diagrams
\begin{flalign}
\begin{gathered}
\xymatrix@C=1.5em@R=1em{
\ar[dd]_-{\id\times\und{\cmp}}
\scrO_1^n \times_{\scrO_0^{\times n}} \scrO_1^{\und{k}} \times_{\scrO_0^{\times \Sigma\und{k}}} \scrO_1^{\und{\und{l}}} \ar[rr]^-{\cmp\times\id}~&~ ~&~
\scrO_1^{\Sigma\und{k}}  \times_{\scrO_0^{\times \Sigma\und{k}}} \scrO_1^{\und{\und{l}}}\ar[dd]^-{\cmp}
\ar@{}[lldd]^(.35){}="a"^(.65){}="b" \ar@{=>}_-{\mathsf{a}} "a";"b"\\
~&~ ~&~ \\
\scrO_1^n \times_{\scrO_0^{\times n}} \scrO_1^{\Sigma\und{\und{l}}}
\ar[rr]_-{\cmp} ~&~ ~&~ \scrO_1^{\Sigma\Sigma\und{\und{l}}}
}
\end{gathered}
\end{flalign}
in $\Grpd$, where $\und{\und{l}} = (\und{l}_1,\dots,\und{l}_n) = 
(l_{11},\dots,l_{1k_1},\dots,l_{n1},\dots ,l_{nk_n})
\in\bbN_0^{\Sigma\und{k}}$ denotes a double indexed tuple with summations 
$\Sigma \und{\und{l}} \coloneqq (\Sigma\und{l}_1,\dots,\Sigma\und{l}_n)\in\bbN_0^n$
and $\Sigma\Sigma\und{\und{l}}\coloneqq\sum_{i=1}^n\sum_{j=1}^{k_i}l_{ij}\in\bbN_0$,
as well as left and right unitor natural isomorphisms filling the diagrams
\begin{flalign}
\begin{gathered}
\xymatrix@C=2em@R=1em{
~&~ \ar[dl]_-{(u\,t^n,\id)~~} \ar@{=}[dd] \scrO_1^n \ar[dr]^-{~~(\id, u^{\times n}\,s^{n})}~&~\\
\ar[dr]_-{\cmp}\scrO_1^1 \times_{\scrO_0} \scrO_1^n 
\ar@{}[r]^(.5){}="a"^(.85){}="b" \ar@{=>}^-{\mathsf{l}} "a";"b" 
~&~ ~~~&~
\ar@{}[l]^(.55){}="a"^(.85){}="b" \ar@{=>}_-{\mathsf{r}} "a";"b"
\scrO_1^n\times_{\scrO_0^{\times n}} \big(\scrO_1^1\big)^{\times n} \ar[dl]^-{\cmp}\\
~&~ \scrO_1^n ~&~
}
\end{gathered}
\end{flalign}
in $\Grpd$. The natural isomorphisms $(\mathsf{a},\mathsf{l},\mathsf{r})$ 
are required to be globular, i.e.\ their images
under the source and target functors are identities in $\scrO_0$, and they must 
satisfy the typical triangle and pentagon axioms, see e.g.\ \cite[Definition 12.3.7]{Yau}.
\end{itemize}
\end{defi}

\begin{rem}\label{rem:pseudooperad}
Spelling out the data from Definition \ref{def:pseudooperad},
one finds that a pseudo-operad $\scrO$ consists of \textit{objects} $c\in\scrO_0$,
\textit{vertical morphisms} $(g : c\to c^\prime)\in \scrO_0$, (horizontal) \textit{$n$-ary operations}
$\psi\in \scrO_1^n$ and \textit{$2$-cells} $(\alpha : \psi\Rightarrow \psi^\prime)\in\scrO_1^n$ between
$n$-ary operations, for all $n\in\bbN_0$. The justification for interpreting
$\psi\in \scrO_1^n$ as an $n$-ary operation is given by the source and target functors,
which allow us to regard $\psi : \und{c}\hto d$ as an operation from an $n$-tuple
of objects $\und{c} =s^n(\psi)\in\scrO_0^{\times n}$ to a single object $d = t^n(\psi)\in\scrO_0$.
Applying the source and target functors to a $2$-cell 
$(\alpha : \psi\Rightarrow \psi^\prime)\in\scrO_1^n$ yields a square
\begin{flalign}\label{eqn:2celltmp}
\begin{gathered}
\xymatrix@R=0.25em@C=0.25em{
\und{c}^\prime  \ar[rr]|{\sla}^-{\psi^\prime}&& d^\prime\\
&{\scriptstyle\alpha}~\rotatebox[origin=c]{90}{$\Rightarrow$}~~&\\
\ar[uu]^-{\und{g}} \und{c} \ar[rr]|{\sla}_-{\psi} && d \ar[uu]_-{h}
}
\end{gathered}\quad,
\end{flalign}
where $(\und{g} = s^n(\alpha) : \und{c}\to \und{c}^\prime)\in\scrO_0^{\times n}$ is 
an $n$-tuple of vertical morphisms and $(h = t^n(\alpha): d\to d^\prime)\in\scrO_0$ is
a single vertical morphism. Notably, a feature of pseudo-operads which is not present
in pseudo-categories \cite{PseudoCats} are the permutation actions from item (v) 
of Definition \ref{def:pseudooperad}. In particular, acting with a permutation $\sigma\in S_n$ on
an $n$-ary operation $\psi : \und{c}\hto d$ gives an
$n$-ary operation $\scrO_1^n(\sigma)(\psi) : \und{c}\sigma\hto d$ from the permuted
tuple of objects $\und{c}\sigma \coloneqq (c_{\sigma(1)},\dots, c_{\sigma(n)})\in\scrO_0^{\times n}$. 
Acting with $\sigma\in S_n$ on a $2$-cell \eqref{eqn:2celltmp} gives a $2$-cell
\begin{flalign}
\begin{gathered}
\xymatrix@R=0.25em@C=0.25em{
\und{c}^\prime\sigma  \ar[rr]|{\sla}^-{\scrO_1^n(\sigma)(\psi^\prime)}&& d^\prime\\
&{\scriptstyle\scrO_1^n(\sigma)(\alpha)}~\rotatebox[origin=c]{90}{$\Rightarrow$}~~&\\
\ar[uu]^-{\und{g}\sigma} \und{c}\sigma \ar[rr]|{\sla}_-{\scrO_1^n(\sigma)(\psi)} && d \ar[uu]_-{h}
}
\end{gathered}
\end{flalign}
with $\und{g}\sigma \coloneqq (g_{\sigma(1)},\dots,g_{\sigma(n)})$ the permuted tuple of vertical morphisms.
\sk

The compositions in the groupoids $\scrO_0$ and $\scrO_1^n$
define, respectively, a vertical composition
$g^\prime\, g: c\to c^{\prime\prime}$ of vertical morphisms 
$g:c\to c^\prime$ and $g^\prime:c^\prime \to c^{\prime\prime}$,
and a vertical composition $\alpha^\prime\,\alpha : \psi\Rightarrow \psi^{\prime\prime}$ of $2$-cells
$\alpha : \psi\Rightarrow \psi^\prime$ and $\alpha^\prime : \psi^\prime\Rightarrow \psi^{\prime\prime}$
between $n$-ary operations, for all $n\in\bbN_0$. 
These vertical compositions are strictly associative and unital with respect to the 
identities $(\id_{c}:c\to c)\in \scrO_0$ and $(\id_\psi : \psi\Rightarrow \psi)\in \scrO_1^n$.
The operadic composition functors $\cmp$ define operadic compositions
$\psi\cmp \und{\phi}:\und{\und{a}}\hto d$ of $n$-ary operations
$\psi : \und{c} \hto d$ with $n$-tuples 
$\und{\phi}=(\phi_1,\dots,\phi_n): \und{\und{a}}\hto \und{c}$
of $k_i$-ary operations $\phi_i : \und{a}_i\hto c_i$, for $i=1,\dots,n$, 
as well as operadic compositions of (tuples of) $2$-cells
\begin{flalign}
\begin{gathered}
\xymatrix@R=0.25em@C=0.25em{
\und{\und{a}}^\prime  \ar[rr]|{\sla}^-{\und{\phi}^\prime}&& \und{c}^\prime \ar[rr]|{\sla}^-{\psi^\prime} && d^\prime\\
&{\scriptstyle\und{\beta}}~\rotatebox[origin=c]{90}{$\Rightarrow$}~~& &{\scriptstyle\alpha}~\rotatebox[origin=c]{90}{$\Rightarrow$}~~&\\
\ar[uu]^-{\und{\und{f}}} \und{\und{a}} \ar[rr]|{\sla}_-{\und{\phi}} && \und{c} \ar[uu]^-{\und{g}} \ar[rr]|{\sla}_-{\psi}&& d\ar[uu]_-{h}
}
\end{gathered}
~~~\stackrel{\cmp}{\longmapsto}~~~
\begin{gathered}
\xymatrix@R=0.25em@C=0.25em{
\und{\und{a}}^\prime  \ar[rr]|{\sla}^-{\psi^\prime\cmp \und{\phi}^\prime}&& d^\prime \\
&{\scriptstyle\alpha\cmp \und{\beta}}~\rotatebox[origin=c]{90}{$\Rightarrow$}~~~~&\\
\ar[uu]^-{\und{\und{f}}} \und{\und{a}} \ar[rr]|{\sla}_-{\psi\cmp \und{\phi}} && d\ar[uu]_-{h}
}
\end{gathered}\quad.
\end{flalign}
These operadic compositions are only weakly associative, with associator $\mathsf{a}$,
and weakly unital, with unitors $\mathsf{l}$ and $\mathsf{r}$,
with respect to the units obtained by the functor $u:\scrO_0\to \scrO_1^1$.
Note that the two compositions of $2$-cells satisfy the strict interchange law
\begin{flalign}
(\alpha^\prime\,\alpha)\cmp (\und{\beta}^\prime\,\und{\beta}) \,=\,
(\alpha^\prime \cmp \und{\beta}^\prime)~(\alpha\cmp \und{\beta})\quad,
\end{flalign}
as a consequence of the functoriality of $\cmp$.
\end{rem}

\begin{defi}\label{def:pseudomultifunctor}
A \textit{pseudo-multifunctor} $F : \scrO\to\scrP$ between two pseudo-operads
$\scrO$ and $\scrP$ in $\Grpd$ consists of the following data:
\begin{itemize}
\item[(i)] A functor $F_0 : \scrO_0\to \scrP_0$ between the groupoids of objects.

\item[(ii)] A sequence of functors $F_1^n : \scrO_1^n\to \scrP_1^n$
between the groupoids of $n$-ary operations, for all $n\in\bbN_0$, 
which are required to be maps of spans, i.e.\ the diagrams
\begin{flalign}
\begin{gathered}
\xymatrix{
\ar[d]_-{F_0}\scrO_0 ~&~ \ar[l]_-{t^n_{\scrO}}\scrO_1^n \ar[d]_-{F_1^n}\ar[r]^-{s^n_{\scrO}}~&~ \scrO_0^{\times n}\ar[d]^-{F_0^{\times n}}\\
\scrP_0 ~&~ \ar[l]^-{t^n_{\scrP}}\scrP_1^n \ar[r]_-{s^n_{\scrP}}~&~ \scrP_0^{\times n}
}
\end{gathered}
\end{flalign}
in $\Grpd$ commute strictly. These functors must further be equivariant with respect to 
the permutation actions, i.e.\ the diagrams
\begin{flalign}
\begin{gathered}
\xymatrix@C=3em{
\ar[d]_-{\scrO_1^n(\sigma)}\scrO_1^n \ar[r]^-{F_1^n}~&~\scrP_1^n \ar[d]^-{\scrP_1^n(\sigma)}\\
\scrO_1^n \ar[r]_-{F_1^n}~&~\scrP_1^n
}
\end{gathered}
\end{flalign}
in $\Grpd$ commute strictly, for all $\sigma\in S_n$.

\item[(iii)] Natural isomorphisms filling the diagrams
\begin{subequations}
\begin{flalign}
\begin{gathered}
\xymatrix@C=4em{
\ar[d]_-{\cmp_{\scrO}}\scrO_1^n \times_{\scrO_0^{\times n}} \scrO_{1}^{\und{k}} \ar[r]^-{F_1^n\times F_1^{\und{k}}}~&~ \scrP_1^n \times_{\scrP_0^{\times n}} \scrP_{1}^{\und{k}} \ar[d]^-{\cmp_{\scrP}}
\ar@{}[dl]^(.35){}="a"^(.65){}="b" \ar@{=>}_-{F^{\cmp}} "a";"b"\\
\scrO_1^{\Sigma\und{k}} \ar[r]_-{F_1^{\Sigma\und{k}}}~&~\scrP_1^{\Sigma\und{k}}
}
\end{gathered}
\end{flalign}
\begin{flalign}
\begin{gathered}
\xymatrix@C=3em{
\ar[d]_-{u_{\scrO}}\scrO_0 \ar[r]^-{F_0}~&~ \scrP_0\ar[d]^-{u_{\scrP}}
\ar@{}[dl]^(.35){}="a"^(.65){}="b" \ar@{=>}_-{F^{u}} "a";"b"\\
\scrO_1^1 \ar[r]_-{F_1^1}~&~ \scrP_1^1
}
\end{gathered}
\end{flalign}
\end{subequations}
in $\Grpd$. The natural isomorphisms $(F^{\cmp},F^{u})$ are required to be globular, i.e.\ their images
under the source and target functors are identities in $\scrP_0$, and they must 
satisfy analogous coherence axioms to those of a monoidal functor, see e.g.\ \cite[Definition 12.3.18]{Yau}.
\end{itemize}
\end{defi}

\begin{defi}\label{def:multitransformation}
A \textit{(vertical) multitransformation} $\zeta: F\Rightarrow G$ between two pseudo-multifunctors
$F,G: \scrO\to\scrP$ consists of the following data:
\begin{itemize}
\item[(i)] A natural transformation $\zeta_0 : F_0\Rightarrow G_0 : \scrO_0\to\scrP_0$.

\item[(ii)] A sequence of natural transformations $\zeta_1^n : F_1^n\Rightarrow G_1^n : \scrO_1^n\to\scrP_1^n$,
for all $n\in\bbN_0$, satisfying
\begin{flalign}
t^n_{\scrP} \, \zeta_1^n\,=\, \zeta_0\,t^{n}_{\scrO}\quad,\qquad 
s^n_{\scrP}\,\zeta_1^n\,=\ \zeta_0^{\times n}\, s^n_{\scrO}\quad.
\end{flalign}
\end{itemize}
This data must satisfy the following properties:
For all $n$-ary operations $(\psi :\und{c}\hto d)\in\scrO_1^n$ and
all $n$-tuples $\big(\und{\phi}=(\phi_1,\dots,\phi_n) : \und{\und{a}}\hto \und{c}\big)\in\scrO_1^{\und{k}}$ 
of $k_i$-ary operations $(\phi_i :\und{a}_i\hto c_i)\in\scrO_1^{k_i}$, for $i=1,\dots,n$, 
the compositions of $2$-cells 
\begin{subequations}
\begin{flalign}
\begin{gathered}
\xymatrix@R=0.5em@C=0.75em{
G_0(\und{\und{a}}) \ar[rr]|{\sla}^-{G_1^{\Sigma\und{k}}(\psi\cmp \und{\phi})} &&G_0(d)\\
&{\scriptstyle (\zeta_1^{\Sigma\und{k}})_{\psi\cmp \und{\phi}}}~\rotatebox[origin=c]{90}{$\Rightarrow$}_{~_{~}}~&\\
\ar[uu]^-{(\zeta_0)_{\und{\und{a}}}} 
F_0(\und{\und{a}})  \ar[rr]|{\sla}^-{F_1^{\Sigma\und{k}}(\psi\cmp \und{\phi})}&& F_0(d) \ar[uu]_-{(\zeta_0)_{d}}\\
& {\scriptstyle F^{\cmp}_{(\psi,\und{\phi})}}~\rotatebox[origin=c]{90}{$\Rightarrow$}~&\\
\ar@{=}[uu] 
F_0(\und{\und{a}}) \ar[r]|{\sla}_-{F_1^{\und{k}}(\und{\phi})} & F_0(\und{c}) \ar[r]|{\sla}_-{F_1^n(\psi)}& 
F_0(d)\ar@{=}[uu]
}
\end{gathered}
~~=~~
\begin{gathered}
\xymatrix@R=0.5em@C=0.25em{
G_0(\und{\und{a}}) \ar[rrrr]|{\sla}^-{G_1^{\Sigma\und{k}}(\psi\cmp\und{\phi})} && && G_0(d)\\
&& {\scriptstyle G^{\cmp}_{(\psi,\und{\phi})}}~\rotatebox[origin=c]{90}{$\Rightarrow$}~&&\\
\ar@{=}[uu]
G_0(\und{\und{a}}) \ar[rr]|{\sla}^-{G_1^{\und{k}}(\und{\phi})} && G_0(\und{c}) \ar[rr]|{\sla}^-{G_1^n(\psi)}& &
G_0(d)\ar@{=}[uu]\\
&{\scriptstyle (\zeta_1^{\und{k}})_{\und{\phi}}}~\rotatebox[origin=c]{90}{$\Rightarrow$}&  &{\scriptstyle (\zeta_1^n)_{\psi}}~\rotatebox[origin=c]{90}{$\Rightarrow$}&\\
\ar[uu]^-{(\zeta_0)_{\und{\und{a}}}}
F_0(\und{\und{a}}) \ar[rr]|{\sla}_-{F_1^{\und{k}}(\und{\phi})} && \ar[uu]^-{(\zeta_0)_{\und{c}}}F_0(\und{c}) \ar[rr]|{\sla}_-{F_1^{n}(\psi)}& & F_0(d)\ar[uu]_-{(\zeta_0)_{d}}
}
\end{gathered}
\end{flalign}
in $\scrP_1^{\Sigma\und{k}}$ coincide. Furthermore, for all objects $c\in\scrO_0$, 
the compositions of $2$-cells
\begin{flalign}
\begin{gathered}
\xymatrix@R=0.5em@C=0.25em{
G_0(c)\ar[rr]|{\sla}^-{G_1^1 u(c)}&&G_0(c)\\
&{\scriptstyle (\zeta_1^1)_{u(c)}}~\rotatebox[origin=c]{90}{$\Rightarrow$}&\\
\ar[uu]^-{(\zeta_0)_c} F_0(c)\ar[rr]|{\sla}^-{F_1^1u(c)}&& F_0(c)\ar[uu]_-{(\zeta_0)_c}\\
&{\scriptstyle F^u_{c}}~\rotatebox[origin=c]{90}{$\Rightarrow$}&\\
\ar@{=}[uu] F_0(c) \ar[rr]|{\sla}_-{uF_0(c)}&& F_0(c)\ar@{=}[uu] 
}
\end{gathered}
~~=~~
\begin{gathered}
\xymatrix@R=0.5em@C=0.25em{
G_0(c)  \ar[rr]|{\sla}^-{G^1_1 u(c)}&&G_0(c)\\
&{\scriptstyle G^u_{c}}~\rotatebox[origin=c]{90}{$\Rightarrow$}&\\
\ar@{=}[uu] G_0(c)  \ar[rr]|{\sla}^-{uG_0(c)}&&G_0(c) \ar@{=}[uu] \\
&{\scriptstyle u((\zeta_0)_{c})}~\rotatebox[origin=c]{90}{$\Rightarrow$}&\\
\ar[uu]^-{(\zeta_0)_c} F_0(c) \ar[rr]|{\sla}_-{uF_0(c)}&& F_0(c)\ar[uu]_-{(\zeta_0)_c}
}
\end{gathered}
\end{flalign}
\end{subequations}
in $\scrP_1^1$ coincide.
\end{defi}

Using similar compositions as in the case of
pseudo-categories \cite{PseudoCats}, see also \cite[Chapter~12.3]{Yau},
one obtains the following result.
\begin{propo}\label{prop:psMult}
Pseudo-operads (Definition \ref{def:pseudooperad}), pseudo-multifunctors 
(Definition \ref{def:pseudomultifunctor}) and multitransformations
(Definition \ref{def:multitransformation}) assemble into 
a strict $(2,1)$-category $\PsOp$.
\end{propo}

The concept of fibrant pseudo-categories from \cite[Definition 3.4]{Shulman}
admits an immediate generalization to our context of pseudo-operads.
Since every vertical morphism in a pseudo-operad as in Definition \ref{def:pseudooperad}
is invertible, the same observation as in \cite[Lemma 3.20]{Shulman} applies
to our context so that we do not have to introduce the concept of conjoints.
\begin{defi}\label{def:fibrant}
Let $\scrO\in\PsOp$ be a pseudo-operad.
\begin{itemize} 
\item[(a)] A \textit{companion} of a vertical morphism
$g : c\to c^\prime$ is a $1$-ary operation $\hat{g}: c\hto c^\prime$ 
together with $2$-cells
\begin{flalign}
\begin{gathered}
\xymatrix@R=0.25em@C=1em{
c^\prime \ar[rr]|{\sla}^-{u(c^\prime)} && c^\prime\\
&\rotatebox[origin=c]{90}{$\Rightarrow$}&\\
\ar[uu]^-{g}c \ar[rr]|{\sla}_-{\hat{g}} && c^\prime \ar@{=}[uu]
}
\end{gathered}
\qquad\quad\text{and}\qquad\quad
\begin{gathered}
\xymatrix@R=0.25em@C=1em{\\
c \ar[rr]|{\sla}^-{\hat{g}}&& c^\prime\\
&\rotatebox[origin=c]{90}{$\Rightarrow$}&\\
\ar@{=}[uu]c \ar[rr]|{\sla}_-{u(c)}&& c\ar[uu]_-{g}
}
\end{gathered}\quad,
\end{flalign}
such that
\begin{subequations}\label{eqn:companionidentities}
\begin{flalign}\label{eqn:companionidentities1}
\begin{gathered}
\xymatrix@R=0.25em@C=1em{
c^\prime \ar[rr]|{\sla}^-{u(c^\prime)} && c^\prime\\
&\rotatebox[origin=c]{90}{$\Rightarrow$}&\\
\ar[uu]^-{g}c \ar[rr]|{\sla}^-{\hat{g}}&& c^\prime\ar@{=}[uu]\\
&\rotatebox[origin=c]{90}{$\Rightarrow$}&\\
\ar@{=}[uu]
c \ar[rr]|{\sla}_-{u(c)}&& c\ar[uu]_-{g}
}
\end{gathered}
~~=~~
\begin{gathered}
\xymatrix@R=0.25em@C=1em{
c^\prime \ar[rr]|{\sla}^-{u(c^\prime)} && c^\prime\\
&{\scriptstyle u(g)}~\rotatebox[origin=c]{90}{$\Rightarrow$}&\\
\ar[uu]^-{g} c \ar[rr]|{\sla}_-{u(c)}&& c\ar[uu]_-{g}
}
\end{gathered}
\end{flalign}
and
\begin{flalign}\label{eqn:companionidentities2}
\begin{gathered}
\xymatrix@R=0.25em@C=1em{
c \ar[rrrr]|{\sla}^-{\hat{g}} && &&c^\prime\\
&& {\scriptstyle \mathsf{l}_{\hat{g}}}~\rotatebox[origin=c]{90}{$\Rightarrow$}~{\scriptstyle\cong} && \\
\ar@{=}[uu]c \ar[rr]|{\sla}^-{\hat{g}}&& c^\prime \ar[rr]|{\sla}^-{u(c^\prime)} && c^\prime\ar@{=}[uu]\\
&\rotatebox[origin=c]{90}{$\Rightarrow$}&   &\rotatebox[origin=c]{90}{$\Rightarrow$}&  \\
\ar@{=}[uu]c \ar[rr]|{\sla}_-{u(c)}&& c \ar[uu]^-{g} \ar[rr]|{\sla}_-{\hat{g}} && c^\prime\ar@{=}[uu]\\
&& {\scriptstyle \mathsf{r}_{\hat{g}}}~\rotatebox[origin=c]{-90}{$\Rightarrow$}~{\scriptstyle\cong} && \\
\ar@{=}[uu]c \ar[rrrr]|{\sla}_-{\hat{g}}&& &&c^\prime\ar@{=}[uu]
}
\end{gathered}
~~=~~
\begin{gathered}
\xymatrix@R=0.25em@C=1em{
c \ar[rr]|{\sla}^-{\hat{g}}&& c^\prime\\
& {\scriptstyle \id_{\hat{g}}}~\rotatebox[origin=c]{90}{$\Rightarrow$}&\\
\ar@{=}[uu]c \ar[rr]|{\sla}_-{\hat{g}}&& c^\prime\ar@{=}[uu]
}
\end{gathered}\quad.
\end{flalign}
\end{subequations}

\item[(b)] The pseudo-operad $\scrO$ is called \textit{fibrant}
if every vertical morphism has a companion.
We denote by $\PsOp^{\mathrm{fib}}\subseteq\PsOp$ the full $2$-subcategory
of fibrant pseudo-operads. 
\end{itemize}
\end{defi}

In the main part of our paper, we require constructions which allow
us to relate between pseudo-operads and ordinary ($\Set$-valued colored symmetric) operads.
The following result is a direct generalization of \cite[Theorem 2.15]{FFTvAQFT} from 
pseudo-categories to pseudo-operads.
\begin{theo}\label{theo:2adjunction}
The constructions presented below define a $2$-adjunction
\begin{flalign}\label{eqn:tauiotaadjunction}
\xymatrix@C=3.5em{
\tau ~:~ \PsOp^{\mathrm{fib}} \ar@<1.2ex>[r]_-{\perp}~~&~~\ar@<1.2ex>[l] \Op^{(2,1)} ~:~ \iota
}
\end{flalign}
between the $(2,1)$-category $\PsOp^{\mathrm{fib}}$ of fibrant pseudo-operads,
pseudo-multifunctors and multitransformations and the $(2,1)$-category $\Op^{(2,1)}$
of ordinary ($\Set$-valued colored symmetric) operads, multifunctors and multinatural isomorphisms.
\end{theo}

The inclusion $2$-functor $\iota: \Op^{(2,1)}\to \PsOp^{\mathrm{fib}}$
and the truncation $2$-functor $\tau: \PsOp^{\mathrm{fib}}\to \Op^{(2,1)}$
are given explicitly by the following constructions.
\begin{constr}\label{constr:iota}
The inclusion $2$-functor $\iota: \Op^{(2,1)}\to \PsOp^{\mathrm{fib}}$
is defined by the following assignments:
\begin{description}
\item[On objects:] To any ordinary operad $\O\in\Op^{(2,1)}$, we assign the fibrant pseudo-operad 
$\iota(\O)\in \PsOp^{\mathrm{fib}}$ which is defined by the following data as in Definition \ref{def:pseudooperad}:
\begin{itemize}
\item[(i)] The groupoid $\iota(\O)_0$ has as objects the objects of the operad $\O$
and as morphisms all \textit{invertible} $1$-ary operations in $\O$. We denote these morphisms vertically
\begin{flalign}
\begin{gathered}
\xymatrix{
c^\prime\\
\ar[u]^-{g}_-{\cong}c
}
\end{gathered}
\end{flalign}
and often suppress the symbol $\cong$ indicating that these are isomorphisms.

\item[(ii)] For each $n\in\bbN_0$, the groupoid $\iota(\O)_1^n$ has as objects 
the $n$-ary operations of the operad $\O$ and as morphisms all commutative squares
\begin{flalign}
\begin{gathered}
\xymatrix{
\und{c}^\prime \ar[r]^-{\psi^\prime} & d^\prime\\
\ar[u]^-{\und{g}}\und{c}\ar[r]_-{\psi}&d\ar[u]_-{h}
}
\end{gathered}
\end{flalign}
under operadic compositions. The source functor $s^n$ sends such square to the $n$-tuple
of vertical morphisms $\und{g} : \und{c}\to\und{c}^\prime$ and the target functor $t^n$
sends it to the single vertical morphism $h : d\to d^\prime$.

\item[(iii)] The operadic composition functors 
$\cmp : \iota(\O)_1^n\times_{\iota(\O)_0^{\times n}} \iota(\O)_1^{\und{k}}\to \iota(\O)_1^{\Sigma\und{k}}$ 
are defined by operadic composition in the operad $\O$ (denoted below by juxtaposition) 
\begin{flalign}
\begin{gathered}
\xymatrix{
\und{\und{a}}^\prime \ar[r]^-{\und{\phi}^\prime} & \und{c}^\prime \ar[r]^-{\psi^\prime} & d^\prime\\
\ar[u]^-{\und{\und{f}}}\und{\und{a}}\ar[r]_-{\und{\phi} }&\und{c}\ar[u]^-{\und{g}} \ar[r]_-{\psi}& d\ar[u]_-{h}
}
\end{gathered}
~~\stackrel{\cmp}{\longmapsto}~~
\begin{gathered}
\xymatrix@C=3em{
\und{\und{a}}^\prime \ar[r]^-{\psi^\prime\,\und{\phi}^\prime} & d^\prime\\
\ar[u]^-{\und{\und{f}}}\und{\und{a}}\ar[r]_-{\psi\, \und{\phi}} & d\ar[u]_-{h}
}
\end{gathered} \quad .
\end{flalign}

\item[(iv)] The operadic unit functor $u : \iota(\O)_0\to \iota(\O)_1^1$ assigns the units 
$\oone$ of the operad $\O$
\begin{flalign}
\begin{gathered}
\xymatrix{
c^\prime\\
\ar[u]^-{g}c
}
\end{gathered}
~~\stackrel{u}{\longmapsto}~~
\begin{gathered}
\xymatrix{
c^\prime \ar[r]^-{\oone_{c^\prime}}& c^\prime\\
\ar[u]^-{g}c \ar[r]_-{\oone_c}& c\ar[u]_-{g}
}
\end{gathered} \quad .
\end{flalign}

\item[(v)] For each $n\in\bbN_0$, the permutation action
$\iota(\O)_1^n: \mathbf{B} S_n^\op\to\Grpd$ is defined by the permutation
action of the operad $\O$
\begin{flalign}
\begin{gathered}
\xymatrix{
\und{c}^\prime \ar[r]^-{\psi^\prime} & d^\prime\\
\ar[u]^-{\und{g}}\und{c}\ar[r]_-{\psi}&d\ar[u]_-{h}
}
\end{gathered}
~~\stackrel{\iota(\O)_1^n(\sigma)}{\longmapsto}~~
\begin{gathered}
\xymatrix@C=3.5em{
\und{c}^\prime\sigma \ar[r]^-{\O(\sigma)(\psi^\prime)} & d^\prime\\
\ar[u]^-{\und{g}\sigma}\und{c}\sigma\ar[r]_-{\O(\sigma)(\psi)}&d\ar[u]_-{h}
}
\end{gathered}\quad,
\end{flalign}
for all $\sigma\in S_n$.

\item[(vi)] The associator $\mathsf{a}$ and the unitors $\mathsf{l}$ and $\mathsf{r}$ 
are trivial, i.e.\ they consist of the identity natural isomorphisms.
\end{itemize}

Note that the pseudo-operad $\iota(\O)$ is indeed fibrant in the sense of Definition \ref{def:fibrant}.
The companion of a vertical morphism $g : c\to c^\prime$ is the
$1$-ary operation $\hat{g} \coloneqq g : c\to c^\prime$.

\item[On morphisms:] To any ordinary multifunctor $F : \O\to \P$ in $\Op^{(2,1)}$, 
we assign the pseudo-multifunctor $\iota(F): \iota(\O) \to \iota(\P)$ in $\PsOp^{\mathrm{fib}}$ 
which is defined by the following data as in Definition \ref{def:pseudomultifunctor}:
\begin{itemize}
\item[(i)] The functor $\iota(F)_0 : \iota(\O)_0\to\iota(\P)_0$ is given by restricting
the multifunctor $F:\O\to \P$ to the wide subgroupoids of invertible $1$-ary operations.

\item[(ii)] For each $n\in\bbN_0$, the functor $\iota(F)_1^n : \iota(\O)_1^n\to\iota(\P)_1^n$
is defined in terms of the multifunctor $F:\O\to \P$ by
\begin{flalign}
\begin{gathered}
\xymatrix{
\und{c}^\prime \ar[r]^-{\psi^\prime} & d^\prime\\
\ar[u]^-{\und{g}}\und{c}\ar[r]_-{\psi}&d\ar[u]_-{h}
}
\end{gathered}
~~\stackrel{\iota(F)_1^n}{\longmapsto}~~
\begin{gathered}
\xymatrix@C=3.5em{
F(\und{c}^\prime) \ar[r]^-{F(\psi^\prime)} & F(d^\prime) \\
\ar[u]^-{F(\und{g})}F(\und{c}) \ar[r]_-{F(\psi)}&F(d)\ar[u]_-{F(h)}
}
\end{gathered}\quad.
\end{flalign}

\item[(iii)] The natural isomorphisms $\iota(F)^{\cmp}$ and $\iota(F)^u$ are the identities.
\end{itemize}

\item[On $2$-morphisms:] To any ordinary multinatural isomorphism 
$\zeta : F \Rightarrow G : \O\to \P$ in $\Op^{(2,1)}$, 
we assign the multitransformation $\iota(\zeta): \iota(F)\Rightarrow \iota(G): 
\iota(\O) \to \iota(\P)$ in $\PsOp^{\mathrm{fib}}$ 
which is defined by the following data as in Definition \ref{def:multitransformation}:
\begin{itemize}
\item[(i)] The natural transformation
$\iota(\zeta)_0 : \iota(F)_0 \Rightarrow \iota(G)_0: \iota(\O)_0 \to \iota(\P)_0$ is defined by the components
\begin{flalign}
\iota(\zeta)_0 \,\coloneqq\, \left\{
\begin{gathered}
\xymatrix@R=0.5em@C=0.25em{
G(c)\\
~\\
\ar[uu]^-{\zeta_c} F(c)
}
\end{gathered}~~:~~c\in \O
\right\}\quad.
\end{flalign}

\item[(ii)] For each $n\in\bbN_0$, the natural transformation
$\iota(\zeta)_1^n : \iota(F)_1^n \Rightarrow \iota(G)_1^n: \iota(\O)_1^n \to \iota(\P)_1^n$ 
is defined by the components
\begin{flalign}
\iota(\zeta)_1^n \,\coloneqq\, \left\{
\begin{gathered}
\xymatrix@R=0.5em@C=0.5em{
G(\und{c})\ar[rr]^-{G(\psi)}&&G(d)\\
&~&\\
\ar[uu]^-{\zeta_{\und{c}}} F(\und{c}) \ar[rr]_-{F(\psi)}&& F(d)\ar[uu]_-{\zeta_{d}}
}
\end{gathered}~~:~~\text{$n$-ary }\big(\psi: \und{c}\to d\big)\in \O
\right\}\quad.
\end{flalign}
\end{itemize}
\end{description}
The assignment $\iota: \Op^{(2,1)}\to \PsOp^{\mathrm{fib}}$ defined above is strictly $2$-functorial.
\end{constr}

\begin{constr}\label{constr:tau}
The truncation $2$-functor $\tau: \PsOp^{\mathrm{fib}}\to \Op^{(2,1)}$
is defined by the following assignments:
\begin{description}
\item[On objects:] To any fibrant pseudo-operad $\scrO\in\PsOp^{\mathrm{fib}}$, we assign
the ordinary operad $\tau(\scrO)\in\Op^{(2,1)}$ which is defined by the following data:
\begin{itemize}
\item[(i)] The objects of the operad $\tau(\scrO)$ are the objects of the groupoid $\scrO_0$.

\item[(ii)] The $n$-ary operations of the operad $\tau(\scrO)$ are equivalence classes 
$[\psi : \und{c}\hto d]\in\scrO_1^n/_{\sim}$ of the $n$-ary operations in $\scrO_1^n$
under the following equivalence relation: Two $n$-ary operations
$(\psi:\und{c}\hto d),(\psi^\prime: \und{c}\hto d)\in\scrO_1^n$ with the same
source and target are equivalent if there exists a globular $2$-cell
\begin{flalign}
\begin{gathered}
\xymatrix@R=0.25em@C=1em{
\und{c} \ar[rr]|{\sla}^-{\psi^\prime}&& d\\
&~\rotatebox[origin=c]{90}{$\Rightarrow$}&\\
\ar@{=}[uu]\und{c} \ar[rr]|{\sla}_-{\psi}&& d\ar@{=}[uu]
}
\end{gathered}\quad,
\end{flalign}
which is automatically an isomorphism because $\scrO_1^n$ is a groupoid.

\item[(iii)] Operadic composition in $\tau(\scrO)$ is defined
by operadic composition $[\psi]\,[\und{\phi}]\coloneqq \big[\psi\cmp\und{\phi}\big]$
of any choice of representatives in the pseudo-operad $\scrO$ and the operadic
units in $\tau(\scrO)$ are given by $\oone_c\coloneqq[u(c)]$. Associativity
and unitality of compositions in $\tau(\scrO)$ hold true strictly
because the associator $\mathsf{a}$ and the unitors $\mathsf{l}$ and $\mathsf{r}$
in Definition \ref{def:pseudooperad} are by hypothesis globular, hence they are trivial
at the level of equivalence classes.

\item[(iv)] The permutation actions in $\tau(\scrO)$ are defined in terms
of the permutation actions in $\scrO$ by
$\tau(\scrO)(\sigma)[\psi : \und{c}\hto d] \coloneqq \big[\scrO_1^n(\sigma)(\psi) : \und{c}\sigma\hto d\big]$,
for all $\sigma\in S_n$.
\end{itemize}

\item[On morphisms:] To any pseudo-multifunctor $F : \scrO\to \scrP$ in $\PsOp^{\mathrm{fib}}$, 
we assign the ordinary multifunctor in $\Op^{(2,1)}$ which is defined by 
\begin{flalign}
\tau(F)\,:\, \tau(\scrO) ~&\longrightarrow~ \tau(\scrP)\quad,\\
\nn \scrO_0\ni c~&\longmapsto~F_0(c)\in\scrP_0\quad,\\
\nn \scrO_1^n/_{\sim} \ni \big[\psi:\und{c}\hto d\big]~&\longmapsto~\big[F_1^n(\psi): F_0(\und{c})\hto F_0(d)\big]\in \scrP_1^n/_{\sim}\quad.
\end{flalign}

\item[On $2$-morphisms:] To any  multitransformation
$\zeta : F \Rightarrow G : \scrO\to \scrP$ in $\PsOp^{\mathrm{fib}}$, 
we assign the ordinary multinatural isomorphism $\tau(\zeta): \tau(F)\Rightarrow \tau(G): 
\tau(\scrO) \to \tau(\scrP)$ in $\Op^{(2,1)}$ which is defined by the components
\begin{flalign}
\tau(\zeta)\,\coloneqq\, \left\{\big[(\hat{\zeta}_0)_c : F_0(c)\hto G_0(c) \big]~~:~~ c\in \scrO_0\right\}
\end{flalign}
that are obtained by any choice of companions (see Definition \ref{def:fibrant})
for the components $(\zeta_0)_c : F_0(c)\to G_0(c)$ of $\zeta_0$. 
By \cite[Lemma 3.8]{Shulman}, different choices of companions define the 
same equivalence class, hence the components of $\tau(\zeta)$ are well-defined.
To prove that these components are multinatural, i.e.\ 
$[G_1^n(\psi)]\,[(\hat{\zeta}_0)_{\und{c}}] = [(\hat{\zeta}_0)_{d}] \, [F_1^n(\psi)]$ 
for all $n$-ary operations $(\psi : \und{c}\hto d)\in \scrO_1^n$, 
we compose the $2$-cell component $(\zeta_1^n)_\psi$ of the natural transformation 
$\zeta_1^n : F_1^n\Rightarrow G_1^n : \scrO_1^n \to \scrP_1^n$ 
with the $2$-cells for the companions from Definition~\ref{def:fibrant} 
according to
\begin{flalign}
\begin{gathered}
\xymatrix@R=0.25em@C=1em{
F_0(\und{c}) \ar[rr]|{\sla}^-{(\hat{\zeta}_0)_{\und{c}}} && G_0(\und{c}) \ar[rr]|{\sla}^-{G_1^n(\psi)}&& 
G_0(d)  \ar[rr]|{\sla}^-{u(G_0(d))}&& G_0(d)\\
&\rotatebox[origin=c]{90}{$\Rightarrow$}& &{\scriptstyle (\zeta_1^n)_\psi}~\rotatebox[origin=c]{90}{$\Rightarrow$}& &\rotatebox[origin=c]{90}{$\Rightarrow$}& \\
\ar@{=}[uu] F_0(\und{c}) \ar[rr]|{\sla}_-{u(F_0(\und{c}))} && F_0(\und{c}) \ar[uu]^-{(\zeta_0)_{\und{c}}}
\ar[rr]|{\sla}_-{F_1^n(\psi)}&&  \ar[uu]_-{(\zeta_0)_{d}} F_0(d)\ar[rr]|{\sla}_-{(\hat{\zeta}_0)_{d}} && G_0(d)\ar@{=}[uu] 
}
\end{gathered}\qquad.
\end{flalign}
These globular $2$-cells exhibit the naturality of $\tau(\zeta)$
by passing to equivalence classes.
\end{description}
The assignment $\tau:  \PsOp^{\mathrm{fib}}\to \Op^{(2,1)}$ defined above is strictly $2$-functorial.
\end{constr}

\begin{proof}[Proof of Theorem \ref{theo:2adjunction}]
We have to exhibit a unit $\eta: \id_{\PsOp^{\mathrm{fib}}}\Rightarrow\iota\,\tau$ and a counit
$\epsilon: \tau\,\iota\Rightarrow\id_{\Op^{(2,1)}} $ for the $2$-adjunction \eqref{eqn:tauiotaadjunction}. 
Using the explicit expressions from Constructions~\ref{constr:iota} and~\ref{constr:tau},
one directly checks that $\tau\,\iota = \id_{\Op^{(2,1)}}$, so we will choose for the counit
$\epsilon\coloneqq \mathrm{Id}$ the identity $2$-natural transformation. In order to define
the unit, consider for each $\scrO\in\PsOp^{\mathrm{fib}}$ 
the component $\eta_{\scrO} : \scrO \Rightarrow\iota\,\tau(\scrO)$
pseudo-multifunctor which is defined
by the following data as in Definition \ref{def:pseudomultifunctor}:
\begin{itemize}
\item[(i)] The functor
\begin{flalign}
 (\eta_{\scrO})_0 \,: \scrO_0~&\longrightarrow~\iota\tau(\scrO)_0\quad,\\
\nn c~&\longmapsto~ c\quad,\\
\nn (g:c\to c^\prime)~&\longmapsto~[\hat{g} : c\hto c^\prime] \,=:\, \big([\hat{g}] : c\to c^\prime\big)\quad.
\end{flalign}

\item[(ii)] For each $n\in\bbN_0$, the functor
\begin{flalign}
 (\eta_{\scrO})_1^n \,: \scrO_1^n~&\longrightarrow~\iota\tau(\scrO)_1^n\quad,\\
\nn (\psi:\und{c} \hto d)~&\longmapsto~ [\psi :\und{c}\hto d]\,=:\, \big([\psi] : \und{c} \to d\big)\quad,\\
\nn \begin{gathered}
\xymatrix@R=0.25em@C=0.25em{
\und{c}^\prime \ar[rr]|{\sla}^-{\psi^\prime}&& d^\prime\\
&{\scriptstyle \alpha}~\rotatebox[origin=c]{90}{$\Rightarrow$}&\\
\ar[uu]^-{\und{g}} \und{c}  \ar[rr]|{\sla}_-{\psi}  && d\ar[uu]_-{h} 
}
\end{gathered}
~~&\longmapsto~~
\begin{gathered}
\xymatrix@R=0.75em@C=1em{
\und{c}^\prime \ar[rr]^-{[\psi^\prime]}&& d^\prime\\
&& \\
\ar[uu]^-{[\und{\hat{g}}]} \und{c}  \ar[rr]_-{[\psi]}  && d\ar[uu]_-{[\hat{h}]} 
}
\end{gathered}\quad.
\end{flalign}
Note that commutativity of the square in $\tau(\scrO)$ follows by composing the $2$-cell
$\alpha$ with the $2$-cells for the companions from Definition~\ref{def:fibrant} 
according to
\begin{flalign}
\begin{gathered}
\xymatrix@R=0.25em@C=0.5em{
\und{c} \ar[rr]|{\sla}^-{\und{\hat{g}}}&& \und{c}^\prime \ar[rr]|{\sla}^-{\psi^\prime}&& d^\prime \ar[rr]|{\sla}^-{u(d^\prime)}&& d^\prime \\
&~\rotatebox[origin=c]{90}{$\Rightarrow$}~&  &{\scriptstyle \alpha}~\rotatebox[origin=c]{90}{$\Rightarrow$}&  &~\rotatebox[origin=c]{90}{$\Rightarrow$}~&\\
\ar@{=}[uu] \und{c} \ar[rr]|{\sla}_-{u(\und{c})}&& \ar[uu]^-{\und{g}} \und{c}  \ar[rr]|{\sla}_-{\psi} 
 && d\ar[uu]_-{h} 
\ar[rr]|{\sla}_-{\hat{h}}&& d^\prime\ar@{=}[uu]
}
\end{gathered}
\end{flalign}
and then passing to equivalence classes.
\item[(iii)] The natural isomorphisms $(\eta_{\scrO})^{\cmp}$ and $(\eta_{\scrO})^u$ are the identities.
\end{itemize}
One checks that these components define a $2$-natural transformation
$\eta: \id_{\PsOp^{\mathrm{fib}}}\Rightarrow\iota\,\tau$. 
\sk

It remains to verify the triangle identities for the unit and counit. 
Using $\epsilon = \mathrm{Id}$, these identities
reduce to verifying that $\eta_{\iota(\O)}=\id$, for all ordinary operads $\O\in\Op^{(2,1)}$,
and that $\tau(\eta_{\scrO}) = \id$, for all fibrant pseudo-operads $\scrO\in\PsOp^{\mathrm{fib}}$.
These are elementary checks which follow from the explicit formulas above.
\end{proof}


\section{\label{app:Lorentzian}Lorentzian geometric details}
In this appendix we collect some technical results of a 
Lorentzian geometric nature which are required in this work.
We refer the reader to \cite{ONeill,BGP,Minguzzi} for the relevant terminology
and background in Lorentzian geometry. However, let us recall the following
standard notations which will be used in the proofs below: 
Given any time-oriented Lorentzian manifold $M$ and two points $p,q\in M$,
one writes $p\ll q$ if there exists a future-pointing timelike curve 
from $p$ to $q$ and $p<q$ if there exists a future-pointing causal curve from $p$ to $q$.
The symbol $p\leq q$ means that either $p=q$ or $p<q$.
\begin{lem}\label{lem:interior}
Consider any object $M\in\Loc_m$ and any causally convex subset $A\subseteq M$.
Then the interior $\mathrm{int}_M(A)\subseteq M$ is a causally convex open subset.
\end{lem}
\begin{proof}
Given any future-pointing causal curve $\gamma:[0,1]\to M$ with endpoints 
$\gamma(0),\gamma(1)\in  \mathrm{int}_M(A)$, causal convexity of 
$A\subseteq M$ implies that $\gamma(s)\in  A$, for all $s\in[0,1]$.
We will now exhibit, for each $s\in (0,1)$, an open neighborhood $U$ of $\gamma(s)$
which is contained in $A$. This implies that $\gamma(s)\in  \mathrm{int}_M(A)$
lies in the interior, for all $s\in[0,1]$, and hence $\mathrm{int}_M(A)\subseteq M$ is causally convex.
\sk

Since $\gamma(0),\gamma(1)\in  \mathrm{int}_M(A)$
are contained in an open subset, there exists by \cite[Chapter 14, Lemma 3]{ONeill} two points $p_-,p_+\in A$
with $p_-\ll \gamma(0)$ and $\gamma(1) \ll p_+$. 
Fixing any $s\in(0,1)$ and using \cite[Theorem 2.24]{Minguzzi}, we deduce from
$p_- \ll \gamma(0)\leq \gamma(s)\leq \gamma(1)\ll p_+$ that $p_-\ll \gamma(s) \ll p_+$.
This implies that there exists a future-pointing \textit{timelike} curve
$\delta : [0,1]\to M\,,~t\mapsto \delta(t)$ with $\delta(0) = \gamma(0)$,
$\delta(1) = \gamma(1)$ and $\delta(\lambda) = \gamma(s)$, for some $\lambda\in (0,1)$.
By using the causal convexity of $A\subseteq M$ once more,
it follows that $\delta(t)\in A$, for all $t\in[0,1]$.
Choosing any $\lambda_-,\lambda_+\in(0,1)$ such that $\lambda_-<\lambda<\lambda_+$,
the subset $U\coloneqq I^{+}_M\big(\delta(\lambda_-)\big)\cap I^{-}_M\big(\delta(\lambda_+)\big)\subseteq M$
is non-empty, because $\delta(\lambda)\in U$ is by construction a point, 
and it is open by \cite[Chapter 14, Lemma 3]{ONeill}.
From causal convexity of $A\subseteq M$
and $\delta(\lambda_-),\delta(\lambda_+)\in A$, it follows that $U\subseteq A$,
hence $U$ provides an open neighborhood of $\gamma(s)=\delta(\lambda)$ which is contained in $A$.
\end{proof}

\begin{lem}\label{lem:IandJ}
Consider any object $M\in\Loc_m$ and any subset $A\subseteq M$. Then
\begin{subequations}
\begin{flalign}\label{eqn:IMA}
I_M^\pm(A)\,\subseteq\, M
\end{flalign}
and 
\begin{flalign}\label{eqn:JMA}
M\setminus \mathrm{cl}_M\big(J^\pm_M(A)\big)\,\subseteq\,M
\end{flalign}
\end{subequations}
are causally convex open subsets.
\end{lem}
\begin{proof}
We begin by observing that
\begin{flalign}
I_M^\pm(A)\,=\,\mathrm{int}_M\big(J_M^\pm(A)\big)\quad,\qquad
M\setminus \mathrm{cl}_M\big(J^\pm_M(A)\big) 
\,=\, \mathrm{int}_{M}\Big(M\setminus J^\pm_M(A)\Big)\quad,
\end{flalign}
where the first equality follows from \cite[Chapter 14, Lemma 6]{ONeill}.
Using Lemma \ref{lem:interior}, it therefore suffices to show that
$J_M^\pm(A)\subseteq M$ and $M\setminus J^\pm_M(A)\subseteq M$ are causally convex.
This is evident for $J_M^\pm(A) \subseteq M$, by definition of the causal future/past,
so it remains to show that $M\setminus J^\pm_M(A)\subseteq M$ is causally convex.
It suffices to consider the case of removing the causal future $M\setminus J^+_M(A)\subseteq M$
since the other case then follows by reversing the time-orientation.
\sk

Let $\gamma:[0,1]\to  M$ be a future-pointing causal curve 
with endpoint $\gamma(1)\in M\setminus J^+_M(A)$, 
i.e.\ $\gamma(1)\not\in J^+_M(A)$. Then $\gamma(s)\not\in J^+_M(A)$, for all $s\leq 1$,
since otherwise there exists $p\in A$ with $p\leq \gamma(s)$,
which implies that $p\leq \gamma(1)$ as a consequence of $\gamma(s)\leq \gamma(1)$, a contradiction with
$\gamma(1)\not\in J^+_M(A)$. Hence, $\gamma(s)\in M\setminus J^+_M(A)$, for all $s\in[0,1]$, which shows
causal convexity of the subset $M\setminus J^+_M(A)\subseteq M$.
\end{proof}

\begin{lem}\label{lem:surfaces}
Consider any positive integer $n\in\bbN$, any object $M\in \Loc_m$ 
and any family of Cauchy surfaces $\Sigma_i\subset M$,
for $i=1,\dots,n$. Then there exists another
Cauchy surface $\Sigma\subset M$ of $M$ 
which is contained in the common chronological future/past 
of this family of Cauchy surfaces, i.e.\ $\Sigma\subset \bigcap_{i=1}^n I^\pm_M(\Sigma_i)$.
This implies that the subset inclusion $\bigcap_{i=1}^n I^\pm_M(\Sigma_i)\subseteq M$ is Cauchy.
\end{lem}
\begin{proof}
It suffices to consider the case of the chronological future $I^+_M$
since the other case follows by reversing the time-orientation.
Since $\bigcap_{i=1}^n I^+_M(\Sigma_i)\subseteq M$ is a finite intersection of
causally convex open subsets, it is causally convex and open, hence globally hyperbolic. 
We choose any Cauchy surface $\Sigma\subset \bigcap_{i=1}^n I^+_M(\Sigma_i)$
and show that this defines a Cauchy surface $\Sigma \subset M$ of $M$.
Given any inextensible future-pointing timelike curve $\gamma : \bbR\to M$, 
it intersects the $i$-th Cauchy surface $\Sigma_i\subset M$ exactly once, say at $s_i\in \bbR$, 
for each $i=1,\dots,n$. The restriction
$\gamma\vert : \big(\max\{s_i\}_{i=1,\dots,n},\infty\big) \to \bigcap_{i=1}^n I^+_M(\Sigma_i) $ 
is an inextensible timelike curve in $\bigcap_{i=1}^n I^+_M(\Sigma_i)$, 
so it intersects the Cauchy surface $\Sigma\subset \bigcap_{i=1}^n I^+_M(\Sigma_i)$
exactly once. This implies that any inextensible future-pointing timelike curve
$\gamma : \bbR\to M$ intersects $\Sigma\subset M$ exactly once,
hence $\Sigma\subset M$ is a Cauchy surface of $M$.
\end{proof}

\begin{propo}\label{prop:bordismgluingregions}
Let $(N,\und{\iota_{0}},\iota_1) : (\und{M_0},\und{\Sigma_0})\hto (M_1,\Sigma_1)$
be any $n$-ary operation in the globally hyperbolic Lorentzian bordism
pseudo-operad $\LBop_m$ from Section \ref{sec:LBord}, represented by the zig-zags
\begin{flalign}
\xymatrix{
\und{M_0} & \ar[l]_-{\subseteq} \und{V_0}  \ar[r]^-{\und{\iota_0}} & N &\ar[l]_-{\iota_1} V_1 \ar[r]^-{\subseteq} & M_1   
}
\end{flalign}
of operations in the operad $\P_{\Loc_m^\perp}$ from Definition \ref{def:PLoc}.
\begin{itemize}
\item[(a)] Let $U_{i}\subseteq V_{0_i}\subseteq M_{0_i}$ be any family of
causally convex open subsets which contain the Cauchy surfaces $\Sigma_{0_i}\subset U_{i}$,
for all $i=1,\dots,n$. Then 
\begin{flalign}\label{eqn:unionTMP}
\Bigg(N \setminus \mathrm{cl}_{N}\bigg(\bigcup_{i=1}^n J^-_{N}\big(\iota_{0_i}(\Sigma_{0_i})\big)\bigg)\Bigg) \cup 
\bigcup_{i=1}^n \iota_{0_i}(U_{i})\,\subseteq\,N
\end{flalign}
is a causally convex open subset which contains the images 
$\bigcup_{i=1}^n \iota_{0_i}(\Sigma_{0_i}) \cup \iota_1(\Sigma_1)
\subset N$ of the Cauchy surfaces.

\item[(b)] Let $U\subseteq V_1\subseteq M_1$ be any causally convex open subset which 
contains the Cauchy surface $\Sigma_1\subset U$. Then
\begin{flalign}
J_{N}^-\big(\iota_{1}(U)\big)\,\subseteq\, N
\end{flalign}
is a causally convex open subset which contains the images 
$\bigcup_{i=1}^n \iota_{0_i}(\Sigma_{0_i}) \cup \iota_1(\Sigma_1)
\subset N$ of the Cauchy surfaces.
\end{itemize}
\end{propo}
\begin{proof}
Let us start with the simpler item (b). Since $U\subseteq V_1\subseteq M_1$
is causally convex open, we have that $\iota_{1}(U)\subseteq N$ is causally convex open.
Then \cite[Chapter 14, Corollary 1]{ONeill} implies that 
$J_{N}^-\big(\iota_{1}(U)\big) = I_{N}^-\big(\iota_{1}(U)\big)\subseteq N$, which is a
causally convex open subset by Lemma \ref{lem:IandJ}. The statement
about the images of the Cauchy surfaces follows from the hypothesis that
$\Sigma_1\subset U$ and the conditions in \eqref{eqn:bordismsurfacesI} and \eqref{eqn:bordismsurfacesII}.
\sk

To show item (a), let us first observe that 
\begin{flalign}
N \setminus \mathrm{cl}_{N}\bigg(\bigcup_{i=1}^n J^-_{N}\big(\iota_{0_i}(\Sigma_{0_i})\big)\bigg)\,=\,
N \setminus \mathrm{cl}_{N}\bigg(J^-_{N}\bigg(\bigcup_{i=1}^n \iota_{0_i}(\Sigma_{0_i})\bigg)\bigg)\,\subseteq \,N
\end{flalign}
is causally convex open by Lemma \ref{lem:IandJ} and that 
$\bigcup_{i=1}^n \iota_{0_i}(U_{i})\subseteq N$ is causally convex open
because each $\iota_{0_i}(U_{i})\subseteq N$ is causally convex open and
$\iota_{0_i}(U_{i})\perp \iota_{0_j}(U_{j})$ are causally disjoint in $N$, for all $i\neq j$.
To show that also the union \eqref{eqn:unionTMP} of these two subsets
is causally convex open, consider the causally convex open subset
$\bigcup_{i=1}^n \iota_{0_i}\big(I^+_{U_{i}}(\Sigma_{0_i})\big)\subseteq N$
which is contained in the intersection of the two subsets.
Since the inclusion $\bigcup_{i=1}^n \iota_{0_i}\big(I^+_{U_{i}}(\Sigma_{0_i})\big)\subseteq 
\bigcup_{i=1}^n \iota_{0_i}(U_{i})$ is Cauchy, it then follows from the argument
in \cite[Lemma B.1]{HKstack} that \eqref{eqn:unionTMP} is causally convex open.
The statement about the images of the Cauchy surfaces follows from the hypothesis that
$\Sigma_{0_i}\subset U_i$, for all $i=1,\dots,n$,
and the conditions in \eqref{eqn:bordismsurfacesI} and \eqref{eqn:bordismsurfacesII}.
\end{proof}

\begin{propo}\label{prop:pushout} 
The pushout $N_0^- \sqcup_{V_{01}\cap V_{10}}^{~} N_1^+$ 
in \eqref{eqn:bordismcomposition} exists as an object in $\Loc_m$.
\end{propo}
\begin{proof}
First, we have to show that this pushout is a manifold, which we will do by verifying
the criterion from \cite[Lemma 2.23]{StolzTeichner}. This criterion states that for 
a cospan $X \stackrel{f}{\leftarrow} U \stackrel{g}{\to} Y$ of open embeddings
of manifolds, the pushout $X\sqcup_{U}Y$ of topological spaces is canonically a 
manifold if and only if the image of the map $(f,g) : U\to X\times Y$ is a closed subset, i.e.\
$(f,g)(U)\subseteq X\times Y$ contains all its boundary points. 
Note that the set of boundary points $\partial_{X\times Y} (f,g)(U)\subseteq X\times Y$
must be contained in the closed subset $\mathrm{cl}_X(f(U))\times \mathrm{cl}_Y(g(U))\subseteq X\times Y$.
Recalling that $f$ and $g$ are open embeddings, one
checks by using suitable open neighborhoods of points in $U$
that the set of boundary points in $f(U)\times g(U)$ is precisely the image of
$(f,g) : U\to X\times Y$ and that there do not exist boundary points in $\partial_X(f(U))\times g(U)$
and in $f(U)\times \partial_Y(g(U))$.
Hence, the criterion in \cite[Lemma 2.23]{StolzTeichner}
for the pushout $X\sqcup_{U}Y$ to be a manifold is equivalent to verifying
that the subset $\partial_X(f(U))\times \partial_Y(g(U))\subseteq X\times Y$ does not contain any boundary 
points of the image $(f,g)(U)\subseteq X\times Y$.
\sk

Using this criterion, we can now verify that $N_0^- \sqcup_{V_{01}\cap V_{10}}^{~} N_1^+$ 
is a manifold. By construction of the subsets $N_1^+\subseteq N_1$ and
$N_0^-\subseteq N_0$ in \eqref{eqn:overhangingremoved}, it follows that the sets of boundary points
\begin{subequations}
\begin{flalign}
\partial_{N_1^+}\Big(\iota_{10}\big(V_{01}\cap V_{10}\big)\Big) \,\subseteq\, 
N_1 \setminus \mathrm{cl}_{N_1}\bigg(\bigcup_{i=1}^n J^-_{N_1}\big(\iota_{10_i}(\Sigma_{1_i})\big)\bigg)\,\subseteq \,N_1^+
\end{flalign}
and 
\begin{flalign}
\partial_{N_0^-}\Big(\iota_{01}\big(V_{01}\cap V_{10}\big)\Big) \,\subseteq\, 
I^-_{N_0^-}\Big(\iota_{01}\big(V_{01}\cap V_{10}\big)\Big)\,\subseteq\, N_0^-
\end{flalign}
\end{subequations}
are contained in open subsets whose preimages
under $N_0^-\stackrel{\iota_{01}}{\leftarrow} V_{01}\cap V_{10} 
\stackrel{\iota_{10}}{\to} N_1^+$ in  $V_{01}\cap V_{10}$ are disjoint. 
(For a pictorial visualization see \eqref{eqn:bordismremoved} and note that the Cauchy surfaces 
of the gray regions separate the two preimages.) 
This implies that $\partial_{N_0^-}\big(\iota_{01}\big(V_{01}\cap V_{10}\big)\big) \times 
\partial_{N_1^+}\big(\iota_{10}\big(V_{01}\cap V_{10}\big)\big)$ does not contain
any boundary points of the image $(\iota_{01},\iota_{10})\big(V_{01}\cap V_{10}\big)\subseteq N_0^-\times N_1^+$
and hence $N_0^- \sqcup_{V_{01}\cap V_{10}}^{~} N_1^+$  is a manifold.
\sk

Since the maps in the pushout are $\Loc_m$-morphisms, the manifold
$N_0^- \sqcup_{V_{01}\cap V_{10}}^{~} N_1^+$ can be endowed with an
orientation, a time-orientation and a Lorentzian metric which are canonically induced from
the ones of the objects $N_0^-\in\Loc_m$ and $N_1^+\in\Loc_m$. Global hyperbolicity
follows from the observation that the subset $\iota_+\iota_{11}(\Sigma_{2})\subseteq
N_0^- \sqcup_{V_{01}\cap V_{10}}^{~} N_1^+$ which is obtained from the Cauchy surface
$\iota_{11}(\Sigma_{2})\subset N_1^+$ is met exactly once by every inextensible
future-pointing timelike curve, hence it defines a Cauchy surface for $N_0^- \sqcup_{V_{01}\cap V_{10}}^{~} N_1^+$.
\end{proof}

\begin{propo}\label{prop:RCMSigmafiltered}
For each object $(M,\Sigma)\in\tau(\LBop_m)$, the 
category $\RC_{(M,\Sigma)}$ from Definition~\ref{def:RCMSigma} is filtered.
\end{propo}
\begin{proof}
Since $\RC_{(M,\Sigma)}$ is a thin category, it is filtered if and only if
it is non-empty and directed. In this proof we will use the notation
$\Sigma^- \coloneqq I^-_M(\Sigma)\subseteq M$ for the chronological past of $\Sigma$.
\sk

To show that the category $\RC_{(M,\Sigma)}$
is non-empty, we have to construct an object $(U,S)\in \RC_{(M,\Sigma)}$. 
Choosing any point $p\in \Sigma^-$ and any second point 
$q\in I^-_{\Sigma^-}(p)\subseteq \Sigma^-$
in the chronological past of $p$, the intersection $U \coloneqq I^-_{\Sigma^-}(p)\cap 
I^+_{\Sigma^-}(q) \subseteq \Sigma^-$ is non-empty and
it is a relatively compact causally convex open subset by \cite[Lemma A.5.12]{BGP}.
Choosing any Cauchy surface $S\subset U$ defines an object $(U,S)\in \RC_{(M,\Sigma)}$. 
\sk

To show that $\RC_{(M,\Sigma)}$ is directed, we have to construct for any two objects 
$(U_1,S_1),(U_2,S_2)\in \RC_{(M,\Sigma)}$
a third object $(U,S)\in\RC_{(M,\Sigma)}$ and two morphisms $(U_1,S_1)\to (U,S)\leftarrow (U_2,S_2)$.
By Definition \ref{def:RCMSigma}, the subset
\begin{flalign}
\cl_{\Sigma^-}(U_1)\cup \cl_{\Sigma^-}(U_2)\,\subseteq\, \Sigma^-
\end{flalign}
is compact, hence it admits a finite cover $\{V_i\subseteq \Sigma^-:i=1,\dots,n\}$
by relatively compact open subsets $V_i\subseteq \Sigma^-$. We define
\begin{flalign}\label{eqn:Uhull}
U\,\coloneqq\,J^+_{\Sigma^-}\bigg(\bigcup_{i=1}^n V_i\bigg)\cap J^-_{\Sigma^-}\bigg(\bigcup_{i=1}^n V_i\bigg)\,\subseteq\, \Sigma^-
\end{flalign}
to be the causally convex hull of the union of this relatively compact open cover, 
which by \cite[Chapter 14, Corollary 1]{ONeill} 
and \cite[Lemma B.4]{HKstack} is a relatively compact causally convex open subset.
By construction, we have that both $\cl_{\Sigma^-}(S_1)\subset U$ and 
$\cl_{\Sigma^-}(S_2)\subset U$ are achronal compact subsets, hence
they extend by \cite[Proposition 3.6]{BernalSanchez2} to two Cauchy surfaces
$\Sigma_1\subset U$ and $\Sigma_2\subset U$. Using Lemma \ref{lem:surfaces},
we can find another Cauchy surface $S\subset U$ which lies in the chronological
future $S\subset I_U^+(\Sigma_1)\cap I_U^+(\Sigma_2)$ of these two Cauchy surfaces. 
This defines an object $(U,S)\in\RC_{(M,\Sigma)}$. The morphism
$(U_1,S_1)\to (U,S)$ exists by the following argument: The subset inclusion $U_1\subseteq U$
holds true by construction \eqref{eqn:Uhull}
and $\mathrm{cl}_{U}(S_1)= \mathrm{cl}_{\Sigma^-}(S_1)\subset I^-_U(S)$
is a consequence of our particular choice of Cauchy surface $S\subset U$.
By the same argument one shows that the morphism $(U_2,S_2)\to (U,S)$ exists.
\end{proof}

\begin{lem}\label{lem:finality}
The forgetful functor
$\mathrm{forget}: \int_{\mathbf{\Sigma}_M}\!\RC\to \QQ_M$ 
from the proof of Lemma \ref{lem:bbArestriction} is final.
\end{lem}
\begin{proof}
Given any object $(\tilde{U},\tilde{S})\in \QQ_M$, we have to show that
the comma category
\begin{flalign}
(\tilde{U},\tilde{S})/\mathrm{forget}\,=\,\begin{cases}
\text{$\mathsf{Obj}$:} & ~(\Sigma,(U,S))\in \int_{\mathbf{\Sigma}_M}\!\RC\text{ such that }
(\tilde{U},\tilde{S})\to (U,S) \text{ in } \QQ_M\\[4pt]
\text{$\mathsf{Mor}$:} & ~(\Sigma,(U,S))\to (\Sigma^\prime,(U^\prime,S^\prime)) \text{ in } \int_{\mathbf{\Sigma}_M}\!\RC
\end{cases}
\end{flalign}
is non-empty and connected. (Note that all the categories involved are thin, i.e.\ there exists
at most one morphism between any two objects. Hence,
the morphisms in the comma category satisfy automatically the required commutative triangles.)
To verify non-emptyness, we use that $\tilde{U}\subseteq M$ is relatively compact,
hence there exists a Cauchy surface $\tilde{\Sigma}\subset M$ with 
$\cl_M(\tilde{U})\subseteq I^-_{M}(\tilde{\Sigma})$. 
(See e.g.\ \cite[Proposition A.5.13]{BGP} for a proof of this claim.) 
This defines an object 
$(\tilde{\Sigma},(\tilde{U},\tilde{S}))\in (\tilde{U},\tilde{S})/\mathrm{forget}$.
To verify connectedness, consider any two objects $(\Sigma,(U,S)),(\Sigma^\prime,(U^\prime,S^\prime))
\in (\tilde{U},\tilde{S})/\mathrm{forget}$. By Lemma~\ref{lem:surfaces}, there
exists a later Cauchy surface $\Sigma^{\prime\prime}\subset M$ such that
$\Sigma\subset I^-_M(\Sigma^{\prime\prime})$ and $\Sigma^\prime\subset I^-_M(\Sigma^{\prime\prime})$.
Then
\begin{flalign}
\xymatrix{
(\Sigma,(U,S)) \ar[r] & (\Sigma^{\prime\prime},(U,S))\ar[r] & 
(\Sigma^{\prime\prime},(U^{\prime\prime},S^{\prime\prime})) & 
\ar[l] (\Sigma^{\prime\prime},(U^\prime,S^\prime)) & \ar[l] (\Sigma^{\prime},(U^\prime,S^\prime))
}
\end{flalign}
defines a sequence of morphisms in $(\tilde{U},\tilde{S})/\mathrm{forget}$
which connects the two objects, where the object 
$(U^{\prime\prime},S^{\prime\prime})\in \RC_{(M,\Sigma^{\prime\prime})}$ in the middle
exists by filteredness of $\RC_{(M,\Sigma^{\prime\prime})}$, see Proposition \ref{prop:RCMSigmafiltered}.
\end{proof}


%
%


\end{document}